\documentclass[sigplan]{acmart}

\acmConference[]{ArXiv}{April, 2020}{ArXiv}

\acmBooktitle{ArXiv}

\acmYear{2020}
\acmISBN{} % \acmISBN{978-x-xxxx-xxxx-x/YY/MM}
\acmDOI{} % \acmDOI{10.1145/nnnnnnn.nnnnnnn}
\startPage{1}
\setcopyright{none}
\usepackage{booktabs}   %% For formal tables:
\usepackage{subcaption} %% For complex figures with subfigures/subcaptions
\usepackage{amsmath, amssymb, amsfonts}
\usepackage{algorithm}

\usepackage{ stmaryrd }
\usepackage[noend]{algpseudocode}

\usepackage{tablefootnote}
\usepackage{pgf}
\usepackage{ulem}
\usepackage{colored-comments}

\usepackage{ragged2e}
\usepackage{paralist}

\usepackage{tikz}
\usetikzlibrary{arrows,shapes,snakes,automata,backgrounds,positioning}

\tikzstyle{player}=[state,draw,rounded rectangle,align=center]
\tikzstyle{widget}=[draw=red,rectangle, rounded rectangle=10pt,dashed,minimum size=6mm,fill=yellow]
\tikzset{every loop/.style={looseness=7}}

\tikzstyle{player1}=[state,draw,rounded rectangle,align=center]
\tikzstyle{player2}=[state,draw,rectangle,align=center]
\tikzstyle{subgraph}=[rectangle,draw,dashed,minimum width=180,minimum height=100,font=\sffamily\Large\bfseries]
\tikzstyle{subgraph2}=[rectangle,draw,dashed,minimum width=95,minimum height=90,font=\sffamily\Large\bfseries]

\definecolor{armygreen}{rgb}{0.29, 0.33, 0.13}
\definecolor{ao(english)}{rgb}{0.0, 0.5, 0.0}
\definecolor{myblue}{rgb}{0.3, 0.4, 0.6}
\definecolor{pakistangreen}{rgb}{0.2, 0.8, 0.2}
\definecolor{hardness}{rgb}{0.16, 0.5, 0.0}
\definecolor{pseudo}{rgb}{1.0, 0.5, 0.0} 
\definecolor{bscc}{rgb}{0.24, 0.7, 0.44}
\definecolor{mikadoyellow}{rgb}{1.0, 0.77, 0.05}
\definecolor{cblue}{rgb}{0.19, 0.73, 0.56}

\definecolor{bluemdp}{rgb}{0.25, 0.29, 0.3}
\definecolor{redmdp}{rgb}{0.6, 0.0, 0.0}

\usepackage{xspace}
\usepackage{todonotes}
\usepackage{wrapfig}

\newtheorem{remark}{Remark}[section]
\newtheorem{assumption}{Assumption}[section]

\makeatletter
\DeclareRobustCommand*\cal{\@fontswitch\relax\mathcal}
\makeatother

\DeclareMathOperator{\even}{even}
\DeclareMathOperator{\odd}{odd}
\DeclareMathOperator{\infny}{inf}
\DeclareMathOperator{\proj}{proj}

\newcommand{\sgec}{\ensuremath{{\cal EC}_{\rm I}}}
\newcommand{\ugec}{\ensuremath{{\cal EC}_{\rm II}}}
\newcommand{\nec}{\ensuremath{{\cal EC}_{\rm III}}}
\newcommand{\typeone}{\ensuremath{{\sf Type\:{\rm I}}\ }}
\newcommand{\typetwo}{\ensuremath{{\sf Type\:{\rm II}}\ }}
\newcommand{\typethree}{\ensuremath{{\sf Type\:{\rm III}}\ }}

\DeclareMathOperator{\parobj}{p}

\newcommand{\QPL}{\ensuremath{\sf{QPL}\ }}
\newcommand{\QPLnospace}{\ensuremath{\sf{QPL}}}
\newcommand{\univ}{\ensuremath{\mathtt{A}}}
\newcommand{\qual}{\ensuremath{\mathtt{AS}}}
\newcommand{\posi}{\ensuremath{\mathtt{NZ}}}
\newcommand{\exis}{\ensuremath{\mathtt{E}}}
\newcommand{\event}{\ensuremath{{\sf \lozenge}}}

\newcommand{\parit}{\ensuremath{\mathtt{parity}}}

\newcommand{\induniv}{\ensuremath{a}}
\newcommand{\indqual}{\ensuremath{as}}
\newcommand{\indposi}{\ensuremath{nz}}
\newcommand{\indexis}{\ensuremath{e}}
\newcommand{\setuniv}{\ensuremath{\mathcal{A}}}
\newcommand{\setqual}{\ensuremath{\mathcal{AS}}}
\newcommand{\setposi}{\ensuremath{\mathcal{NZ}}}
\newcommand{\setexis}{\ensuremath{\mathcal{E}}}

\newcommand{\combi}[2]{\ensuremath{\bigwedge(#1,#2)}}

\newcommand*{\satisfy}[3]{#1,#2\models_{#3}}
\newcommand*{\satisfies}[2]{#1\models_{#2}}

\newcommand*{\nmodels}{\ensuremath{\nvDash}}
\newcommand*{\pipe}{\,|\,}
\newcommand*{\nd}{:\ }
\newcommand*{\nin}{\ensuremath{\not\in}}
\newcommand*{\eqdef}{\ensuremath{\overset{\text{\tiny def}}{=}}}
\newcommand{\abs}[1]{\ensuremath{\left|#1\right|}}

\newcommand*{\fun}[3]{#1 \colon #2 \to #3 }
\newcommand*{\limcup}[2]{\underset{#1}{\overset{#2}{\bigcup}}}

\newcommand*{\limwedgeone}[1]{\underset{#1}{\bigwedge}}

\newcommand{\initState}{\ensuremath{s_{{\sf init}}} }

\renewcommand{\proj}[1]{\ensuremath{{\sf proj}_{#1}} }

\newcommand{\commentsymbol}{//}% or \% or $\triangleright$
\algrenewcommand\algorithmiccomment[1]{\hfill \commentsymbol{} \textit{#1}}

\newcommand{\markovChain}{\ensuremath{\mathcal{M}} }

\newcommand{\MDP}{\ensuremath{\Gamma}}

\newcommand{\MDPtoMC}[2]{\ensuremath{{#1}^{[{#2}]}}}
\newcommand{\paths}{\ensuremath{\mathsf{Paths}}}
\newcommand{\fpaths}{\ensuremath{\mathsf{Fpaths}}}

\newcommand{\nat}{\ensuremath{\mathbb{N}} }

\newcommand{\rat}{\ensuremath{\mathbb{Q}} }

\newcommand{\last}[1]{\ensuremath{\mathsf{Last}(#1)} }

\newcommand{\post}{\ensuremath{\mathsf{Post}} }

\newcommand{\fn}{\ensuremath{\mathsf{f}} }

\newcommand{\NPinter}{\ensuremath{\text{NP} \cap \text{coNP}}}

\newcommand{\dists}{\ensuremath{\mathcal{D}} }

\newcommand{\prob}{\ensuremath{\mathsf{Pr}} }

\newcommand{\yes}{\ensuremath{\textsc{Yes}} }

\newcommand{\supp}{\ensuremath{\mathsf{Supp}} }

\newcommand{\zug}[1]{( #1  )}
\newcommand{\stam}[1]{}

\makeatletter
\tikzset{circle split part fill/.style  args={#1,#2}{%
 alias=tmp@name, % Jake's idea !!
  postaction={%
    insert path={
     \pgfextra{% 
     \pgfpointdiff{\pgfpointanchor{\pgf@node@name}{center}}%
                  {\pgfpointanchor{\pgf@node@name}{east}}%            
     \pgfmathsetmacro\insiderad{\pgf@x}
      \fill[#1] (\pgf@node@name.base) ([xshift=-\pgflinewidth]\pgf@node@name.east) arc
                          (0:180:\insiderad-\pgflinewidth)--cycle;
      \fill[#2] (\pgf@node@name.base) ([xshift=\pgflinewidth]\pgf@node@name.west)  arc
                           (180:360:\insiderad-\pgflinewidth)--cycle;            %  \end{scope}   
         }}}}}  
 \makeatother

\colorlet{drouge}{red}
\colorlet{frouge}{red!20!white}
\colorlet{dbleu}{blue}
\colorlet{fbleu}{blue!40!white}
\colorlet{dviolet}{blue!50!red}
\colorlet{fviolet}{blue!50!red!40!white}
\colorlet{dvert}{green!80!black}
\colorlet{fvert}{green!80!black!20!white}
\colorlet{djaune}{yellow!80!black}
\colorlet{fjaune}{yellow!80!black!20!white}
\colorlet{dgris}{white!60!black}
\colorlet{fgris}{white!90!black}
\colorlet{dgrisf}{white!30!black}
\colorlet{fgrisf}{white!70!black}
\colorlet{dorange}{red!50!yellow}
\colorlet{forange}{red!50!yellow!30!white}

\tikzstyle{ptrond}=[draw,circle,minimum height=2mm]
\tikzstyle{ptcarre}=[draw,minimum width=3mm,minimum height=3mm]
\tikzstyle{moyrond}=[draw,circle,minimum height=5mm]
\tikzstyle{moycarre}=[draw,minimum width=4mm,minimum height=4mm]
\tikzstyle{rond}=[draw,circle,minimum height=7mm]
\tikzstyle{carre}=[draw,minimum width=6mm,minimum height=6mm]
\tikzstyle{rouge}=[draw=drouge,fill=frouge]
\tikzstyle{vert}=[draw=dvert,fill=fvert]
\tikzstyle{jaune}=[draw=djaune,fill=fjaune]
\tikzstyle{bleu}=[draw=dbleu,fill=fbleu]
\tikzstyle{violet}=[draw=dviolet,fill=fviolet]
\tikzstyle{orange}=[draw=dorange,fill=forange]
\tikzstyle{gris}=[draw=dgris,fill=fgris]
\tikzstyle{grisf}=[draw=dgrisf,fill=fgrisf]
\tikzstyle{rvert}=[style=rond,style=vert]
\tikzstyle{rrouge}=[style=rond,style=rouge]
\tikzstyle{roundrect}=[draw,rounded rectangle, minimum width=6mm,minimum height=6mm]
\tikzstyle{splitrond}=[draw,circle split,minimum height=7mm,circle split part fill={blue!50,red!50}]
\tikzstyle{splitrondbv}=[draw,circle split,minimum height=7mm,circle split part fill={fbleu,fvert}]
\tikzstyle{splitrondrg}=[draw,circle split,minimum height=7mm,circle split part fill={frouge,fgris}]
\tikzstyle{splitrondbo}=[draw,circle split,minimum height=7mm,circle split part fill={fbleu,forange}]

\def\listof#1{\expandafter\@listof#1+\@end}
\def\@listof#1+#2\@end{\def\@tempa{#1}\ifx\@tempa\@empty\else 
    \langle #1\rangle \fi
  \def\@tempa{#2}\ifx\@tempa\@empty\else,\@listof#2\@end\fi}
\def\stackof#1{\begin{array}{@{}>{\scriptstyle}c@{}}\expandafter\@stackof#1+\@end}
\def\@stackof#1+#2\@end{\def\@tempa{#1}\ifx\@tempa\@empty\else 
    \langle #1\rangle \fi
  \def\@tempa{#2}\ifx\@tempa\@empty\end{array}\else\\[-1mm]\@stackof#2\@end\fi}

\begin{document}

%% Title information
\title[Mixing Objectives in MDPs]{Mixing Probabilistic and non-Probabilistic Objectives in Markov Decision Processes}         %% [Short Title] is optional;
                                        %% when present, will be used in
                                        %% header instead of Full Title.
% \titlenote{with title note}             %% \titlenote is optional;
                                        %% can be repeated if necessary;
                                        %% contents suppressed with 'anonymous'
% \subtitle{Subtitle}                     %% \subtitle is optional
% \subtitlenote{with subtitle note}       %% \subtitlenote is optional;
                                        %% can be repeated if necessary;
                                        %% contents suppressed with 'anonymous'

%% Author information
%% Contents and number of authors suppressed with 'anonymous'.
%% Each author should be introduced by \author, followed by
%% \authornote (optional), \orcid (optional), \affiliation, and
%% \email.
%% An author may have multiple affiliations and/or emails; repeat the
%% appropriate command.
%% Many elements are not rendered, but should be provided for metadata
%% extraction tools.

%% Author with single affiliation.
\author{Rapha\"{e}l Berthon}
% \authornote{with author1 note}          %% \authornote is optional;
                                        %% can be repeated if necessary
\orcid{nnnn-nnnn-nnnn-nnnn}             %% \orcid is optional
\affiliation{
%   \position{Position1}
%   \department{Department1}              %% \department is recommended
  \institution{Universit\'{e} libre de Bruxelles \& University of Antwerp}            %% \institution is required
%   \streetaddress{Street1 Address1}
%   \city{City1}
%   \state{State1}
%   \postcode{Post-Code1}
%   \country{Country1}                    %% \country is recommended
}
% \email{first1.last1@inst1.edu}          %% \email is recommended

%% Author with two affiliations and emails.
\author{Shibashis Guha}
% \authornote{with author2 note}          %% \authornote is optional;
                                        %% can be repeated if necessary
% \orcid{nnnn-nnnn-nnnn-nnnn}             %% \orcid is optional
\affiliation{
%   \position{Position2a}
%   \department{Department2a}             %% \department is recommended
  \institution{Universit\'{e} libre de Bruxelles}           %% \institution is required
%   \streetaddress{Street2a Address2a}
%   \city{City2a}
%   \state{State2a}
%   \postcode{Post-Code2a}
%   \country{Country2a}                   %% \country is recommended
}
% \email{first2.last2@inst2a.com}         %% \email is recommended
% \affiliation{
%   \position{Position2b}
%   \department{Department2b}             %% \department is recommended
%   \institution{Institution2b}           %% \institution is required
%   \streetaddress{Street3b Address2b}
%   \city{City2b}
%   \state{State2b}
%   \postcode{Post-Code2b}
%   \country{Country2b}                   %% \country is recommended
% }
% \email{first2.last2@inst2b.org}         %% \email is recommended

\author{Jean-Fran\c{c}ois Raskin}
% \authornote{with author2 note}          %% \authornote is optional;

\affiliation{
  \institution{Universit\'{e} libre de Bruxelles}           %% \institution is required
}
                                        %% can be repeated if necessary
% \orcid{nnnn-nnnn-nnnn-nnnn}             %% \orcid is optional

%% Abstract
%% Note: \begin{abstract}...\end{abstract} environment must come
%% before \maketitle command
\begin{abstract}
%Text of abstract \ldots.
In this paper, we consider algorithms to decide the existence of strategies in MDPs for Boolean combinations of objectives.
These objectives are omega-regular properties that need to be enforced either {\em surely}, {\em almost surely}, {\em existentially}, or with {\em non-zero probability}. 
In this setting, relevant strategies are {\em randomized infinite memory strategies}: both infinite memory {\em and} randomization may be needed to play optimally. 
We provide algorithms to solve the general case of Boolean combinations and we also investigate relevant subcases. We further report on complexity bounds for these problems.
\end{abstract}

%% 2012 ACM Computing Classification System (CSS) concepts
%% Generate at 'http://dl.acm.org/ccs/ccs.cfm'.

%% End of generated code

%% Keywords
%% comma separated list
  %% \keywords are mandatory in final camera-ready submission

%% \maketitle
%% Note: \maketitle command must come after title commands, author
%% commands, abstract environment, Computing Classification System
%% environment and commands, and keywords command.
\maketitle

%\section{Introduction}
	\section{Introduction}

Recently, there have been several works on how to mix the semantics of games and Markov decision processes~\cite{DBLP:conf/stacs/BruyereFRR14,AKV16,DBLP:journals/iandc/BruyereFRR17,BRR17,DBLP:conf/concur/ChatterjeeP19}. This setting provides means to model the interaction between a system and its environment that is uncontrollable but obeys stochastic dynamics. The setting is then used to reason on strategies of the system that ensure for example some properties with {\em certainty} and others with {\em high probability}.

Here, we extend this line of work by studying a general setting where objectives for the system are Boolean combinations of atoms. These atoms are omega-regular properties, expressed as parity conditions, that need to be ensured either {\em surely} ({\tt A}), {\em almost surely} ({\tt AS}), {\em existentially} ({\tt E}), or with {\em non-zero probability} ({\tt NZ}). 
Sure ({\tt A}) and existential ({\tt E}) atoms are {\em non-probabilistic} while almost sure ({\tt AS}) and non-zero ({\tt NZ}) atoms are {\em probabilistic}. The coexistence of atoms of both types that need to be satisfied by a unique strategy makes this problem out of reach of classical techniques used to solve MDPs with {\sf CTL} objectives or with {\sf PCTL} objectives for example. 

\paragraph{Infinite memory and randomization}
% \rbchanged{rewrite this with~\ref{etessami2007multi} in mind. Remove example 1, and mention Figures 3 and 5.} 
In some previous works on models that mix games and MDPs~\cite{DBLP:conf/stacs/BruyereFRR14,AKV16,DBLP:journals/iandc/BruyereFRR17,BRR17},
% \sgcomment{Do all these papers talk about infinite memory strategies?}, infinite memory strategies are necessary to play optimally; randomization is not necessary.
In~\cite{etessami2007multi}, combination of parity conditions are studied. In that paper, randomization is necessary, but not infinite memory.
In the setting that is considered in the current paper, relevant strategies for the systems are {\em randomized infinite memory strategies}: both infinite memory {\em and} randomization may be needed to play optimally.
This implies that the techniques used here are more complex than for the previous works.
Note that randomization is already necessary when considering conjunctions of two {\tt NZ} atoms. The example we give above in Figure~\ref{fig:GEC} is encompassed by the formalism of~\cite{etessami2007multi}, and shows why we need to add randomization when compared to the work in~\cite{BRR17}. In the MDP of Figure~\ref{fig:GEC}, there does not exist a deterministic choice from state $s_0$ between action $a$ and action $b$ that ensures ${\tt NZ}(p_1) \land {\tt NZ}(p_2)$, while a randomized strategy can enforce this objective by taking $a$ with probability $\alpha$ ($0 < \alpha < 1$) and taking $b$ with probability $1-\alpha$.

\begin{figure}[t]
\centering
% \vspace{-15pt}
\hspace{-15pt}
\centering
\scalebox{1}{
	\begin{tikzpicture}
		\node[player] (s1) at (0,0) {0,1} ;
		\node[player] (s0) at (2,0) {0,0} ;
		\node[player] (s2) at (4,0) {1,0} ;

		\node (s4) at (0,-.6) {\large $s_1$} ;
		\node (s5) at (2,-.6) {\large $s_0$} ;
		\node (s6) at (4,-.6) {\large $s_2$} ;
		
 \path[-latex]	(s0) edge              node[above] {\small $a,1$} (s1)				
				(s0) edge              node[above] {\small $b,1$} (s2)
				(s1) edge[loop left]   node {$a,1$} (s1)
				(s2) edge[loop right]  node {$b,1$} (s2)
		;
	\end{tikzpicture}
}
\caption{\label{fig:GEC}The MDP depicted here has 3 states and two parity functions $p_1$ and $p_2$. The numbers assigned by the parity functions to the states are depicted by the integers inside the states. 
A parity condition is enforced if the maximum value of states that appear infinitely often is even.
% E.g. in state $s_1$, the value of the parity function $p_1$ is $0$ while the value of the parity function $p_2$ is $1$. 
% In state $s_0$, $a$ and $b$ are two possible actions to play. 
% The two possible actions from state $s_0$ are $a$ and $b$.
By taking $a$ from $s_0$, state $s_1$ is reached with probability one, and by taking $b$, $s_2$ is reached with probability one.  Clearly on this example, a randomized strategy is needed to win for $\posi(\parobj_1)\wedge\posi(\parobj_2)$ from state $s_0$.}
\end{figure}
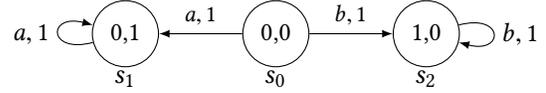
% \vspace{-.5cm}
% \rbchanged{Should we give details that $\univ + \qual$ is a simple extension of previous works, but adding NZ requires more  work?}
\paragraph{Main contributions} Our main contributions are summarized in Table~\ref{tab:main}. We provide a $\Sigma_2^{\sf P}$ algorithm to decide the existence of a strategy to enforce a Boolean combination of atomic objectives. We also show that this problem is both {\sf NP} and {\sf coNP} hard. Then we provide additional results for relevant subclasses of Boolean combinations. For the conjunctive case, we prove the existence of a polynomial algorithm that uses an {\sf NP} oracle while the problem is shown to be {\sf coNP} hard. For conjunctions that contain only one sure atom (1A) and a number of other atoms, the complexity goes down to $\sf{NP}\cap\sf{coNP}$ and it is at least as hard as solving parity games. The complexity of this algorithm is dominated by the complexity of solving parity games.  A polynomial time solution to parity games would lead to a polynomial time solution for our problem. Also the recent quasi-polynomial time solutions for parity games, see e.g.~\cite{DBLP:conf/stoc/CaludeJKL017}, can be used to obtain a quasi-polynomial time solution to our problem.  
% (as this is a immediate corollary of our result, we do not come back to that later.)
Finally, for conjunctions that do not contain sure atoms, the problem can be solved in polynomial time.  
% Our algorithms are based on non-trivial extensions of the notion of end-components, a classical tool at the basis of most of the algorithms for the analysis of long-run properties of finite Markov decision processes.
\vspace{-.5cm}
\begin{center}
\begin{table}[]
\begin{tabular}{c|c|c|}
\cline{2-3}
                                                            & Hardness               & Membership                              \\ \hline
\multicolumn{1}{|c|}{$\wedge(\qual,\posi,\exis)$}           & ${\sf P}$               & 
\begin{tabular}{@{}c@{}}${\sf P}$ \\ (Thm. \ref{thm:poly})\end{tabular}
\\ \hline
\multicolumn{1}{|c|}{$\wedge(1\univ,\qual,\posi,\exis)$}    & ${\sf parity}$          & 
\begin{tabular}{@{}c@{}}${\sf NP}\cap{\sf coNP}$ \\ (Thm. \ref{thm:NPcapCoNP})\end{tabular}
\\ \hline
\multicolumn{1}{|c|}{$\wedge(\univ,\qual,\posi,\exis)$}     & ${\sf coNP}$            & \begin{tabular}{@{}c@{}} ${\sf P}^{\sf NP}(=\Delta_2^{\sf P})$ \\ (Thm. \ref{thm:multi-parity})\end{tabular}  \\ \hline
\multicolumn{1}{|c|}{$\mathbb{B}(\univ,\qual,\posi,\exis)$} & 
\begin{tabular}{@{}c@{}} ${\sf NP}$ and ${\sf coNP}$ \\ (Thm. \ref{thm:hardness}) \end{tabular} & \begin{tabular}{@{}c@{}}${\sf NP}^{\sf NP}(=\Sigma_2^{\sf P})$ \\ (Thm. \ref{thm:Sigma2P}) \end{tabular} \\ \hline
\end{tabular}
\caption{\label{tab:main}Table of the main complexity results.}
\end{table}
\end{center}

% \begin{center}
% \begin{table}[]
% \begin{tabular}{c|c|c|c|c|}
% \cline{2-5}
%                                  & $\wedge(\qual,\posi,\exis)$ & $\wedge(1\univ,\qual,\posi,\exis)$ & $\wedge(\univ,\qual,\posi,\exis)$ & $\mathbb{B}(\univ,\qual,\posi,\exis)$ \\ \hline
% \multicolumn{1}{|c|}{Membership} & $\sf{P}$                          & $\sf{NP}\cap\sf{coNP}$             & $\sf{P}^{\sf{NP}}(=\Delta_2^{\sf{P}})$                & $\sf{NP}^{\sf{NP}}(=\Sigma_2^{\sf{P}})$                   \\ \hline
% \multicolumn{1}{|c|}{Hardness}   & $\sf{P}$                          & $\sf{parity}$                      & $\sf{coNP}$                       & $\sf{NP}\cup\sf{coNP}$                \\ \hline
% \end{tabular}
% \caption{\label{Tab:main}Table of the main complexity results.}
% \end{table}
% \end{center}
% \vspace{-1.8cm}
\paragraph{Related works}
Logical formalisms to express properties of transition systems and Markov decision processes were plentifully studied in the literature. But most of the results in the literature only consider either logics based on {\em non-probabilistic} atoms, e.g. {\sf CTL}, or logics based on {\em probabilistic} atoms only, e.g. {\sf PCTL}. The logic {\sf PCTL} is used to express constraints on the probability of events that are temporal properties of paths.  In~\cite{brazdil2008controller}, the strategy synthesis problem for MDPs with {\sf PCTL} objectives is studied. The full logic, i.e. with arbitrary probabilistic thresholds, is undecidable but the qualitative fragment of the logic (thresholds 0 and 1, corresponding to {\tt NZ} and {\tt AS} in our setting) is decidable in ${\sf EXPTIME}$. 
% We conjecture that 
This high complexity is due to the succinctness of PCTL.
As PCTL cannot express our non-probabilistic atoms, the two formalisms have incomparable expressive power.
Settings that mix both non-probabilistic properties, such as {\tt A} or {\tt E}, with probabilistic ones such as {\tt AS} or {\tt NZ}, are more recent. We make now a more detailed review of the recent relevant works in that direction.

%In \cite{brazdil2008controller}: Games with $PCTL$ objectives ($PCTL$ is the same as $PCTL*$?). $\sf{EXPTIME}$-complete for qualitative probabilities ($1$ and $>0$). Undecidable for arbitrary probabilities.

In \cite{DBLP:conf/stacs/BruyereFRR14,DBLP:journals/iandc/BruyereFRR17}, MDPs with mean-payoff and shortest path objectives are considered. This work was, to the best of our knowledge, the first work to consider the synthesis of strategies that optimize an {\em expectation} (a probabilistic property) while satisfying a long-run {\em worst-case objective} (a non-probability objective). Similarly, the authors of~\cite{AKV16} consider the synthesis of strategies that ensure a parity condition {\em surely} and at the same time an $\epsilon$-optimal {\em expected} mean-payoff. Those works introduce refinements of the notion of {\em end-components} that we need to further refine here.

The authors of~\cite{michalewski2016measure} study an extension of MSO, called MSO+$\nabla$ which uses a probabilistic second order quantifier. The logic MSO+$\nabla$ is expressive enough to encode the problem we study here, but this logic has been proved to be  undecidable~\cite{bojanczyk2019mso,BFFGMMPRS20}. In~\cite{bojanczyk2016thin,bojanczyk2017emptiness}, a fragment of MSO+$\nabla$, called thin MSO has been introduced. The logic Thin MSO is expressive enough to encode the model-checking problem of the qualitative fragment ${\sf CTL}^*+{\sf PCTL}^*$ (union of ${\sf CTL}^*$ and ${\sf PCTL}^*$) over Markov chains. Their algorithm has non-elementary complexity. The algorithm was recently improved in~\cite{fournier2018alternating} where a model-checking algorithm with $\sf{3NEXPTIME}\cap \sf{co}-3\sf{NEXPTIME}$ complexity is proposed. The works in~\cite{bojanczyk2016thin,fournier2018alternating} do not consider the richer model of Markov decision processes as we do here. 

In~\cite{carayol2014randomization}, the authors study qualitative tree automata, that is automata with a probabilistic acceptance condition. The non-emptiness problem of nondeterministic tree automaton with such acceptance condition has been proved decidable, but the problem has been proved undecidable for nondeterministic tree automaton with such acceptance condition~\cite{BFFGMMPRS20}. There is a deep connection between tree automata and Markov Decision Processes, as the existence of a strategy on an MDP corresponds to deciding the non-emptiness of a qualitative tree automaton with unary alphabet.

In~\cite{michalewski2016regular,bojanczyk2017emptiness}, the authors study {\em subzero} automata: a class of tree automata with an acceptance condition that mixes the classical Rabin acceptance condition with probabilistic constraints. The problem of determining if a subzero automaton accepts some regular tree is decidable. This class of automaton can in turn be used to solve synthesis problem for {\em finite-memory} strategies (that are equivalent to regular trees) that enforce a first parity condition $p_1$ surely ({\tt A}) and a second parity condition $p_2$ almost-surely ({\tt AS}). Our work consider more general properties (both {\tt E} and {\tt NZ} in addition to {\tt A} and {\tt AS}, and their Boolean combinations) and more general strategies: randomized infinite memory strategies, and not only finite memory deterministic strategies (regular trees).

In this paper, we provide non-trivial extensions of results in~\cite{BRR17} where only the case of one sure parity objective ({\tt 1A}) and one almost-sure parity objective ({\tt 1AS}) is considered. An $\sf{NP}\cap\sf{coNP}$ algorithm is provided there for this special case. In this paper, in addition to a $\Sigma_2^{\sf P}$ algorithm for the general case $\mathbb{B}(\univ,\qual,\posi,\exis)$, we also provide an algorithm that solves conjunctions of one sure parity objective ({\tt 1A}) and any number of almost-sure ({\tt AS}), existential ({\tt E)}, and non-zero probability ({\tt NP}) parity objectives with the same worst-case complexity as that of~\cite{BRR17}. Algorithms in~\cite{BRR17} heavily rely on notions of very good end-components (VGEC) and ultra good end-components (UGEC). Here, we need generalization of VGEC and UGEC, and additional technical results to build algorithms for our more general setting.

Finally, the authors of~\cite{DBLP:conf/concur/ChatterjeeP19} consider the synthesis of {\em finite-memory strategies} for MDPs with a sure parity ({\tt S}) and an almost-sure parity ({\tt AS}) objectives. The restriction to finite memory strategies leads to simpler algorithms but the complexity is similar, i.e. $\sf{NP}\cap\sf{coNP}$. The authors of~\cite{DBLP:conf/concur/ChatterjeeP19} also consider the case of $2\frac{1}{2}$-player games. In that setting the problem is $\sf{coNP}$-complete.
\vspace{-.3cm}
\paragraph{Structure of the paper} In Section~2, we introduce necessary preliminaries about MDPs, and we formally define the class of properties that we consider, i.e. Boolean combinations of {\tt A}, {\tt AS}, {\tt E}, and {\tt NZ} atoms.
In Sections~3, 4 and 5, we study notions of end-components that are the main technical ingredients of our algorithms. Section~6 introduces additional techniques needed to handle {\tt E} and {\tt NZ} atoms. In Section~7, we study the complexity of algorithms for the general case, and several relevant fragments.

% Due to lack of space, full proofs are provided in an appendix which is submitted as a separate PDF document as required by the submission guidelines.
% Due to lack of space, full proofs are provided in the appendix.
Due to lack of space, full proofs are provided in
% ~\todo[inline]{full version}.
	
	\section{Preliminaries}
\label{sec:prelims}
% \sgcomment{We should do a global check when we say "condition", and when we say "objective".}
For $k \in \mathbb{N}$, we denote by $[k]_0$
%\sgcomment{Possibly $[k]_0$ is not used in the main paper and can be omitted} 
and $[k]$ the set of natural numbers $\lbrace 0, \ldots, k \rbrace$ and $\lbrace 1, \ldots, k \rbrace$ respectively. Given a finite set $A$, a (rational) \textit{probability distribution} over $A$ is a function $\fun{\prob}{A}{[0, 1] \cap \rat}$ such that $\sum_{a\in A} \prob(a) = 1$. 
We denote the set of probability distributions on $A$ by $\dists(A)$.
The \textit{support} of the
probability distribution $\prob$ on $A$ is
$\supp(\prob) = \left\lbrace a \in A \;\vert\; \prob(a) >
  0\right\rbrace$.
% The \textit{support} of a probability distribution $\dist \in \dists$ on $A$ is $\supp(\dist) = \left\lbrace a \in A \;\vert\; \dist(a) > 0\right\rbrace$.
% A distribution is called \emph{Dirac} if $|\supp(\dist)| = 1$.\rbchanged{We never reuse this notation}

%\subsection{Weighted Markov decision processes and Markov chains}
% \noindent \textbf{Markov chain.}
\paragraph{Markov chain}
%\label{prelim_m}
We denote by $\mathbb{N}$ the set $\{1,2,\ldots\}$, and by $\mathbb{N}_0$ the set $\mathbb{N}\cup \{0\}$. A \emph{Markov chain} (MC, for short) is a tuple $\mathcal{M} = \zug{S,E,\prob}$, where $S$ is a set of states, $E \subseteq S \times S$ is a set of edges (we assume in this paper that the set $E(s)$ of outgoing edges from $s$ is nonempty and finite for all $s \in S$), and $\prob: S \to \dists(E)$ assigns a probability distribution  -- on the set $E(s)$ of outgoing edges from $s$ -- to all states $s \in S$. In the following, $\prob(s,(s,s'))$ is denoted $\prob(s,s')$, for all $s \in S$. The Markov chain $\markovChain$ is \emph{finite} if $S$ is finite.
%A Markov chain is infinite if the set $S$ is not finite. Unless stated otherwise, we always consider a Markov chain to be finite.
%The size of $\mathcal{M}$ is the number of states $|S|$, and will be denoted $|\mathcal{M}|$.

For $s \in S$, the set of infinite paths in $\mathcal{M}$ starting from \\ $s$ is $\paths^{\mathcal{M}}(s) = \lbrace \pi = s_0 s_1 \ldots \in S^\omega \mid s_0 = s, \forall\, n \in \nat_0,\, \\ \prob(s_n,s_{n+1}) > 0 \rbrace$. The set of all infinite paths in $\mathcal{M}$ is \\$\paths^{\mathcal{M}} = \bigcup_{s \in S} \paths^{\mathcal{M}}(s)$.
For $\pi = s_0 s_1 \ldots \in \paths^{\mathcal{M}}$, we denote by $\pi(i,l)$ the sequence of $l-i+1$ states (or $l+1$ edges) $s_i \ldots s_{i+l}$, and for simplicity, we denote $\pi(i,0)$ by $\pi(i)$. 
The infinite suffix of $\pi$ starting in $s_n$ is denoted by $\pi(n,\infty) \in \paths^{\mathcal{M}}$.
The set of finite paths starting from a state $s \in S$ is defined as $\fpaths^{\mathcal{M}}(s) = \lbrace \pi = s \ldots s' \in S^+ \mid \exists \bar{\pi} \in \paths^{\mathcal{M}},\; \pi\bar{\pi} \in \paths^{\mathcal{M}}(s) \rbrace$ and $\fpaths^{\mathcal{M}} = \bigcup_{s \in S} \fpaths^{\mathcal{M}}(s)$.
For $\pi=s \ldots s'$, we denote by $\last{\pi}$, the last state $s'$ in $\pi$.
As in~\cite{vardi1985automatic}, we extend the probability distribution to the space of infinite paths by considering cylinders defined by finite prefixes and using Carath\'{e}odory's extension theorem. 
% We denote the function associated by $\prob_{\mathcal{M}}$.
We denote this probability distribution over the set of infinite paths beginning from some initial state $s$ by $\prob^s_{\mathcal{M}}$. When $s$ is clear from the context, we omit it and only denote this distribution by $\prob_{\mathcal{M}}$.
\paragraph{Markov decision process}
A finite \emph{Markov decision process} (MDP, for short) is a tuple $\Gamma = ( S,E,Act,\prob)$, where $S$ is a finite set of states, $Act$ is a finite set of actions, and $E \subseteq S \times Act \times S$ is a set of edges, and $\prob\colon S \times Act \to 
% (\dists(E)\cup {{S}\to{\{0\}}})
\dists(E)$ is a partial function that assigns a probability distribution -- on the set $E(s,a)$ of outgoing edges from $s$ -- to all states $s \in S$ if action $a \in Act$ is taken from $s$. For all $s\in S$ there exists at least one $a\in Act$ such that $E(s,a)$ is defined. Given $s \in S$ and $a \in Act$, we define $\post(s,a) = \lbrace s' \in S \mid \prob(s,a,s') > 0 \rbrace$. Then, for all state $s \in S$, we denote by $Act(s)$ the set of actions  $\lbrace a \in Act \mid \post(s,a) \neq \varnothing \rbrace$.
We assume that, for all $s \in S$, we have $Act(s) \neq \varnothing$.
%The size of $\Gamma$ will be denoted $|\Gamma|$, and will refer to the number of states of $\Gamma$ times the number of actions, that is $|S| \cdot |Act|$. 
% For a finite path $\pi$ in an MDP, we also use $\last{\pi}$ to denote the last state in $\pi$.
Given an MDP $\MDP = ( S,E,Act,\prob)$, and a set of states $C\subseteq S$, we define the restriction of $\MDP$ to $C$, denoted $\MDP{\downharpoonright C}$, as the MDP $(C,E',Act,\prob')$ where $E'= \{(s,a,s') \pipe s,s'\in C, \ a\in Act, \post(s,a)\subseteq C, \text{ and } (s,a,s') \in E \}$, and 
$\prob'$ is a partial function defined as $\prob'(s,a)=\prob(s,a)$ if $(s,a,s') \in E'$ for $a \in Act$, and $s,s' \in C$, and is undefined otherwise.
% for $s\in C,\ a\in Act$ we take $\prob'(s,a)=\prob(s,a)$ if 
% $S\backslash C\not\in \post(s,a)$ 
% $\post(s,a) \subseteq C$,
% and $\prob'(s,a)=\varnothing$ otherwise.
% $\prob'(s,a)$ is not defined otherwise.

A \emph{strategy} in $\Gamma$ is a function $\fun{\sigma}{S^{+}}{\dists(Act)}$
%\footnote{More generally, a strategy can be defined as \emph{randomized} strategy which is a function $\sigma: S^{+} \to \dists(Act)$. We show that for solving all the problems considered here pure or deterministic strategies suffice.}
such that for all $s_0 \ldots s_n \in S^{+}$, we have $\supp(\sigma(s_0\ldots s_n))\subseteq Act(s_n)$.
% \Sigma_{a\in Act(s_n)}\prob(\sigma(s_0 \ldots s_n) =a)=1$, for all $s_0 \ldots s_n \in S^{+}$. 
% We denote by $\textsf{strat}(\Gamma)$ the set of strategies available in $\Gamma$\rbchanged{Do we use this notation?}.
% Consider a measurable function $\fn$ that associates a value to infinite paths in Markov chains.
% Then, we call $\sup_{\sigma \in \textsf{strat}(\Gamma)} \expect^{\MDPtoMC{\Gamma}{\sigma}}_{\initState}(\fn)$ the expected value of $\fn$ in $\Gamma$.
A strategy $\sigma$ can be encoded by a transition system $\mathtt{T}=\zug{Q,S, act,\delta,\iota}$ where $Q$ is a (possibly infinite) set of states, called modes, $\fun{act}{Q \times S}{\dists(Act)}$ selects a distribution on actions such that, for all $q \in Q$ and $s \in S$, 
%$act(q,s) \in Act(s)$, $\delta: Q \times S \to Q$ is a mode update function and $\iota: S \to Q$ selects an initial mode for each state $s \in S$.
we have, $act(q,s) \in \dists(Act(s))$.
The function $\fun{\delta}{Q \times S}{Q}$ is a mode update function and $\fun{\iota}{S}{Q}$ selects an initial mode for each state $s \in S$.
If the current state is $s \in S$, and the current mode is $q \in Q$, then the strategy chooses the distribution $act(q,s)$, and the next state $s'$ is chosen according to the distribution $act(q,s)$.
Formally, $\zug{Q,S, act,\delta,\iota}$ defines the strategy $\sigma$ such that $\sigma(\rho \cdot s) = act(\delta^*(\iota(\rho(0)), \rho),s)$ for all $\rho \in S^{*}$, and $s \in S$, where $\delta^*$ extends $\delta$ to sequence of states starting from $\iota$ as expected, i.e., $\delta^*(\iota(\rho(0)), \rho \cdot s) = \delta(\delta^*(\iota(\rho(0)),\rho),s)$, and $\delta^*(\iota(\rho(0)), \varepsilon) = \iota(\rho(0))$.
% The amount of memory used by such a strategy is defined to be $|Q|$.
We denote by $\mathtt{T}_\sigma$ a transition system with minimal number of modes that corresponds to a strategy $\sigma$.
A strategy is said to be \emph{memoryless} if there exists a transition system encoding the strategy with $|Q| = 1$, that is, the choice of action only depends on the current state.
A memoryless strategy can be seen as a function $\fun{\sigma}{S}{\dists(Act)}$.
% such that for all $s \in S$, we have $\sigma(s) \in \dists(Act(s))$.
% Formally, a strategy is memoryless if for all finite paths $\rho_1$ and $\rho_2$ in $\fpaths^{\MDPtoMC{\Gamma}{\sigma}}$ such that $\last{\rho_1} = \last{\rho_2}$, we have $\sigma(\rho_1) = \sigma(\rho_2)$.
Formally, a strategy $\sigma$ is memoryless if for all  finite sequences of states $\rho_1$ and $\rho_2$ in $S^{+}$ such that $\last{\rho_1} = \last{\rho_2}$, we have $\sigma(\rho_1) = \sigma(\rho_2)$.
A strategy is called a \emph{finite memory} strategy if there exists a transition system encoding the strategy in which $Q$ is finite.
A strategy is \emph{deterministic} if $\fun{\sigma}{S^{+}}{Act}$.
For deterministic strategies, we have $\fun{act}{Q \times S}{Act}$ such that for all $q \in Q$ and $s \in S$, we have $act(q,s) \in Act(s)$.
Note that the state space of $\MDPtoMC{\Gamma}{\sigma}$ is $Q \times S$.
For a sequence $\pi$ of states in $\MDPtoMC{\Gamma}{\sigma}$, we denote by $\proj{S}(\pi)$ the corresponding sequence of states in the MDP $\MDP$.
Once we fix a strategy $\sigma$ encoded by the transition system $\zug{Q,S, act,\delta,\iota}$
% , and an initial state $\initState\in S$ 
in an MDP $\Gamma = ( S,E,Act,\prob )$, we obtain an MC $\MDPtoMC{\Gamma}{\sigma}=(S',E',\prob')$, where $S'=Q\times S$ is the set of states, $E'=\{(q\times s)\times(q'\times s') \pipe q,q'\in Q,\ s,s'\in S, \delta(q,s)=q',\ \exists a\in Act,\ a\in\supp(act(q,s))\text{ and } (s,a,s')\in E\}$ is the set of edges, and for $q,q'\in Q,\ s,s'\in S$  we have the probability distribution $\prob'(q,s)(q',s')=\Sigma_{a\in \supp(act(q,s))}act(q,s)(a)\cdot \prob(s,a,s')$ if $q'=\delta(q,s)$ and is not defined otherwise. In the sequel, by abuse of notation, we write the projection onto the second component, that is $s$ instead of $(q,s)$, for a state of this MC, unless specifically stated.
% \rbchanged{two "otherwise" one aft
% \sgcomment{Define the MC.}
% Given a target set $T$, we define the attractor of $T$, denoted $\mathtt{Attr}_{1}(T)$ as the fixed point of the sequence 
% \(
% \mathtt{Attr}^{n+1}_{1}(T) = \mathtt{Attr}^{n}_{1}(T) \cup \{ s \in S  \mid  \exists a\in Act,\  \post(s,a)\subseteq  \mathtt{Attr}^{n}_{1}(T), \}
% \)
% and $\mathtt{Attr}_{Env}(T)$ as the fixed point of the sequence 
% \(
% \mathtt{Attr}^{n+1}_{Env}(T) = \mathtt{Attr}^{n}_{Env}(T) \cup \{ s \in S  \mid  \forall a\in Act,\  \post(s,a)\cap  \mathtt{Attr}^{n}_{Env}(T) \neq \varnothing\}
% \)
% with $\mathtt{Attr}^{0}_{1}(T)=\mathtt{Attr}^{0}_{Env}(T) = T$.
% \noindent \textbf{Two-player games.}
\paragraph{One and two-player games}
For a given objective, an MDP $\Gamma = ( S,E,Act,\prob)$ %(resp. weighted MDP) 
can also be considered to have the semantics of a zero-sum two-player turn-based game 
%(resp. weighted two-player turn-based game) 
where the game is played for infinitely many rounds and the exact probabilities are not important (this is the case when we will consider $\univ$ and $\exis$ atoms).
The first round starts from a designated initial state $\initState \in S$.
In each round, Player~$1$ chooses an action $a \in Act(s)$ from a state $s$ while Player~$2$ that is adversarial resolves the non-determinism by choosing a state $s'$ such that $\prob(s,a,s') > 0$.
We denote by $G_\Gamma=\zug{S,E,Act}$ the two-player game that is obtained from an MDP $\Gamma=\zug{S,E,Act,\prob}$.
When the players resolve the non-determinism co-operatively, we have a one-player game.
Equivalently, in a one-player game, Player~$1$ chooses both action $a$ as well as the state $s'$.
% Sometimes we refer to Player~$1$ and Player~$2$ as $Sys$ and $Env$ respectively.

Given a target set $T$, we define the attractor of $T$, denoted $\mathtt{Attr}_{1}(T)$ as the set of states from which there exists a strategy for Player~$1$ to reach $T$ with certainty.
This corresponds to reachability in a classical ``and-or" graph.
For a two-player game, given $T$, an algorithm to obtain its attractor computes a sequence of sets of states ($\mathtt{Attr}^{n}_{1}(T)$)$_{n\ge 0}$ defined as follows:
\begin{inparaenum}[(i)]
\item $\mathtt{Attr}^{0}_{1}(T)= T$; and
\item for all $n\geq 0$: $\mathtt{Attr}^{n+1}_{1}(T) = \mathtt{Attr}^{n}_{1}(T) \cup \{ s \in S  \mid  \exists a\in Act,\  \post(s,a)\subseteq  \mathtt{Attr}^{n}_{1}(T) \}$.
\end{inparaenum}
Clearly $\mathtt{Attr}^{n+1}_{1}(T) \supseteq \mathtt{Attr}^{n}_{1}(T)$.
If $S$ is finite, then there exists an $m \in \nat_0$ such that $\mathtt{Attr}^{n}_{1}(T) = \mathtt{Attr}^{m}_{1}(T)$ for all $n \ge m$.
The algorithm for the case of one-player game only changes in the induction step where we have for all $n\geq 0$: $\mathtt{Attr}^{n+1}_{1}(T) = \mathtt{Attr}^{n}_{1}(T) \cup \{ s \in S  \mid  \exists a\in Act,\  \post(s,a) \cap \mathtt{Attr}^{n}_{1}(T) \neq \varnothing \}$.
The algorithm for the one-player case corresponds to classical graph reachability.

\stam{
For a two-player game $G$, the strategies available for Player~$1$ are defined in the same way as in the case of MDPs.
A strategy of Player~$2$ is a function $\fun{\mu}{S^{+}\cdot Act}{S}$, with the restriction that if $\mu(s_0 s_1 \dots s_n \cdot a)=s$ then 
%$s \in {\sf Supp}(\delta(s_n,a))$. 
$\prob(s_n,a,s) > 0$.
%\marginpar{Should we add the notion of support ?}
The set of strategies for Player 1 and Player 2 are denoted by $\mathsf{strat}_1(G)$ and $\mathsf{strat}_2(G)$ respectively.
In a two-player game, there is no randomness. 
Therefore, once the deterministic strategies are fixed, there is a unique path that occurs in the game, from an initial state $\initState$.
For a game $G$, given two strategies $\sigma_1 \in \mathsf{strat}_1(G)$ and $\sigma_2 \in \mathsf{strat}_2(G)$, we denote by $\pi_{(G,s,\sigma_1,\sigma_2)}$ the path that occurs in the two-player game $G$ under strategies $\sigma_1$ and $\sigma_2$ from state $s$.
Then, considering a function $\fn$ that associates a value to each infinite path of a two-player game, we denote by $V_s^{\fn}(G)$ the value $\sup\limits_{\sigma_1 \in \mathsf{strat}_1(G)} \inf\limits_{\sigma_2 \in \mathsf{strat}_2(G)} \fn(\pi_{(G,s,\sigma_1,\sigma_2)})$.% that is the outcome of the two-player game $G$.
The definitions concerning the memories of strategies also apply to two-player games.
}
We denote the size of an MC $\mathcal{M}$, MDP $\MDP$ and two-player game $G$ by $|\mathcal{M}|$, $|\MDP|$ and $|G|$ respectively.
For each case, the size is the sum of the number of states, the number of edges, and the size of the representation of the transition matrix, that is, $|S|+|E|+\abs{\prob}$.

\paragraph*{Parity conditions and qualitative parity logic} Given an MDP $\MDP$, a parity condition is a function $\fun{\parobj}{S}{\mathbb{N}_0}$. 
Given a path $\pi\in S^{\omega}$, the set $\inf(\pi)=\{s\in S\ |\ \forall i\geq 0,\ \exists j\geq i,\ \textrm{ such that } \pi(j) =s\}$ is the set of states visited infinitely often on this path.
A path satisfies a parity condition $\parobj$ if $\max\{\parobj(s) \:|\: s\in \inf(\pi)\}$ is even.
%Given a MDP, we will study objectives that strategies can satisfy on this MDP. These objectives will be defined with the following grammar:
Given a parity condition $\parobj$, its dual is the condition $\bar{\parobj}\colon s \mapsto 1+\parobj(s)$.
We denote by $\parit$ the set of parity conditions. 
A path satisfies $\bar{\parobj}$ iff it does not satisfy $\parobj$.
We now define qualitative parity logic (\QPL) which is defined by the following grammar.
% \vspace{-.4cm}
\begin{align*}
atom &= \univ(\parobj) \pipe \exis(\parobj) \pipe \qual(\parobj) \pipe \posi(\parobj) \ (\parobj\in \parit)
\\
\varphi &= atom \pipe \varphi \wedge \varphi \pipe  \varphi \vee \varphi \pipe \neg \varphi 
\end{align*}
Given an MDP $\Gamma = ( S,E,Act,\prob)$, a state $s \in S$, and a parity condition $\parobj$, for the atomic formulas, we say that $s$ under strategy $\sigma$
%We now define when strategy $\sigma$ on MDP $\MDP$ with initial state $s$ \emph{satisfies} the following atomics objectives: 
\begin{itemize}
\item \emph{surely} satisfies $\parobj$, denoted $\satisfy{s}{\sigma}{\MDP}\univ(\parobj)$, iff $\\ \forall \pi \in \paths^{\MDPtoMC{\Gamma}{\sigma}}(s)$, we have that $\pi$ satisfies $\parobj$.
\item \emph{almost-surely} satisfies $\parobj$, denoted $\satisfy{s}{\sigma}{\MDP} \qual(\parobj)$, iff $\prob_{\MDPtoMC{\MDP}{\sigma}}(\{\pi \in \paths^{\MDPtoMC{\Gamma}{\sigma}}(s)\nd \pi$ satisfies $\parobj\})=1$.
\item satisfies $\parobj$ with \emph{non-zero probability}, denoted $\satisfy{s}{\sigma}{\MDP} \posi(\parobj)$, iff $\prob_{\MDPtoMC{\MDP}{\sigma}}(\{\pi \in \paths^{\MDPtoMC{\Gamma}{\sigma}}(s)\nd \pi$ satisfies $\parobj\})>0$.
\item \emph{existentially} satisfies $\parobj$, denoted $\satisfy{s}{\sigma}{\MDP} \exis(\parobj)$, iff $\exists \pi \in \paths^{\MDPtoMC{\Gamma}{\sigma}}(s)$, such that $\pi$ satisfies $\parobj$.
\end{itemize}

Given two \QPL formulas $\varphi$ and $\psi$, and a strategy $\sigma$ we define the semantics of Boolean connectives as follows:
\begin{itemize}
\item $\satisfy{s}{\sigma}{\MDP} \varphi\wedge \psi$ iff $\satisfy{s}{\sigma}{\MDP} \varphi$ and $\satisfy{s}{\sigma}{\MDP} \psi$
\item $\satisfy{s}{\sigma}{\MDP} \varphi\vee \psi$ iff $\satisfy{s}{\sigma}{\MDP} \varphi$ or $\satisfy{s}{\sigma}{\MDP} \psi$
\item $\satisfy{s}{\sigma}{\MDP} \neg \univ(\parobj)$ iff $\satisfy{s}{\sigma}{\MDP} \exis(\bar\parobj)$
\item $\satisfy{s}{\sigma}{\MDP} \neg \exis(\parobj)$ iff $\satisfy{s}{\sigma}{\MDP} \univ(\bar\parobj)$
\item $\satisfy{s}{\sigma}{\MDP} \neg \qual(\parobj)$ iff $\satisfy{s}{\sigma}{\MDP} \posi(\bar\parobj)$
\item $\satisfy{s}{\sigma}{\MDP} \neg \posi(\parobj)$ iff $\satisfy{s}{\sigma}{\MDP} \qual(\bar\parobj)$
\item $\satisfy{s}{\sigma}{\MDP} \neg (\varphi\wedge \psi)$ iff $\satisfy{s}{\sigma}{\MDP} \neg \varphi \vee \neg \psi$
\item $\satisfy{s}{\sigma}{\MDP} \neg (\varphi\vee \psi)$ iff $\satisfy{s}{\sigma}{\MDP} \neg \varphi \wedge \neg \psi$
\end{itemize}

Given a formula $\varphi$, we will use $\satisfies{s}{\MDP}\varphi$ to denote $\exists \sigma:\ \satisfy{s}{\sigma}{\MDP}\varphi$ .
Given a formula $\varphi$, let $\llbracket\varphi\rrbracket\eqdef\{s\in S\ |\ \satisfies{s}{\MDP}\varphi\}$.
We note that satisfying surely a parity condition is the same as winning the parity objective in the two-player game corresponding to the MDP $\MDP$. 
Satisfying existentially is the same as finding a satisfying path in the one-player game associated to this MDP.
% (that is, the one-player game that can also be considered as a two-player game where the players collaborate with each other).
% or equivalently having an accepting run on the nondeterministic automaton associated to this MDP.

%We also define the reachability atomic objectives: given some target set $T\subseteq S$, and an initial state $s$, then $\satisfy{s}{\sigma}{\MDP} F(\parobj)$ iff for all $\pi \in \paths^{\MDPtoMC{\Gamma}{\sigma}}(s)$, there exists $\tau$ a prefix of $\pi$ such that the last state of $\tau$ belongs to $T$.

% We now define the semantics of the boolean combinations of these formulas. A strategy satisfies a conjunction of formulas when 
% the paths it yields satisfies all these objectives. It satisfies a disjunction of formulas if it satisfies at least one of these formulas. 
%It will satisfy the negation of an objective if it can satisfy the dual of this objective applied to the dual objectives. 
Given an MDP $\MDP=(S,E, Act, \prob)$, a state $s \in S$, and a \QPL formula $\varphi$, the \emph{{\sf QPL}-synthesis problem} is to find a strategy $\sigma$ such that $\satisfy{s}{\sigma}{\MDP}\varphi$. The \emph{{\sf QPL}-realizability problem} is to decide whether $\satisfies{s}{\MDP}\varphi$. In what follows, we focus on the {\sf QPL}-realizability problem, but the algorithms we provide give all the elements necessary to build a winning strategy when such a strategy exists, and so they can be easily extended to solve {\sf QPL}-synthesis.

\begin{remark} \label{rem:negation}
We define the negation of the formulas using classical De Morgan's laws.
We note that the logic \QPL is closed under negation.
It is also important to note that in this semantics, $\satisfies{s}{\MDP} \neg \varphi$ is not equivalent to $s\nmodels_{\MDP} \varphi$. 
Indeed, $\satisfies{s}{\MDP} \neg \varphi$ implies that there exists a strategy $\sigma$ such that $\satisfy{s}{\sigma}{\MDP} \neg \varphi$ whereas $s \nmodels_{\MDP} \varphi$ implies that for all strategies $\sigma$, we have $\satisfy{s}{\sigma}{\MDP} \neg \varphi$.
% We state this formally in the following lemma.
%$\satisfies{s}{\MDP} \neg \varphi$ is not the same as $s \nmodels_{\MDP} \varphi$.
% Given an MDP $\Gamma = ( S,E,Act,\prob)$, a state $s \in S$, and a \QPL formula $\varphi$ we have that $s \nmodels_{\MDP} \varphi$ implies $\satisfies{s}{\MDP} \neg \varphi$ while the converse is not true.

%On the other hand, we have that $s \nmodels_{\MDP} \varphi$ implies that there exists a strategy $\sigma$ such that $\satisfy{s}{\sigma}{\MDP} \neg \varphi$.
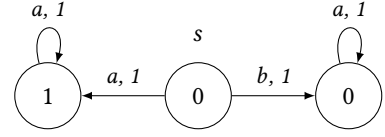
\begin{figure}[t]
\centering
\vspace{5pt}
\scalebox{1}{
	\begin{tikzpicture}
		\node () at (0,.8) {$s$} ;
		\node[player] (s) at (0,0) {$0$} ;
		\node[player] (s1) at (-2,0) {$1$} ;
		\node[player] (s2) at (2,0) {$0$} ;
		
		\path[-latex]  (s) edge node[above] {\small $a$, 1} (s1)
			(s) edge node[above] {\small $b$, 1} (s2)
			(s1) edge [loop above] node[above] {\small $a$, 1} (s1)
			(s2) edge [loop above] node[above] {\small $a$, 1} (s2)
		;
	\end{tikzpicture}
}
\caption{\label{fig:lemma}MDP in the proof of Remark \ref{rem:negation}}
\end{figure}

\subsection{Proof of  Remark \ref{rem:negation}}
\begin{proof}
% \sgcomment{The proof is straightforward; can be moved to appendix.}
Note we have that $s \nmodels_{\MDP} \varphi$ means that for all strategies $\sigma$, we have that $s, \sigma \nmodels_{\MDP} \varphi$, that is $\satisfy{s}{\sigma}{\MDP} \neg \varphi$. Hence $\satisfies{s}{\MDP} \neg \varphi$.

We now show that the converse does not hold using the example in Figure \ref{fig:lemma}.
On each edge, the name of the action and the corresponding probability is written and inside parentheses in each state, we write the value that the parity function assigned to the state.
Consider the formula $\varphi = A(\parobj)$.
Now $\satisfies{s}{\MDP} \varphi$ since there is a strategy (choosing action $b$ from s) such that  $\satisfy{s}{\sigma}{\MDP} \varphi$.
Also $\satisfies{s}{\MDP} \neg \varphi$, where $\neg \varphi$ is the formula $E(\bar\parobj)$ since there is a strategy (choosing action $a$ from s) such that  $\satisfy{s}{\sigma}{\MDP} \neg \varphi$.

Now for $s \nmodels_{\MDP} \varphi$ to be true, we need that for all strategies $\sigma$, we have that $\satisfy{s}{\sigma}{\MDP} \neg \varphi$, that is, for all strategies $\sigma$, we need that $\satisfy{s}{\sigma}{\MDP} E(\bar{\parobj})$ which is clearly not true since as we show above that there indeed exists a strategy $\sigma$ such that $\satisfy{s}{\sigma}{\MDP} \varphi$, and hence the result.
%\qed
\end{proof}

Similarly, even though $\satisfies{s}{\MDP} \varphi\vee \psi$ is equivalent to $\satisfies{s}{\MDP} \varphi $ or $\satisfies{s}{\MDP}\psi$, we note that $\satisfies{s}{\MDP} \varphi\wedge \psi$ is not the same as $ \satisfies{s}{\MDP} \varphi$ and $\satisfies{s}{\MDP}\psi$. 
Also using De Morgan's laws, the negation can be applied only to the parity objectives to get their duals, for example, $\neg(\univ(\parobj_1) \wedge \univ(\parobj_2))$ is the same as $\exis(\bar{\parobj}_1) \vee \exis(\bar{\parobj}_2)$.
We can indeed define a negation free normal form that can be obtained by taking the DNF and pushing negations down to the atoms.
In the rest of the paper, we thus restrict our attention to the subclass of the logic that is free of negation and disjunction.
\end{remark}
% This, along with the easy removal of negations by use of De Morgan's laws, and duality, we 
% %will end up studying the following fragment of our initial grammar: 
% study here only the conjunctive positive fragment of PSL given by the following grammar without $\vee$ and $\neg$.

% \begin{align*}
% atom &= \univ(\parobj) \pipe \exis(\parobj) \pipe \qual(\parobj) \pipe \posi(\parobj) \ (\parobj\in \parit)
% \\
% \varphi &= atom \pipe \varphi \wedge \varphi
% \end{align*}

\paragraph{Additional objectives}
We define the following additional objectives, introduced for technical reasons, even though they are not part of $\QPLnospace$\footnote{Given an MDP $\MDP$, we could have expressed these conditions using $\QPLnospace$, but this would involve constructing a larger and more complex MDP from the given MDP $\MDP$.}. 
A parity condition is called a \emph{B\"uchi condition} if it is defined as $\fun{\parobj}{S}{\{1,2\}}$. 
A path $\pi \in S^{\omega}$ satisfies a \emph{conjunction of parity conditions} $\bigwedge_{x\in X}\parobj_x$ if for all $x\in X$ we have $\max\{\parobj_x(s)\:|\: s\in \inf(\pi)\}$ is even. 
% We call it a Streett condition. 
It is not hard to see that conjunctions of parity conditions can be expressed as \emph{Streett conditions}
% \footnote{Recall that a \emph{Streett condition} and a conjunction of parity conditions have different algorithmic complexity (the complexity of transforming multiple parity conditions, to a Streett condition, and then solving the Streett condition is higher than solving a Streett condition~\cite{CHP07}) but they belong to the same complexity class.}
.
% In this paper, we only study the complexity class of our different problems.
A path $\pi$ satisfies a \emph{reachability condition} towards a set $R\subseteq S$, denoted $\event R$, if there exists $i\in\mathbb{N}_0$ such that $\pi(i)\in R$.
% Given a path $\pi$, and a condition $\varphi$, then $\event \varphi$ denotes that there exists $i\in\mathbb{N}$ such that path $\pi(i,\infty)$ satisfies $\varphi$. 
% % We sometimes use $\event R$ for $\event Reach(R)$.
% Given a path $\pi$ and any condition $\varphi$, then $\globally \varphi$ denotes that for all $i\in\mathbb{N}$ the path $\pi(i,\infty)$ satisfies $\varphi$. 
% % We sometimes use $\globally R$ for $\globally Reach(R)$.
% We use notations $\globally R$, $\event \globally R$, $\globally \event R$ with same semantics as in LTL.

Given an MDP, we can define, in the same way as previously the sure, almost-sure, non-zero, and existential objectives for these conditions, as well as conjunctions and disjunctions of these objectives.

\paragraph{End-components}
An \emph{end-component} (EC, for short) $M = (C,A)$ such that $C \subseteq S$, and $A:C \to 2^{Act}$
% \subseteq S \times Act$ 
is a \emph{sub-MDP} of $\Gamma$ (for all $s \in C, \text{ we have } A(s) \subseteq Act(s)$, and for all $a \in A(s)$, we have $\post(s,a) \subseteq C$) that is strongly connected. 
We denote by $EC(\MDP)$ the set of end-components of MDP $\MDP$.
% An \emph{end-component} (EC, for short) $C \subseteq S$ of an MDP $\Gamma$ is a strongly connected component of $\Gamma$ such that there exists $A \subseteq Act$ and for all all $s \in C, a \in A$, we have $\post(s,a) \subseteq T$.
By abuse of notation, in the sequel, we often refer to a set $C \subseteq S$ to be an end-component when there exists a function $A: C \to 2^{Act}$ such that $(C,A)$ is an end-component.
A \emph{maximal EC} (MEC, for short) is an EC that is not included in any other EC. 
% We denote by $EC(\Gamma)$ the set of all end-components of $\Gamma$. We denote by $\mec(\Gamma)$ the set of all maximal end-components of $\Gamma$. 
% For every strategy, each infinite path in an MDP will eventually end up in one maximal end-component almost surely, whatever strategy is considered. This is stated in the following lemma:
% \begin{comment}
For every strategy in an MDP, the set of states seen infinitely often during a path form an end-component with probability~$1$.
Formally:
\begin{proposition}[\cite{BK08}]
\label{prop:long_run}
Given an MDP $\MDP$, for all strategies $\sigma$, 
% , for all equivalent transition systems $\mathtt{T}_\sigma=\zug{Q,S,act,\delta,\iota}$
for all states $s\in S$, we have $\prob_{\MDPtoMC{\MDP}{\sigma}}(\{\pi \in \paths^{\MDPtoMC{\Gamma}{\sigma}}(s) \mid   \inf(\pi)\in EC(\MDP)\})= 1$.
% , for all states $s\in S$, we have $\prob(\{\pi \in \paths^{\MDPtoMC{\Gamma}{\sigma}}(\iota(s),s) \mid   \inf(\proj{S}(\pi))\in EC(\MDP)\})= 1$.
\end{proposition}
% \end{comment}
% By abuse of notation, in the sequel, we use $\paths^{\MDPtoMC{\Gamma}{\sigma}}(s)$ instead of $\paths^{\MDPtoMC{\Gamma}{\sigma}}(q,s)$.
%\subsection{Weighted Two-Player Games}
%\label{prelim_game}

% \textbf{Ultra Good End-Components (UGECs) and Very Good End-Components (VGECs).}
% \sgcomment{Do we change the order for type 1 (VGEC) and type 2 (UGEC)?}

\section{\typeone end-components}
\label{sec:typeone}
In this section, we define \typeone ECs that are a generalization of super-good end components as defined in \cite{AKV16}.
% In the following sections, we use other kinds of end-components whose definitions are also related to that of \typeone EC.

% The intuition behind these \typeone ECs is that there exists a strategy that satisfies some parity objectives and at the same time it also reaches with probability $1$ the maximum value (that is also even) for some set of other parity objectives (possibly overlapping with the first set of objectives).
% In Lemma~\ref{thm:reach-Almagor}, we introduce Street-B\"uchi games that we use to check a formula made of a combination of sure and almost-sure conditions. 
% Using Lemma~\ref{lem:implies_is_univ}, we show how to use this specific formula to compute some specific \typeone ECs.
% We then use the technical Lemma~\ref{lem:split-typeone} to relate these specific \typeone ECs to a more general set of \typeone ECs, that are of interest to us.
Lemma~\ref{lem:alg-typeone} is the main result of the section, where we state that we can compute the set of maximal \typeone ECs. This will be used later, to compute the set of maximal ECs of other kinds, namely \typetwo and \typethree that are used in Sections~\ref{sec:typetwo} and~\ref{sec:typethree} to solve satisfiability of formulas of the form $\limwedgeone{\induniv\in{\setuniv}}\univ(\parobj_{\induniv})\wedge \limwedgeone{\indqual\in{\setqual}}\qual(\parobj_{\indqual})$ and $\combi{\setuniv}{\setqual}\wedge \posi(\parobj_{\indposi})$ respectively.
Lemmas~\ref{thm:reach-Almagor},~\ref{lem:implies_is_univ} and~\ref{lem:split-typeone} are technical lemmas that are required in the proof of Lemma~\ref{lem:alg-typeone}.
The proof of Lemma~\ref{thm:reach-Almagor} uses the notion of Street-B\"uchi games.

Given two sets of parity conditions $\{\parobj_{\induniv}\pipe \induniv\in\setuniv\}$ and $\{\parobj_{\indqual}\pipe \indqual\in\setqual\}$, an end-component $C$ of $\MDP$ is \typeone($\setuniv,\setqual$)
if the following property holds:

\begin{itemize}
	\item $\mathbf{(I_1)}$ $\forall\, s\in C,\, \satisfies{s}{\MDP{\downharpoonright C}} \limwedgeone{\induniv\in\setuniv}\univ(\parobj_{\induniv}) \wedge \limwedgeone{\indqual\in\setqual} \qual(\event C^{\max}_{\even}(\parobj_{\indqual}))$, where	
	$$
	C^{\max}_{\even}(\parobj_{\indqual}) = \big\lbrace s\in C \mid (\parobj_{\indqual}(s) \text{ is even}) \wedge \big(\forall\, s'\in C,\,$$ $$\parobj_{\indqual}(s') \text{ is odd } \implies \parobj_{\indqual}(s') < \parobj_{\indqual}(s)\big)\big\rbrace$$
	contains the states with even priorities that are larger than any odd priority in $C$ (this set can be empty for arbitrary ECs but needs to be non-empty for $\\ $ {\sf Type\:{\rm I}}($\setuniv,\setqual$) ECs);
\end{itemize}

We write \typeone($\setuniv,\setqual$) EC as \typeone EC when the parity sets are clear from the context.
We introduce the following notations: $\sgec(\MDP,\setuniv,\setqual)$ is the set of all \typeone($\setuniv,\setqual$) ECs, and $\mathcal{T}_{\mathbf{I},\MDP,\setuniv,\setqual} = {\displaystyle \cup_{U \in \sgec(\MDP,\setuniv,\setqual)} U}$ is the set of states belonging to some \typeone EC. 
Given an EC $C$, we say a state $s\in C$ is of \typeone for $C$ if $C$ is $\typeone$.
In this paper, we only consider 
% \typeone($\setuniv,\setuniv$) ECs, and \typeone($\setuniv,\{\induniv\}$)  ECs with $\induniv\in\setuniv$.
\typeone($\setuniv,\setqual$) ECs where $\setqual$ is either $\setuniv$ or $\{\induniv\}$.
% we have $\satisfies{s}{\MDP{\downharpoonright C}} \limwedgeone{\induniv\in\setuniv}\univ(\parobj_{\induniv}) \wedge \limwedgeone{\indqual\in\setqual} \qual(\event C^{\max}_{\even}(\parobj_{\indqual}))$.
% For a \QPL formula $\varphi$, we denote by $\sigma_{\varphi}$ a strategy satisfying this formula.\rbchanged{Do we use this?}

Intuitively, within a \typeone EC, 
%$\playerOne$ has 
there is a strategy to visit all $C^{\max}_{\even}(\parobj_{\indqual})$ for all $\indqual\in\setqual$ with probability $1$ while guaranteeing $\univ(\parobj_{\induniv})$ for all $\induniv\in\setuniv$. 
%he also has 
We note that this property must hold while staying inside the end-component $C$.
%Figure~\ref{fig:full} gives an example of UGEC.
This notion strengthens the notion of \textit{super-good end-component} (SGEC in \cite{AKV16}), that are defined for some parity condition $\parobj_a$, and are \typeone($\{a\},\{a\}$) ECs.
% Finding solutions for  $\mathbf{(I_1)}$  isn't very classical. 
In the case of SGEC, it has been shown in~\cite{AKV16} that the existence of a strategy to enforce condition $\mathbf{I_1}$ in $\MDP$ can be reduced to checking the existence of a winning strategy in a game, constructed in polynomial time from $\MDP$, with a conjunction of one parity objective and one B\"uchi objective.
% such a condition can be transformed into a parity-B\"uchi game, that is a game with one parity and one B\"uchi condition where 
The existence of a winning strategy in such a game is in ${\sf NP} \cap {\sf coNP}$.
% The existence of a winning strategy in the parity-B\"uchi game also implies the existence of a winning strategy for the original $\univ(\parobj) \wedge \qual(\event C^{\max}_{\even}(\parobj))$ condition.
The structure of this game is different from the one of the original MDP, as its size is polynomially increased to transform the qualitative reachability condition into a sure B\"uchi. In the sequel, we generalize this result to multiple parity conditions. We illustrate the reduction by the following example.

\begin{example}
Consider the example in Figure~\ref{fig:SGEC} where an MDP (on the left side of the figure) that is a \typeone($\{a\},\{a\}$) EC for a parity condition $\parobj_a$ is transformed into a game (on the right side of the figure) that satisfies $\univ(\parobj_a) \wedge \univ(\Box\event R)$. In order to convert the $\qual(\event C^{\max}_{\even}(\parobj_{a}))$ condition of $\mathbf{(I_1)}$ into the $\univ(\Box\event R)$ condition, we add two states to the game: The top-most state and the bottom-most state.
% We restrict our attention to the end-component inside the dashed box that is a \typetwo EC.
In the MDP on the left of Figure~\ref{fig:GEC}, a strategy that alternates between playing action $a$ and playing action $b$ at state $s$ indeed satisfies the condition $\mathbf{(I_1)}$.
% \begin{figure}[t]
% \centering
% \vspace{-15pt}
% \hspace{-15pt}
% \centering
% \scalebox{0.9}{
% 	\begin{tikzpicture}
% 		\node[player,initial,initial text={}] (sinit) at (2,2) {0,1} ;
% 		\node[player] (s1) at (2,0) {0,1} ;
		
% 		\node[player] (s2) at (4,0) {2,2} ;
% 		\node[player] (s3) at (4,2) {1,1} ;
% 		\node[player] (s4) at (6,2) {0,0} ;
% 		\node[player] (s5) at (8,2) {0,5} ;
% 		\node[player] (s6) at (8,0) {2,1} ;

% 		\node[player] (s7) at (6,4) {3,3} ;
% 		\node[player] (s8) at (8,4) {4,4} ;

% 		\node () at (6,1.3) {\large {\color{blue}$s$}} ;
% 	    \node[subgraph] (Nc) at (5.75,1) {};

%  \path[-latex]  (sinit) edge[bend left=20] node[right] {\small $a,1$} (s1)
% 				(s1) edge[bend left=20] node[left] {\small $a,0.5$} (sinit)				
% 				(s1) edge node[above] {\small $a,0.5$} (s2)
% 				(s2) edge[bend right=20] node[right] {\small $a,1$} (s4)
% 				(s4) edge[bend right=20] node[above] {\small $b,1$} (s3)
% 				(s3) edge node[left] {\small $a,0.5$} (s2)
% 				(s3) edge[bend right=20] node[below] {\small $a,0.5$} (s4)
% 				(s4) edge[bend left=20] node[right, yshift=.1cm] {\small $a,1$} (s5)
% 				(s5) edge[bend left=20] node[below] {\small $a,0.5$} (s4)
% 				(s5) edge node[right] {\small $a,0.5$} (s6)
% 				(s6) edge[bend left=20] node[left] {\small $a,1$} (s4)
				
% 				(s4) edge[bend left=20] node[left] {\small $c,1$} (s7)
% 				(s7) edge[bend left=20] node[right] {\small $a,0.5$} (s4)
% 				(s7) edge node[above] {\small $a,0.5$} (s8)
% 				(s8) edge[bend left=20] node[above,xshift=.6cm] {\small $a,1$} (s4)
				
% 		;
% 	\end{tikzpicture}
% }
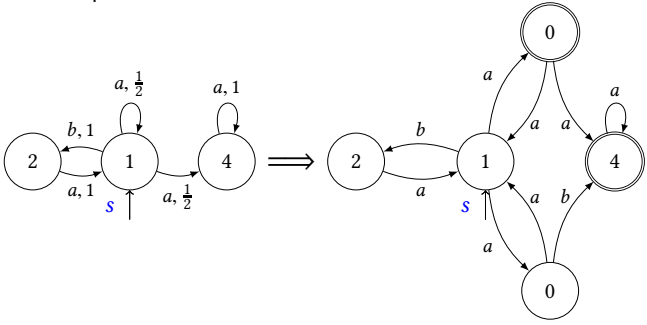
\begin{figure}[t]
\centering
\vspace{-15pt}
\hspace{-15pt}
\centering
\scalebox{0.86}{
	\begin{tikzpicture}
% 		\node[player] (sinit) at (2,2) {0,1} ;
% 		\node[player] (s1) at (2,0) {0,1} ;
		
		\node[player,initial,initial below,initial text={}] (s2) at (2.5,0) {1} ;
		\node[player] (s3) at (1,0) {2} ;
		\node[player] (s4) at (4,0) {4} ;
		
		\node (l) at (5,0) {{\LARGE $\Longrightarrow$}} ;
		
		\node[player,initial,initial below,initial text={}] (q2) at (8,0) {1} ;
		\node[player] (q3) at (6,0) {2} ;
		\node[player] (q4a) at (9,-2) {0} ;
		\node[player,double] (q4b) at (9,2) {0} ;
		\node[player,double] (q5) at (10,0) {4} ;

% 		\node[player] (s7) at (6,4) {3,3} ;
% 		\node[player] (s8) at (8,4) {4,4} ;

		\node () at (2.2,-0.7) {\large {\color{blue}$s$}} ;
		\node () at (7.7,-0.7) {\large {\color{blue}$s$}} ;
	   % \node[subgraph] (Nc) at (5.75,1) {};

 \path[-latex]  
    %             (sinit) edge[bend left=20] node[right] {\small $a,1$} (s1)
				% (s1) edge[bend left=20] node[left] {\small $a,0.5$} (sinit)				
				% (s1) edge node[above] {\small $a,0.5$} (s2)
				(s2) edge[bend right=20] node[above] {\small $b,1$} (s3)
				(s3) edge[bend right=20] node[below] {\small $a,1$} (s2)
	    		(s2) edge [loop above] node[above] {\small $a, \frac{1}{2}$} (s2)
				(s2) edge[bend right=20] node[below] {\small $a, \frac{1}{2}$} (s4)
	    		(s4) edge [loop above] node[above] {\small $a, 1$} (s4)

				(q2) edge[bend right=20] node[above] {\small $b$} (q3)
				(q3) edge[bend right=20] node[below] {\small $a$} (q2)
	    		(q2) edge[bend right=20] node[left,yshift=-.2cm] {\small $a$} (q4a)
				(q2) edge[bend left=20] node[left,yshift=.2cm] {\small $a$} (q4b)
	    		(q4a) edge[bend right=20] node[above,yshift=.1cm] {\small $a$} (q2)
				(q4b) edge[bend left=20] node[below,yshift=-.1cm] {\small $a$} (q2)
	    		(q4a) edge[bend left=20] node[above,yshift=.1cm] {\small $b$} (q5)
				(q4b) edge[bend right=20] node[below,yshift=-.1cm] {\small $a$} (q5)
	    		(q5) edge [loop above] node[above] {\small $a$} (q5)
				% (s4) edge[bend right=20] node[above] {\small $b,1$} (s2)
				% (s3) edge node[left] {\small $a,0.5$} (s2)
				% (s3) edge[bend right=20] node[below] {\small $a,0.5$} (s4)
				% (s4) edge[bend left=20] node[right, yshift=.1cm] {\small $a,1$} (s5)
				% (s5) edge[bend left=20] node[below] {\small $a,0.5$} (s4)
				% (s5) edge node[right] {\small $a,0.5$} (s6)
				% (s6) edge[bend left=20] node[left] {\small $a,1$} (s4)
	   % 		(s2) edge [loop above] node[above] {\small $a$, 1} (s2)
				
				% (s4) edge[bend left=20] node[left] {\small $c,1$} (s7)
				% (s7) edge[bend left=20] node[right] {\small $a,0.5$} (s4)
				% (s7) edge node[above] {\small $a,0.5$} (s8)
				% (s8) edge[bend left=20] node[above,xshift=.6cm] {\small $a,1$} (s4)
				
		;
	\end{tikzpicture}
}
\caption{\label{fig:SGEC}An example of a \typeone EC at left, and a game associated to it at right.}
\end{figure}
Now consider the game on the right side of Figure~\ref{fig:SGEC}.
The top-most state and the right-most state shown in double circles form the set $R$.
A strategy 
% in the game on the right side of Figure~\ref{fig:SGEC} 
to satisfy $\univ(\parobj_a) \wedge \univ(\Box\event R)$ 
% where the states in $R$ are the top-most and right-most ones plays 
is as follows: When in state $s$, alternate between playing action $a$ and playing action $b$. When in the bottom-most state, play action $b$.
The $\univ(\parobj_a)$ atom is clearly satisfied. The sure B\"uchi $\univ(\Box\event R)$ holds for the following reason. When action $a$ is chosen in state $s$, if player $2$ chooses to go to the bottom-most state, then from this state player $1$ plays action $b$, and reaches the right-most state that is absorbing and in $R$. If player $2$ chooses to go to the top-most state always, as this state is in $R$, the B\"uchi condition is again satisfied.

We now state the first step of the reduction.
The result below relies on a reduction to a two-player game $G^{\MDP}_{R,\{\parobj_{\induniv}\pipe \induniv\in\setuniv\}}$ with a conjunction of one B\"uchi and multiple parity conditions, that we call a Streett-B\"uchi game.
The approach to the proof is similar to Lemma 3 of~\cite{AKV16}, that studies the case where $\setuniv$ is a singleton. 
\end{example}
\begin{lemma}
\label{thm:reach-Almagor}
Given an MDP $\MDP=( S,E,Act,\prob)$, a state $s_0 \in S$, a set of parity conditions $\{\parobj_{\induniv}\pipe \induniv\in\setuniv\}$, and a target set $R \subseteq S$, it can be decided if $\satisfies{s_0}{\MDP} \limwedgeone{\induniv\in\setuniv}\univ(\neg (\event R)\rightarrow \parobj_a) \wedge \qual(\event R) $. If the answer is $\yes$, then there exists a finite-memory witness strategy. This decision problem is ${\sf coNP}$ complete.
\end{lemma}
% The properties that are of interest to us require us to compute the maximal \typeone($\setuniv,\setuniv$) ECs. Handled naively, we should double the size of the structure for each element in $\setuniv$, resulting in a game of size exponential in $\abs{\setuniv}$. We show that it is possible to compute these ECs by only considering for all $\induniv_i\in\setuniv$ the \typeone($\setuniv,\{\parobj_{\induniv_i}\}$) ECs. This is formalized in the following statement:

% \sgcomment{This paragraph can be removed.}
% This lemma relies on a reduction to a two-player game $G^{\MDP}_{R,\{\parobj_{\induniv}\pipe \induniv\in\setuniv\}}$ with a conjunction of one B\"uchi and multiple parity conditions, that we call a Streett-B\"uchi game. A formal definition of this game is given in Appendix~\ref{app:typeone}. The approach is the same as in Lemma 3 of~\cite{AKV16}, that studies the case where $\setuniv$ is a singleton. 
%%PROOF BEGINS HERE

\begin{proof}
To establish this lemma, given an MDP $\MDP$, a set of parity conditions $\{\parobj_{\induniv}\pipe \induniv\in\setuniv\}$, and a target set $R \subseteq S$, we construct a game $G^{\MDP}_{R,\{\parobj_{\induniv}\pipe \induniv\in\setuniv\}}$ with a conjunction of a B\"uchi and multiple parity conditions. 
We call this game a Streett-B\"uchi game, and its formal definition is as follows:
% We show the construction here.
% For the ease of explanation and w.l.o.g., we make the hypothesis that in $\MDP$, every state belongs either to $\playerOne$ or $\playerTwo$ and we have $E \subseteq S_1 \times S_2 \cup S_2 \times S_1$, i.e., states of $\playerOne$ and $\playerTwo$ alternate. Then 
% The MDP $G^{\MDP}_{R,\{\parobj_{\induniv}\pipe \induniv\in\setuniv\}}=( S',E',Act')$ where.
  \begin{itemize}
  	\item The state space of $G^{\MDP}_{R,\{\parobj_{\induniv}\pipe \induniv\in\setuniv\}}$ is a copy of the state space of $\MDP$ where the states in $S \setminus R$ have been copied $2\cdot \abs{S} \cdot \abs{R} \cdot \abs{Act}$ times as shown in Figure~\ref{fig:ggt}. $S'$ is defined as $S'=S \cup ((S \setminus R) \times Act \times \{0,1 \})$. 
%   	States labelled with $(s,a,0)$ represent environment-controlled states in a game, while states labelled $(s,a,1)$ represent player-controlled states.
  	\item We consider an arbitrary set of actions $Act'$ such that $Act\subseteq Act'$ and for all $s\in S\setminus R$ and $a\in Act$, we have $\abs{Act'}\geq \abs{\{s'\pipe (s,a,s')\in E\}}$. This second condition is to make sure that in the states labelled $(s,a,1)$, we always have enough actions available to choose the successor state.
	\item The set $E'$ of edges is defined according to Figure~\ref{fig:ggt} with the additional property that states in $R$ are made absorbing. That is, $E'$ is the union of the following sets:
	\begin{itemize}
% 		\item $\{ (s,a,s') \mid (s,a,s') \in E \cap ((S\setminus R) \times Act \times \{0,1\} \times S) \}$,
		\item $\{ (s,a,(s,a,0)) \mid s \in S \setminus R,\, a\in Act(s) ) \}$,
		\item $\{ (s,a,(s,a,1)) \mid s \in S \setminus R,\, a\in Act(s) ) \}$,
		\item $\{ ((s,a,0),a,s') \mid (s,a,s') \in E \cap ((S\setminus R) \times a\times S) \}$,
		\item $\{ ((s,a,1),a_{s,a,s'},s') \mid (s,a,s') \in E \cap ((S\setminus R) \times Act \times S),\, \text{and for all }\\ {(s,a,s'),(s,a,s'')\in S} \ \text{with } s'\neq s'',\,\text{we have } a_{s,a,s'}\neq a_{s,a,s''} \}$,
		\item $\{ (s,a,s) \mid s \in R\text{ for one }a\in Act \}$.
	\end{itemize}
	
% 	\item The new probability probability distribution $\prob'$ can be arbitrary as long as every $e\in E'$ has non-zero probability.
	\item The parity conditions $\parobj'_a$ are defined as $\parobj'_a(s)=\parobj_a(s)$ for all $s \in S \setminus R$, $\parobj'_a((s,a,i))=0$ for all $s \in S \setminus R$, $a\in Act$, $i\in\{0,1\}$, and $\parobj'_a(s)=0$ for all $s \in R$. We have to change the priority of states of $R$, because as they are made absorbing they may be losing for the parity condition otherwise, as on Figure~\ref{fig:rv}.

	\item The set of B\"uchi states is $B=R \cup \{ (s,a,0) \mid s\in S\setminus R,\, a\in Act \}$.
\end{itemize}

% \draw (6,0) node[carre,bleu] (s1) {$i$};
% \draw (8,1) node[carre,double,thick,jaune] (s2) {$0$};
% \draw (8,-1) node[rond,thick,jaune] (s3) {$0$};
% \draw (4,0) node (l) {{\LARGE $\Longrightarrow$}};

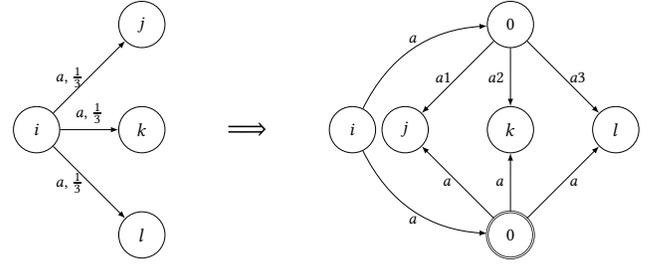
\begin{figure}[t]
\centering
\vspace{5pt}
\scalebox{0.7}{
	\begin{tikzpicture}
		\node[player] (s0) at (0,0) {$i$} ;
		\node[player] (q1) at (2,2) {$j$} ;
		\node[player] (q2) at (2,0) {$k$} ;
		\node[player] (q3) at (2,-2) {$l$} ;
		\node[player] (s1) at (6,0) {$i$} ;
		\node[player] (s2) at (9,2) {$0$} ;
		\node[player,double] (s3) at (9,-2) {$0$} ;
		\node[player] (q4) at (7,0) {$j$} ;
		\node[player] (q5) at (9,0) {$k$} ;
		\node[player] (q6) at (11,0) {$l$} ;
		\node (l) at (4,0) {{\LARGE $\Longrightarrow$}} ;
		
		\path[-latex]  
			(s0) edge node[left] {\small $a$, $\frac{1}{3}$} (q1)
			(s0) edge node[above] {\small $a$, $\frac{1}{3}$} (q2)
			(s0) edge node[left] {\small $a$, $\frac{1}{3}$} (q3)
			
			(s1) edge[bend left=30] node[above] {\small $a$} (s2)
			(s1) edge[bend right=30] node[below] {\small $a$} (s3)
			
			(s3) edge node[left] {\small $a$} (q4)
			(s3) edge node[left] {\small $a$} (q5)
			(s3) edge node[right] {\small $a$} (q6)
			
			(s2) edge node[left] {\small $a1$} (q4)
			(s2) edge node[left] {\small $a2$} (q5)
			(s2) edge node[right] {\small $a3$} (q6)
		;
	\end{tikzpicture}
}
\caption{Modifying edges associated to action $a$ on a state $s$ with priority $i$ by adding new edges and new states. The bottom-most state in the figure at right corresponds to $(s,a,0)$ and is B\"uchi winning. The top-most state in the figure corresponds to $(s,a,1)$.}
\label{fig:ggt}
\end{figure}

% \begin{figure}[tbh]
% \centering
% \scalebox{0.8}{\begin{tikzpicture}[every node/.style={font=\small,inner sep=1pt}]
% \draw (0,0) node[carre,bleu] (s0) {$i$};
% \draw (6,0) node[carre,bleu] (s1) {$i$};
% \draw (8,1) node[carre,double,thick,jaune] (s2) {$0$};
% \draw (8,-1) node[rond,thick,jaune] (s3) {$0$};
% \draw (4,0) node (l) {{\LARGE $\Longrightarrow$}};
% \draw[-latex] (s0) to (1,0.5);
% \draw[-latex] (s0) to (1,0);
% \draw[-latex] (s0) to (1,-0.5);
% \draw[-latex] (s2) to (9,1.5);
% \draw[-latex] (s2) to (9,1);
% \draw[-latex] (s2) to (9,0.5);
% \draw[-latex] (s3) to (9,-1.5);
% \draw[-latex] (s3) to (9,-1);
% \draw[-latex] (s3) to (9,-0.5);
% \draw[-latex] (s1) to (s2);
% \draw[-latex] (s1) to (s3);
% \end{tikzpicture}}
% \caption{Replacing a $\player{2}$ state with priority $i$ by three new states, including a B\"uchi one.}
% \label{fig:ggt}
% \end{figure}

\begin{figure}[tbh]
\centering
\scalebox{0.8}{\begin{tikzpicture}[every node/.style={font=\small,inner sep=1pt}]
\draw (-1.5,0) node[rond] (s0) {$0$};
\draw (0,0) node[rond] (s1) {$1,R$};
\draw (1.5,0) node[rond] (s2) {$0$};
\draw (4.5,0) node[rond] (s3) {$0$};
\draw (6,0) node[rond,double,thick] (s4) {$0,R$};
\draw (7.5,0) node[rond] (s5) {$0$};
% \draw (8,1) node[carre,double,thick,jaune] (s2) {$0$};
% \draw (8,-1) node[rond,thick,jaune] (s3) {$0$};
\draw (3,0) node (l) {{\LARGE $\Longrightarrow$}};
\draw[-latex] (s0) to (s1);
\draw[-latex] (s1) to (s2);
\draw[-latex] (s3) to (s4);
\draw (s4) edge[-latex,out=60,in=120,looseness=4,distance=1cm] (s4);
\end{tikzpicture}}
\caption{We change the parity value of the target state, since they are made absorbing, otherwise there would be no winning strategy in the parity-B\"uchi game while there is a parity winning strategy ensuring to reach $R$ in the initial MDP.}
\label{fig:rv}
\end{figure}
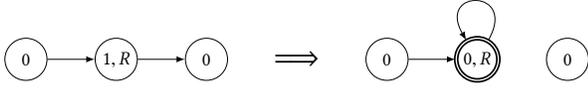

It is established in~\cite{AKV16} that for one parity condition $p$ we have $s_0 \models \univ(\neg (\Box\event R)\rightarrow p) \wedge \qual(\Box\event R)$ if and only if there exists a winning strategy in the B\"uchi parity game $G^{\MDP}_{R,\{\parobj_{\induniv}\pipe \induniv\in\setuniv\}}$ from state $s_0$. The proof of the result can be generalized to our setting: $\satisfies{s_0}{\MDP} \limwedgeone{\induniv\in\setuniv}\univ(\neg (\event R)\rightarrow \parobj_a) \wedge \qual(\event R) $ iff   $\satisfies{s_0}{G^{\MDP}_{R,\{\parobj_{\induniv}\pipe \induniv\in\setuniv\}}}\limwedgeone{\induniv\in\setuniv}\univ(\parobj_a)\wedge \univ(\Box\event B)$ where $B$ has been defined above. Finding a strategy for the second objective above is solving a Streett game, which is classically co-NP complete.
We do not give the full proof of this result, as there would be no change from~\cite{AKV16}, but give a quick outline. Note that corresponding to every path in MDP $\MDP$, there exists a path in the game $G^{\MDP}_{R,\{\parobj_{\induniv}\pipe \induniv\in\setuniv\}}$, and vice versa. 

Given a finite-memory strategy $\sigma$ for the Streett-B\"uchi game, we obtain a strategy $\sigma'$ for the MDP by playing exactly as in the game. 
We have that every path in the game under strategy $\sigma$ either reaches $R$ or satisfies for all $\induniv\in\setuniv$ the condition $\parobj_{\induniv}$.
Since for every path in the MDP under strategy $\sigma'$; there exists a corresponding path in the game under strategy $\sigma$, every path in the MDP under strategy $\sigma'$  also either reaches $R$ or satisfies for all $\induniv\in\setuniv$ the condition $\parobj_{\induniv}$. 
Since $\sigma'$ is a finite-memory strategy, in every state of the MDP there is a non-zero probability of reaching $R$ in some fixed $k$ steps.
This probability is also always greater than some lower bound, hence the probability of reaching $R$ is $1$. 

Reciprocally, given a winning strategy $\sigma$ in an MDP $\MDP$, there exists a winning strategy $\sigma'$ in  $G^{\MDP}_{R,\{\parobj_{\induniv}\pipe \induniv\in\setuniv\}}$. 
The strategy $\sigma'$ consists in behaving like $\sigma$ in $\MDP$, and in states $(s,a,1)$, where $\sigma$ is not defined,  strategy $\sigma'$ takes the action which corresponds to a probabilistic transition in $\MDP$ that would lead to the shortest possible path to $R$.
In the MDP under strategy $\sigma$ every path that does not reach $R$ satisfies for all  $\induniv\in\setuniv$ the condition $\parobj_{\induniv}$.
Since corresponding to every path in the game under strategy $\sigma'$ there exists a path in the MDP under strategy $\sigma$, we have that in the game under strategy $\sigma'$ every path that does not reach $R$ satisfies for all  $\induniv\in\setuniv$ the condition $\parobj_{\induniv}$. 
Now, for condition $B$, there are two possibilities. 
The first one is that some $(s,a,0)$ is visited infinitely often, and the condition $B$ is satisfied. 
The second one is that eventually the states $(s,a,0)$ are not visited any more. 
In this case in every $(s,a,1)$, the action leading to the shortest path to $R$ is taken. 
Then the set $R$ is eventually reached and $B$ is also satisfied.
\end{proof}

%%PROOF ENDS HERE
The following lemma relates the $\limwedgeone{\induniv\in\setuniv}\univ(\neg (\event R)\rightarrow \parobj_a) \wedge \qual(\event R) $ objective of Lemma~\ref{thm:reach-Almagor} and the $\limwedgeone{\induniv\in\setuniv}\univ(\parobj_a) \wedge \qual(\event R) $ objective of \typeone($\setuniv,\{ a\}$) ECs under the condition that for all $s\in S$ it holds that $\satisfies{s}{\MDP} \limwedgeone{\induniv\in\setuniv}\univ(\parobj_a)$. 
\begin{lemma}
\label{lem:implies_is_univ}
Given an MDP $\MDP=( S,E,Act,\prob)$, a state $s_0 \in S$, a set of parity conditions $\{\parobj_{\induniv}\pipe \induniv\in\setuniv\}$, and a target set $R \subseteq S$, if for all $s\in S$ it holds that $\satisfies{s}{\MDP} \limwedgeone{\induniv\in\setuniv}\univ(\parobj_a)$ then we have that $\satisfies{s_0}{\MDP} \limwedgeone{\induniv\in\setuniv}\univ(\neg (\event R)\rightarrow \parobj_a) \wedge \qual(\event R) $ if and only if $\satisfies{s_0}{\MDP} \limwedgeone{\induniv\in\setuniv}\univ(\parobj_a) \wedge \qual(\event R) $.
\end{lemma}
\begin{proof}
The right to left implication is simple, as $\limwedgeone{\induniv\in\setuniv}\univ(\parobj_a) \wedge \qual(\event R)$ implies $\limwedgeone{\induniv\in\setuniv}\univ(\neg (\event R)\rightarrow \parobj_a) \wedge \qual(\event R)$. For the left to right implication, consider a strategy $\sigma_0$ such that $\satisfy{s_0}{\sigma_0}{\MDP} \limwedgeone{\induniv\in\setuniv}\univ(\neg (\event R)\rightarrow \parobj_a) \wedge \qual(\event R) $. Recall that from every $s\in S$ we have $\sigma_s$ such that  $\satisfy{s}{\sigma_s}{\MDP} \limwedgeone{\induniv\in\setuniv}\univ(\parobj_a)$. Consider the following strategy $\sigma$: it plays like $\sigma_0$, but when it reaches some $r\in R$, it switches to play like $\sigma_r$ forever. It clearly satisfies $\qual(\event R)$. It also satisfies $\limwedgeone{\induniv\in\setuniv}\univ(\parobj_a)$: indeed, on a path $\pi$, either $R$ is never reached, and thanks to the use of $\sigma_0$ we have that $\univ(\neg (\event R)\rightarrow \parobj_a)$ implying $\parobj_a$ holds on this path for all $\induniv\in\setuniv$. Otherwise, if $\pi$ eventually reaches $R$, $\sigma$ switches to some strategy $\sigma_r$ that satisfies $\limwedgeone{\induniv\in\setuniv}\univ(\parobj_a)$ by assumption.
\end{proof}

%  The proof of this lemma can be found in Appendix~\ref{app:typeone}. 
%  This lemma relates the $\limwedgeone{\induniv\in\setuniv}\univ(\neg (\event R)\rightarrow \parobj_a) \wedge \qual(\event R) $ objective of Lemma~\ref{thm:reach-Almagor} and the $\limwedgeone{\induniv\in\setuniv}\univ(\parobj_a) \wedge \qual(\event R) $ objective of \typeone($\setuniv,\{ a\}$) ECs under the condition that for all $s\in S$ it holds that $\satisfies{s}{\MDP} \limwedgeone{\induniv\in\setuniv}\univ(\parobj_a)$. 
 As $\mathbf{I_1}$ implies $\satisfies{s}{\MDP} \limwedgeone{\induniv\in\setuniv}\univ(\parobj_a)$, pruning states that do not satisfy $\limwedgeone{\induniv\in\setuniv}\univ(\parobj_a)$ before using Lemma~\ref{thm:reach-Almagor} and Lemma~\ref{lem:implies_is_univ} is always possible.
% \sgcomment{Comments in the remaining part of this section yet to be addressed}
Lemma~\ref{thm:reach-Almagor} and 
Lemma~\ref{lem:implies_is_univ} can only be used to compute \typeone($\setuniv,\{ a\}$) ECs. 
% As we are interested in \typeone($\setuniv,\setuniv$) ECs. 
We have the following lemma to relate \typeone($\setuniv,\{a\}$) ECs and \typeone($\setuniv,\setuniv$) ECs.

%  We show in Section~\ref{sec:complex} how to do so for more than one parity condition. 
\begin{lemma}
\label{lem:split-typeone}
In an EC $C$, for all $s\in C$ we have that $\satisfies{s}{\MDP{\downharpoonright C}} \limwedgeone{\induniv\in\setuniv}\univ(\parobj_{\induniv}) \wedge \limwedgeone{\induniv\in\setuniv} \qual(\event C^{\max}_{\even}(\parobj_{\induniv}))$, iff for all 
% $i\in[1,\abs{\setuniv}]$, and all 
${\induniv}_i\in \setuniv$, 
% we have that $C$ satisfies
and for all 
$s\in C$, we have that $\satisfies{s}{\MDP{\downharpoonright C}} \limwedgeone{\induniv\in\setuniv}\univ(\parobj_{\induniv}) \wedge \qual(\event C^{\max}_{\even}(\parobj_{a_i}))$.
\end{lemma}

% Proof of Lemma~\ref{lem:split-typeone} appears in  Appendix~\ref{app:typeone}.
\begin{proof}
% First, we remark that if $C$ is such that $\forall\, s\in C,\,$ we have $\satisfies{s}{\MDP{\downharpoonright C}} \limwedgeone{\induniv\in\setuniv}\univ(\parobj_{\induniv}) \wedge \limwedgeone{\induniv\in\setuniv} \qual(\event C^{\max}_{\even}(\parobj_{\induniv}))$, then for  all $i\in[1,\abs{\setuniv}]$, and all $a_i\in A$, we have that $C$ satisfies $\forall\, s\in C,\, \satisfies{s}{\MDP{\downharpoonright C}} \limwedgeone{\induniv\in\setuniv}\univ(\parobj_{\induniv}) \wedge \qual(\event C^{\max}_{\even}(\parobj_{a_i}))$. We now show that the converse is also true: indeed, 
We prove here the right to left implication as the other direction is obvious.
Assume that for $i\in[\abs{\setuniv}]$, we have a strategy $\sigma_{a_i,s}$ for $\limwedgeone{\induniv\in\setuniv}\univ(\parobj_{\induniv}) \wedge \qual(\event C^{\max}_{\even}(\parobj_{a_i}))$.
We define the following strategy $\sigma_{s}$:

\begin{enumerate}
    \item Let $s_0=s$
    % \item Play $\sigma_{a_1,s}$ until reaching $C^{\max}_{\even}(\parobj_{a_{1}})$
    \item For $i$ going from $1$ to $\abs{\setuniv}$, play $\sigma_{a_{i},s_{i-1}}$ until reaching some state $s_i$ of $C^{\max}_{\even}(\parobj_{a_i})$
\end{enumerate}

As in all possible paths we end up having some $i$ such that $\sigma_{a_{i},s_{i-1}}$ is followed forever, and all these strategy satisfy $ \limwedgeone{\induniv\in\setuniv}\univ(\parobj_{\induniv})$, we have that $\sigma_s$ also satisfies $ \limwedgeone{\induniv\in\setuniv}\univ(\parobj_{\induniv})$. As every $\sigma_{a_{i},s_{i-1}}$ has probability $1$ of reaching $C^{\max}_{\even}(\parobj_{a_i})$, the strategy $\sigma_s$ satisfies $\limwedgeone{\induniv\in\setuniv} \qual(\event C^{\max}_{\even}(\parobj_{\induniv}))$. As a conclusion $\satisfy{s}{\sigma_{s}}{\MDP{\downharpoonright C}} \limwedgeone{\induniv\in\setuniv}\univ(\parobj_{\induniv}) \wedge \limwedgeone{\induniv\in\setuniv} \qual(\event C^{\max}_{\even}(\parobj_{\induniv}))$ hence the result holds.
%\qed
\end{proof}
The following lemma states that we can compute the maximal  \typeone($\setuniv,\setuniv$) ECs.
The remaining part of this section is devoted to the proof of the following lemma.
The proof of this lemma involves a detailed algorithmic procedure for computing the set of maximal \typeone($\setuniv,\setuniv$) ECs. 
% We compute the maximal \typeone($\setuniv,\setuniv$) ECs as follows. 
In this procedure we iteratively compute the maximal \typeone($\setuniv,\{\ \induniv_i\}$) ECs for all $\induniv_i\in\setuniv$. The combination of Lemma~\ref{thm:reach-Almagor} and Lemma~\ref{lem:implies_is_univ} is used for the computation of the set maximal \typeone($\setuniv,\{ \induniv_i\}$) ECs. Every time we do this computation, we prune all the states that do not belong to at least one of these ECs and solve Streett games again. We note that computing the maximal \typeone($\setuniv,\{ \induniv_i\}$) ECs follows an approach similar to the the procedure in~\cite{AKV16} that computes the set of maximal SGECs. The difference is that we add an additional step in our algorithm, and use Lemma~\ref{lem:split-typeone} to be able to combine  the different $\{\parobj_{\induniv_i}\}$. We note that a naive generalization of the algorithm in~\cite{AKV16} to compute the set of maximal \typeone($\setuniv,\setuniv$) ECs results in an $\sf{EXPTIME}$ complexity
% \footnote{We give an intuition of the naive $\sf{EXPTIME}$ algorithm in Appendix~\ref{app:typeone}.}
, while we end up with a $\sf{P}^{\sf{NP}}$ complexity as we show later in Section~\ref{sec:complex}.
% The proof of Lemmas~\ref{lem:sgecunion2} and \ref{lem:ec-is-sgec} can be found in Appendix~\ref{app:typeone}.
% We state below the correctness of Algorithm~\ref{proc:sgec}.
\begin{lemma}
\label{lem:alg-typeone}
Given an MDP $\MDP=( S,E,Act,\prob)$, it is possible to compute the set 
% $\mathcal{C}$ 
of maximal \typeone($\setuniv,\setuniv$) ECs 
% such that for all $s\in C$, we have $\satisfies{s}{\MDP{\downharpoonright C}} \limwedgeone{\induniv\in\setuniv}\univ(\parobj_{\induniv}) \wedge \limwedgeone{\induniv\in\setuniv} \qual(\event C^{\max}_{\even}(\parobj_{\induniv}))$
.
This can be done by solving iteratively a number of Streett games that are polynomial in $\abs{\setuniv}$ and $\abs{S}$.
\end{lemma}

% The proof of Lemma \ref{lem:alg-typeone} and a detailed algorithmic procedure for computing the set of maximal \typeone($\setuniv,\setuniv$) ECs  can be found in Appendix~\ref{app:typeone}. 
% % We compute the maximal \typeone($\setuniv,\setuniv$) ECs as follows. 
% In that procedure we iteratively compute the maximal \typeone($\setuniv,\{\ \induniv_i\}$) ECs for all $\induniv_i\in\setuniv$. The combination of Lemma~\ref{thm:reach-Almagor} and Lemma~\ref{lem:implies_is_univ} is used for the computation of the set maximal \typeone($\setuniv,\{ \induniv_i\}$) ECs. Every time we do this computation, we prune all the states that do not belong to at least one of these ECs and solve Streett games again. We note that computing the maximal \typeone($\setuniv,\{ \induniv_i\}$) ECs follows an approach similar to the the procedure in~\cite{AKV16} that computes the set of maximal SGECs. The difference is that we add an additional step in our algorithm, and use Lemma~\ref{lem:split-typeone} to be able to combine  the different $\{\parobj_{\induniv_i}\}$. We note that a naive generalization of the algorithm in~\cite{AKV16} to compute the set of maximal \typeone($\setuniv,\setuniv$) ECs results in an $\sf{EXPTIME}$ complexity\footnote{We give an intuition of the naive $\sf{EXPTIME}$ algorithm in Appendix~\ref{app:typeone}.}, while we end up with a $\sf{P}^{\sf{NP}}$ complexity as we show later in Section~\ref{sec:complex}.
% The proof is also done in Appendix~\ref{app:typeone}.
Given a set $\mathcal{C}$ of end-components, we denote by $S_{\mathcal{C}}$ the set of states belonging to some end-component in $\mathcal{C}$. Formally: $S_{\mathcal{C}}=\{s\in S,\ \exists C\in\mathcal{C}\pipe s\in C\}$. 
At the end, in 
% the algorithm doing so, that denoted 
Algorithm~\ref{proc:sgec}, we compute the set of maximal \typeone($\setuniv,\setuniv$) ECs $C$\footnote{Note that in~\cite{AKV16} the case where $\setuniv$ is a singleton $\{\parobj\}$ has been solved}.
% generalization of SGEC, set of states satisfying 
% % that is exactly such that, $\forall i\in [0,\abs{\setuniv}]$, and $\forall\, s'\in C$, we have $\satisfies{s'}{\MDP{\downharpoonright C}} \limwedgeone{\induniv\in\setuniv}\univ(\parobj_{\induniv}) \wedge \qual(\event C^{\max}_{\even}(\parobj_{a_i}))$. 
% Given an EC $C$, an algorithm is given in~\cite{AKV16} to check whether $C$ is a \typeone($\{\parobj_{\induniv}\},\{\parobj_{\induniv}\}$) EC. 
% Then in~\cite{AKV16} $\setuniv$ is a singleton $\{\parobj\}$. 
To find these ECs $C$, we look at the possible even parity value that may appear in states of $C^{max}_{even}(\parobj)$, and to do so efficiently we look at all even values that $p$ may take. For each of these values, we solve a parity-B\"uchi game. 
% $\univ(\parobj_{\induniv})\wedge \univ (B_{\induniv})$, where $B_{\induniv}$ is a B\"uchi objective.
% As these games are in $\NPinter$, and a linear number of those are solved, 
These games can be solved by using a linear number of calls to an $\NPinter$ oracle, and hence the problem is in $\sf{P}^{\NPinter} = \NPinter$~\cite{brassard1979note}. 
A naive generalization of~\cite{AKV16} would result in an $\sf{EXPTIME}$ complexity. Indeed, 
if we consider objectives of the form $\satisfies{s}{\MDP{\downharpoonright C}} \limwedgeone{\induniv\in\setuniv}\univ(\parobj_{\induniv}) \wedge \limwedgeone{\induniv\in\setuniv} \qual(\event C^{\max}_{\even}(\parobj_{\induniv}))$, we would consider every tuple of length $\abs{\setuniv}$ of possible even parity values that may appear in states of $C^{max}_{even}(\parobj_{\induniv})$. The number of such tuples is exponential in $\abs{\setuniv}$, which leads to the $\sf{EXPTIME}$ complexity. Our approach is to use Lemma~\ref{lem:split-typeone} to check for every $a_i\in\setuniv$ if $\satisfies{s'}{\MDP{\downharpoonright C}} \limwedgeone{\induniv\in\setuniv}\univ(\parobj_{\induniv})\wedge \qual(\event C^{max}_{even}(\parobj_{a_i}))$. This means we can test separately for every $a_i\in\setuniv$ the possible even parity values of states in $C^{max}_{even}(\parobj_{a_i})$, which only considers a polynomial number of possibilities. We can check this condition by converting it to a game with multiple parity, and a single B\"uchi condition. Since a B\"uchi condition can be regarded as a parity condition, this game can be regarded as a multiple-parity game, and thus can be translated to a Streett game. 
% When compared to the case where $\setuniv$ is singleton, where
% we solve a single parity and one B\"uchi condition as done in \cite{AKV16}, in the general case, we need to solve a game with multiple parity, and a single B\"uchi condition.
% % to multiple parity and one B\"uchi condition is the only change in the procedure, 
We solve a number of Streett games linear in $\abs{S}$ and $\abs{\setuniv}$ to check if $C$ is a \typeone($\setuniv,\{ a_i\}$) EC.
% We first show the modified algorithms first defined in~\cite{AKV16}, 
Towards this, we modify the algorithms defined in~\cite{AKV16}, where the case when $\setuniv$ is singleton has been solved.
We begin with Algorithm~\ref{proc:issg} to check if an EC $C$ of MDP $\MDP$ is a \typeone($\setuniv,\{ a_i\}$).
% for objective $\limwedgeone{\induniv\in\setuniv}\univ(\parobj_{\induniv}) \wedge \qual(\event C^{\max}_{\even}(\parobj_{a_i}))$:

\begin{algorithm}
\caption{}
\label{proc:issg}
\textbf{Input} : An EC $C$ of MDP $\MDP$, parity conditions $\{\parobj_{\induniv}\pipe \induniv\in\setuniv\}$, $a_i\in \setuniv$ \\
\textbf{Output} : Yes if and only if for all $s\in C$ we have $\satisfies{s}{\MDP}\limwedgeone{\induniv\in\setuniv}\univ(\parobj_{\induniv}) \wedge \qual(\event C^{\max}_{\even}(\parobj_{a_i}))$ \\
\begin{algorithmic}[1]
    \State Compute $G_C$ the Streett game associated to $\limwedgeone{\induniv\in\setuniv}\univ(\parobj_{\induniv}) \wedge \qual(\event C^{\max}_{\even}(\parobj_{a_i}))$.
    \If {all states win in $G_C$}
    \State return ``yes".
    \Else 
    \State return ``no" and the set of winning states.
    \EndIf
\end{algorithmic}
\end{algorithm}

Now, given an MDP $\MDP$, an objective $\limwedgeone{\induniv\in\setuniv}\univ(\parobj_{\induniv}) \wedge \\ \qual(\event C^{\max}_{\even}(\parobj_{a_i}))$ and an odd number $k$, Algorithm~\ref{proc:maxsgeck} computes the set $\mathcal{MEC}_{\rm I}^k$ of maximal \typeone($\setuniv,\{ a_i\}$) ECs of $\MDP$ whose maximum odd value for parity $\parobj_{a_i}$ is $k$. 
Similarly to~\cite{AKV16}, we focus on maximum odd ranks because we remove the attractor of states that have an odd rank greater than $k$. 
If we tried to obtain the maximum even rank, we would look for a decomposition in EC, such that every EC in this decomposition may have a different maximum odd rank. 
We proceed by removing the attractor of states having an odd rank greater than this maximum. 
This is not possible without already having the decomposition, and so we would end up guessing the decomposition and the maximum odd value in every EC; this would be less efficient.

\begin{algorithm}
\caption{}
\label{proc:maxsgeck}
\textbf{Input} : An MDP $\MDP$, parity conditions $\{\parobj_{\induniv}\pipe \induniv\in\setuniv\}$, $a_i\in \setuniv$, $k\in odd$. \\
\textbf{Output} : The set  $\mathcal{MEC}_{\rm I}^k$ of maximal \typeone($\setuniv,\{ a_i\}$) ECs with maximum odd parity value $k$ \\
\begin{algorithmic}[1]
    \State\label{step:decomp} Compute $\mathcal{C}$ the maximal EC decomposition of $\MDP$.
    \For{\texttt{all $C\in \mathcal{C}$}}
    % \State For all $C\in \mathcal{C}$
    % \State $odd_{>k}(C):= \{s\in C\ |\ \parobj_{a_i}(s)>k \wedge \parobj_{a_i}(s)\in odd\}$ 
    \If {$odd_{>k}(C)\neq \varnothing$} 
    
    \Comment{$odd_{>k}(C):= \{s\in C\ |\ \parobj_{a_i}(s)>k \wedge \parobj_{a_i}(s)\in odd\}$}
    % \State $\MDP := \MDP\backslash attr_{env}(odd_{>k}(C))$, and begin again from Step~\ref{step:decomp}. 
    \State $\MDP := Attr_{1}(odd_{>k}(C))$, and begin again from Step~\ref{step:decomp}. 
    \EndIf
    \If {$C$ is not \typeone} \Comment{Call to Algorithm \ref{proc:issg}}
    \State $\MDP:= Attr_{1}(S\backslash \mathcal{MEC}_{\rm I}^k)$, and begin again from Step~\ref{step:decomp} \Comment{$\mathcal{MEC}_{\rm I}^k$ is the set of winning states in Algorithm \ref{proc:issg}}.
    % \Else 
    % \State $\MDP:=\MDP\backslash Attr_{env}(S\backslash \vgec^k)$, and begin again from
    \EndIf
    \EndFor
\State return $\mathcal{C}$.
\end{algorithmic}
\end{algorithm}

% Now, given all these maximal SGEC for maximum odd value $k$ (denoted $maxSGEC_k$), we get the maximal SGEC by taking the maximal elements of all these SGEC for value $k$, this is done in Algorithm~\ref{proc:maxsgec}.
We get the set of maximal \typeone($\setuniv,\{ a_i\}$) ECs by choosing the maximal elements from the set $\limcup{k\in odd}{}\mathcal{MEC}_{\rm I}^k$.
This is done in Algorithm \ref{proc:maxsgec}.
Note that it is possible that for some odd $k$, for $C \in \mathcal{MEC}_{\rm I}^k$, there may exist $C' \in \mathcal{MEC}_{\rm I}^{\ell}$, where $\ell$ is odd, and $\ell < k$ such that $C'\subsetneq C$.
Hence we take the maximal elements of the union of $\mathcal{MEC}_{\rm I}^k$ over all $k$ to obtain the set of maximal \typeone($\setuniv,\{ a_i\}$) ECs.

\begin{algorithm}
\caption{}
\label{proc:maxsgec}
\textbf{Input} : An MDP $\MDP$, parity conditions $\{\parobj_{\induniv}\pipe \induniv\in\setuniv\}$, $a_i\in \setuniv$. \\
\textbf{Output} : The set of maximal \typeone($\setuniv,\{ a_i\}$) ECs \\
\begin{algorithmic}[1]
    \State Return the set of maximal elements (for the inclusion) of the set $\limcup{k\in odd}{}\mathcal{MEC}_{\rm I}^k$.
\end{algorithmic}
\end{algorithm}
% \sgcomment{Add intuition to why we take maximal after taking union.}

For an MDP $\MDP=(S,E,Act,\prob)$, we now write Algorithm~\ref{proc:sgec} that computes the maximal \typeone($\setuniv,\setuniv$) ECs.

\begin{algorithm}
\caption{}
\label{proc:sgec}
\textbf{Input} : An MDP $\MDP$, parity conditions $\{\parobj_{\induniv}\pipe \induniv\in\setuniv\}$. \\
\textbf{Output} : The set $\mathcal{C}$ of maximal \typeone($\setuniv,\setuniv$) ECs \\
\begin{algorithmic}[1]
    \State Initialize $\mathcal{M}=\MDP$
    \State\label{step:decompo} For $i$ going from $1$ to $\abs{\setuniv}$, compute the decomposition into maximal \typeone($\setuniv,\{ a_i\}$) EC of the MDP associated to $S_{\mathcal{C}}$. We call this decomposition $\mathcal{C}_i$ which is the output of Algorithm \ref{proc:maxsgec}. If $S_{\mathcal{C}_i}\subsetneq S_{\mathcal{C}_1}$, take $\mathcal{M}=(S_{\mathcal{C}_i},E,Act,\prob)$ and begin again from $i=1$
    \Comment{Note that if we reach $\mathcal{C}_i$, we have that $\mathcal{C}_1=\mathcal{C}_2=\ldots =\mathcal{C}_{i-1}$}
    \State \label{step:return}Return $\mathcal{C}=\mathcal{C}_1$.
\end{algorithmic}
\end{algorithm}

% \sgcomment{We can move Lemmas \ref{lem:sgecunion2} and \ref{lem:ec-is-sgec} here.}
    % \item Compute the maximal SGEC \ decomposition of $\mathcal{M}$ for $a_1$, denoted $\mathcal{C}_1$.

We use the following lemmas to show the correctness of Algorithm \ref{proc:sgec}.
\begin{lemma}
\label{lem:sgecunion2}
If $C_1$ and $C_2$ are two \typeone ECs, and $C_1\cup C_2$ is an EC, then $C_1\cup C_2$ is a \typeone EC.
\end{lemma}

\begin{proof}
Observe that we have strategy $\sigma_i$ in $C_i$ such that for all $s\in C_i$, we have that $\satisfy{s}{\sigma_i}{C_{3-i}}\limwedgeone{\induniv\in\setuniv}\univ(\parobj_{\induniv}) \wedge \qual(\event (C_1\cup C_2)^{\max}_{\even}(\parobj)$. If ${C_1}^{\max}_{\even}(\parobj)={C_2}^{\max}_{\even}(\parobj)$ then strategy $\sigma$ playing as $\sigma_i$ in $C_i$ is a witness that $C_1\cup C_2$ is an \typeone EC.
Now w.l.o.g., we assume $(C_1\cup C_2)^{\max}_{\even}(\parobj)={C_1}^{\max}_{\even}(\parobj)>{C_2}^{\max}_{\even}(\parobj)$. We define strategy $\sigma$ as follows: if the initial state belongs to $C_1$ it plays as $\sigma_1$. Otherwise, if the initial state is in $C_2$, it is defined as follows: 
\begin{enumerate}
    \item Play for $\abs{C_2}$ steps uniformly at random while staying in $C_1\cup C_2$.
    \item If a state $q\in C_1$ is reached, then play as $\sigma_1$ forever.
    \item Else play as $\sigma_2$ until reaching ${C_2}_{\even}^{\max}(\parobj)$ and start again from 1.
\end{enumerate}
Either this strategy reaches $C_1$ at some point, and then it satisfies both $\limwedgeone{\induniv\in\setuniv}\univ(\parobj_{\induniv})$ and $\qual(\event (C_1\cup C_2)^{\max}_{\even}(\parobj)$. This happens with probability one.
In the remaining cases, we stay in $C_2$ forever. This means that it will either play $\sigma_2$ forever at some point, satisfying $\univ(\parobj)$, or alternate between bounded length random walks and parts where it follows $\sigma_2$ until reaching ${C_2}_{\even}^{\max}(\parobj)$ infinitely often, also satisfying $\univ(\parobj)$.
%\qed
\end{proof}

% \begin{lemma}
% \label{lem:ec-is-sgec}
% Given an MDP $\MDP=(S,E,Act,\prob)$, and two sets of parity conditions $\{\parobj_{\induniv}\pipe \induniv\in\setuniv\}$ and $\{\parobj_{\indqual}\pipe \indqual\in\setqual\}$, if for all $s\in S$, there exists $T\in \sgec(\MDP,\setuniv,\setqual)$, such that $s\in T$, then the set of maximal $\typeone$ ECs of $\MDP$ is the set of maximal ECs of $\MDP$.
% \end{lemma}

% \begin{proof}
% % We want to prove that if every state $s\in S$ belongs to some \typeone EC of $\MDP$, then the set of maximal \typeone EC of $\MDP$ is the same as the set of maximal EC of $\MDP$. 
% Let $C$ be a maximal EC. 
% Now, assume by way of contradiction that we have a \typeone EC $S$, and two distinct states $q,q'$ such that $q\in S$ and $q\in C$, and  
% %such that some state $q\in S$ and $q\in C$ and some other state 
% $q'\in S$ but $q'\nin C$. Then $C\cup S$ is also an EC, and bigger than $C$, so $C$ wasn't maximal, which is contradictory. This means there exists disjoint \typeone ECs $C_1,\dots C_k$ such that $C=\limcup{i\in [k]}{} C_i$.
% By Lemma~\ref{lem:sgecunion2} we have that the end-component $\limcup{i\in [k]}{} C_i$ is a \typeone EC, implying that the EC $C$ is also a \typeone EC.
% %\qed
% \end{proof}

We now have the tools to prove Lemma~\ref{lem:alg-typeone}.

\begin{proof}We denote by $\mathcal{C}^*$ the correct decomposition we want to compute, and by $\mathcal{C}$ the decomposition obtained by applying Algorithm~\ref{proc:sgec}. We show equality between $\mathcal{C}^*$ and $\mathcal{C}$ by proving $\mathcal{C}\subseteq \mathcal{C}^*$ and $\mathcal{C}^*\subseteq \mathcal{C}$.

Consider some $C \in \mathcal{C}^*$.
Since all states in $C$ satisfy $\limwedgeone{\induniv\in\setuniv}\univ(\parobj_{\induniv}) \\ \wedge \qual(\event C^{\max}_{\even}(\parobj_{a_i}))$ for all $a_i \in \setuniv$, none of these states gets removed by Algorithm \ref{proc:sgec}.
Also Algorithm \ref{proc:sgec} computes the maximal ECs such that all states inside each such maximal EC satisfies the property.
Thus $C \in \mathcal{C}$.
Hence $\mathcal{C}^* \subseteq C$.
% Let's take $\mathcal{C}$ a decomposition into maximal end-components ensuring the conditions associated to all $a_i$, and $C\in\mathcal{C}$. Then for all $a_i$, we have that all states of $C$ belong to at least one end-component (namely $C$), and so shouldn't be removed by the computation of maximal end-components for $a_i$. As a consequence, no state of $C$ will be removed, and as a consequence, $\mathcal{C}^*\subseteq \mathcal{C}$.

% \sgcomment{This direction of the proof also requires to be changed.}
Now we prove that for all $C\in \mathcal{C}$, there exists $C^*\in\mathcal{C}^*$ such that $C=C^*$. Let us denote by $s_1\ldots s_n$ the states of $C$. For all $i\in[n]$, since $s_i$ belongs the \typeone EC $C$, it also belongs to some maximal \typeone EC $C_i^*$. 
It is easy to see that for all $i,j\in[n]$, we have that  $C_i^*\cup C_j^*\cup C$ is an EC, 
% To do so, for all $i,j\in[n]$, for all $\widehat{s}\in C_i^*$ and $s_{term}\in C_j^*$, we will prove that it is possible to reach $s_{term}$ from $\widehat{s}$ with probability $1$ while staying in $C_i^*\cup C_j^*\cup C$. Indeed, as $\widehat{s}$ and $s_i$ belong to a common EC $C_i^*$, it is possible to reach $s_i$ from $\widehat{s}$ with probability $1$. Then, as $s_{i}$ and $s_j$ belong to a common EC $C$, it is possible to reach $s_j$ from $s_{i}$ with probability $1$. Finally, as $s_{j}$ and $s_{term}$ belong to a common EC $C_j^*$, it is possible to reach $s_{term}$ from $s_{j}$ with probability $1$. As a consequence it is possible to reach $s_{term}$ from $\widehat{s}$ with probability $1$. Thus, from every state of $C_i^*\cup C_j^*\cup C$, it is possible to reach every states of $C_i^*\cup C_j^*\cup C$, and thus $C_i^*\cup C_j^*\cup C$ is an EC. 
and in general that $C\cup\bigcup_{i\in[n]} C_i^*$ is an EC. By Lemma~\ref{lem:sgecunion2} we have that $C\cup\bigcup_{i\in[0,n]} C_i^*$ is a \typeone, and by maximality of the $C_i^*$ and of $C$ we have that $C=C_1^*=\ldots=C_n^*=C\cup\bigcup_{i\in[n]} C_i^*$ and thus there exists $C^*\in \mathcal{C}^*$ such that $C=C^*$.
% Now we prove that all end-components in this decomposition $\mathcal{C}$ ensure all of the conditions $a_i$. Consider an MDP $\MDP=(S,E,Act,\prob)$ such that every state $s\in S$ belongs to some SGEC of $\MDP$. By Lemma~\ref{lem:ec-is-sgec} we have that all EC are also SGEC. 
% As a consequence, for all EC $\mathcal{C}$ and for all $a_i$ the decomposition into maximal EC would be exactly $\mathcal{C}$: in every $C\in \mathcal{C}$, all conditions $a_i$ hold, and we have that $\mathcal{C}\subseteq \mathcal{C}^*$.

For the number of Streett games solved, we look at the steps of Algorithm~\ref{proc:sgec}.
Algorithm \ref{proc:issg} solves a single Streett game.
Thus we compute the number of calls to Algorithm~\ref{proc:issg}.
% We want to compute the set of $\typeone$ states.
From Algorithm ~\ref{proc:sgec}, we call Algorithm~\ref{proc:maxsgec} for $O(\abs{\setuniv} \cdot \abs{S})$ times.
For a fixed MDP $\mathcal{M}$, we call Algorithm~\ref{proc:maxsgec} for $\abs{\setuniv}$ times, and we change $\mathcal{M}$ for $O(\abs{S})$ times, every time a state is removed.
From Algorithm~\ref{proc:maxsgec}, for a given $a_i$, we call Algorithm~\ref{proc:maxsgeck} for $\min(k, S)$ times where $k$ is the maximum odd value less than or equal to $\max_{s \in S}\parobj_{a_i}(s)$.
Note that the number of odd values of parity $\parobj_{a_i}$ is no more than the number of states in the MDP.
From Algorithm~\ref{proc:maxsgeck}, every time a state is removed, we call Algorithm \ref{proc:issg}.
Thus the number of Streett games solved is $O(\abs{\setuniv} \cdot \abs{S}^3)$.
% Looking at the complexity of Algorithm~\ref{proc:sgec}, we have that for all $i$ Step~\ref{step:decompo} solves number of Streett games linear in $\abs{S}$ and the maximum $\abs{a_i}$ . This step is executed at most $\abs{\setuniv}$ before repeating Step~\ref{step:decompo} from the beginning  or going to Step~\ref{step:return}. As we only restart Step~\ref{step:decompo} if we have removed a state, we loop at most $\abs{S}$ times before going to Step~\ref{step:return} and stopping. As a consequence the algorithm solves a number of Streett games polynomial in $\abs{\setuniv}$, the maximum $\abs{a_i}$, and $\abs{S}$.
%\qed
\end{proof}

\section{\typetwo end-components}
\label{sec:typetwo}
In this section, we define \typetwo ECs that are a generalization of UGEC defined in \cite{BRR17}. In the setting of the current paper, they are a generalization of \typeone ECs, that have an additional condition.
% \rbchanged{, and we give an algorithm to compute the maximal \typetwo ECs in Lemma~\ref{lem:alg-typeone}.}
The main result of the section, Lemma~\ref{lem:SAS}, shows an equivalence between solving the realizability problem for formulas of the form $\limwedgeone{\induniv\in{\setuniv}}\univ(\parobj_{\induniv})\wedge \limwedgeone{\indqual\in{\setqual}}\qual(\parobj_{\indqual})$ and solving the realizability problem for formulas involving sure parity conditions and almost-sure reachability of the {\sf Type\:{\rm II}} end-components.
To do so, the two directions of the proof are done separately. 
For the right to left direction, we are given a strategy $\sigma_T$. First Lemma~\ref{lem:suf} shows that inside a \typetwo EC, there always exists a strategy $\sigma$ for $\limwedgeone{\induniv\in{\setuniv}}\univ(\parobj_{\induniv})\wedge \limwedgeone{\indqual\in{\setqual}}\qual(\parobj_{\indqual})$.
We rely on Lemmas~\ref{lemma_stratTwo1}, and Lemma~\ref{lemma_stratTwo2} to prove that such a $\sigma$ can be constructed, and then we use Lemma~\ref{lemStratTwo3} to show that $\sigma$ satisfies $\limwedgeone{\induniv\in{\setuniv}}\univ(\parobj_{\induniv})\wedge \limwedgeone{\indqual\in{\setqual}}\qual(\parobj_{\indqual})$.
In Lemma~\ref{lem:LeftRight}, we use the strategy $\sigma_T$ to reach a \typetwo EC and thereafter play $\sigma$, completing the proof of this direction.

% For the more general case, when we are given a strategy to reach such an EC under the conditions mentioned above, we give in Definition~\ref{def:global} a strategy for  $\limwedgeone{\induniv\in{\setuniv}}\univ(\parobj_{\induniv})\wedge \limwedgeone{\indqual\in{\setqual}}\qual(\parobj_{\indqual})$, concluding this direction of the proof.
Lemma~\ref{lemma_rightleft} states the result in the other direction. 
Towards this, we first show that for a strategy satisfying the formula $\limwedgeone{\induniv\in{\setuniv}}\univ(\parobj_{\induniv})\wedge \limwedgeone{\indqual\in{\setqual}}\qual(\parobj_{\indqual})$, all the states that are visited under this strategy satisfy this particular formula.
We then introduce the notion of density in Definition~\ref{def:dens} to relate in Lemma~\ref{lem:bp} the states satisfying the formula $\limwedgeone{\induniv\in{\setuniv}}\univ(\parobj_{\induniv})\wedge \limwedgeone{\indqual\in{\setqual}}\qual(\parobj_{\indqual})$ to the Street-B\"uchi game of Section~\ref{sec:typeone}. 
Since  Street-B\"uchi games are related to \typeone ECs, and \typetwo ECs are extensions of \typeone ECs, it then remains to prove that there exists at least one \typetwo EC in the MDP.
% when $\limwedgeone{\induniv\in{\setuniv}}\univ(\parobj_{\induniv})\wedge \limwedgeone{\indqual\in{\setqual}}\qual(\parobj_{\indqual})$ is satisfied, which 
This is done in Lemma~\ref{lem:ugec}.
% by looking at the set of states visited by strategies visiting a minimal number of states, and concluding this side of the proof.

Given two sets of parity conditions $\{\parobj_{\induniv}\pipe \induniv\in\setuniv\}$ and $\{\parobj_{\indqual}\pipe \indqual\in\setqual\}$, an end-component $C$ of $\MDP$ is \typetwo($\setuniv,\setqual$) 
if the following two properties hold:

\begin{itemize}
	\item $\mathbf{(II_1)}$ $\forall\, s\in C,\, \satisfies{s}{\MDP{\downharpoonright C}} \limwedgeone{\induniv\in\setuniv}\univ(\parobj_{\induniv}) \wedge \limwedgeone{\induniv\in\setuniv} \qual(\event C^{\max}_{\even}(\parobj_{\induniv}))$
	\item $\mathbf{(II_2)}$ $\forall\, s\in C,\, \satisfies{s}{\MDP{\downharpoonright C}} \limwedgeone{\induniv\in\setuniv}\qual(\parobj_{\induniv})\wedge \limwedgeone{\indqual\in\setqual}\qual(\parobj_{\indqual})$
\end{itemize}

We note that condition $\mathbf{(II_1)}$ is exactly the one defining a \typeone($\setuniv,\setuniv$) EC.
We write \typetwo($\setuniv,\setqual$) EC as \typetwo EC when the parity sets are clear from the context.
We introduce the following notations: $\ugec(\MDP,\setuniv,\setqual)$ is the set of all \typetwo($\setuniv,\setqual$) ECs, and $\mathcal{T}_{\mathbf{II},\MDP,\setuniv,\setqual} = {\displaystyle \cup_{U \in \ugec(\MDP,\setuniv,\setqual)} U}$ is the set of states belonging to some \typetwo EC. 
% For a \QPL formula $\varphi$, we denote by $\sigma_{\varphi}$ a strategy satisfying this formula.\rbchanged{Do we use this?}

% \sgcomment{To address comments until following definition.}
Intuitively, a \typetwo($\setuniv,\setqual$) EC, is a \typeone($\setuniv,\setuniv$) EC
%$\playerOne$ has 
% there is a strategy to visit all $C^{\max}_{\even}(\parobj_{\induniv})$ with probability $1$ while guaranteeing all $\univ(\parobj_{\induniv})$, and 
%he also has 
where there also exists an additional strategy staying within the EC and almost-surely satisfying all parity conditions $\parobj_{\induniv}$ and $\parobj_{\indqual}$.
% We note that both properties must hold while staying in the end-component $C$.
%Figure~\ref{fig:full} gives an example of UGEC.
This notion generalizes the notion of a \textit{ultra-good end-component} (UGEC in \cite{BRR17}) which is a \typetwo EC where both $\setuniv$ and $\setqual$ are singletons.
% Let us note that every \typetwo EC is not a \typeone EC.
In the sequel, we use the following notation: $\combi{\setuniv}{\setqual}=\limwedgeone{\induniv\in{\setuniv}}\univ(\parobj_{\induniv})\wedge \limwedgeone{\indqual\in{\setqual}}\qual(\parobj_{\indqual})$.

Finding solutions for $\mathbf{(II_2)}$ is done in~\cite{etessami2007multi}, and it is shown in~\cite{BRR17} how to use simple techniques from~\cite{BK08}. Winning strategies for $\mathbf{(II_2)}$ may require either randomization or deterministic finite memory. We showed how to compute ECs satisfying $\mathbf{(II_1)}$ (\typeone ECs) in Lemma~\ref{lem:alg-typeone}.
% , and how to get deterministic finite memory winning strategies. 
In the sequel we relate \typetwo ECs to the formula $\combi{\setuniv}{\setqual}$. In particular, from every state belonging to a \typetwo ECs there exists a strategy satisfying the formula $\combi{\setuniv}{\setqual}$. We illustrate this by the following example.

\begin{example}
\label{ex:typetwo}
Consider the example in Figure~\ref{fig:UGEC}. We show a strategy satisfying $\univ(\parobj_1) \wedge \qual(\parobj_2)$ in an MDP that is also a \typetwo EC. 
% We restrict our attention to the end-component inside the dashed box that is a \typetwo EC.
Indeed every state satisfies condition  $\mathbf{(II_1)}$ when action $a$ is chosen from state $s$ and every state satisfies condition $\mathbf{(II_2)}$ when action $b$ is chosen from state $s$.
% \begin{figure}[t]
% \centering
% \vspace{-15pt}
% \hspace{-15pt}
% \centering
% \scalebox{0.9}{
% 	\begin{tikzpicture}
% 		\node[player,initial,initial text={}] (sinit) at (2,2) {0,1} ;
% 		\node[player] (s1) at (2,0) {0,1} ;
		
% 		\node[player] (s2) at (4,0) {2,2} ;
% 		\node[player] (s3) at (4,2) {1,1} ;
% 		\node[player] (s4) at (6,2) {0,0} ;
% 		\node[player] (s5) at (8,2) {0,5} ;
% 		\node[player] (s6) at (8,0) {2,1} ;

% 		\node[player] (s7) at (6,4) {3,3} ;
% 		\node[player] (s8) at (8,4) {4,4} ;

% 		\node () at (6,1.3) {\large {\color{blue}$s$}} ;
% 	    \node[subgraph] (Nc) at (5.75,1) {};

%  \path[-latex]  (sinit) edge[bend left=20] node[right] {\small $a,1$} (s1)
% 				(s1) edge[bend left=20] node[left] {\small $a,0.5$} (sinit)				
% 				(s1) edge node[above] {\small $a,0.5$} (s2)
% 				(s2) edge[bend right=20] node[right] {\small $a,1$} (s4)
% 				(s4) edge[bend right=20] node[above] {\small $b,1$} (s3)
% 				(s3) edge node[left] {\small $a,0.5$} (s2)
% 				(s3) edge[bend right=20] node[below] {\small $a,0.5$} (s4)
% 				(s4) edge[bend left=20] node[right, yshift=.1cm] {\small $a,1$} (s5)
% 				(s5) edge[bend left=20] node[below] {\small $a,0.5$} (s4)
% 				(s5) edge node[right] {\small $a,0.5$} (s6)
% 				(s6) edge[bend left=20] node[left] {\small $a,1$} (s4)
				
% 				(s4) edge[bend left=20] node[left] {\small $c,1$} (s7)
% 				(s7) edge[bend left=20] node[right] {\small $a,0.5$} (s4)
% 				(s7) edge node[above] {\small $a,0.5$} (s8)
% 				(s8) edge[bend left=20] node[above,xshift=.6cm] {\small $a,1$} (s4)
				
% 		;
% 	\end{tikzpicture}
% }

\begin{figure}[t]
\centering
\vspace{-15pt}
\hspace{-15pt}
\centering
\scalebox{0.9}{
	\begin{tikzpicture}
% 		\node[player] (sinit) at (2,2) {0,1} ;
% 		\node[player] (s1) at (2,0) {0,1} ;
		
		\node[player,initial,initial text={}] (s2) at (4,0) {2,2} ;
		\node[player] (s3) at (4,2) {1,1} ;
		\node[player] (s4) at (6,2) {0,0} ;
		\node[player] (s5) at (8,2) {0,5} ;
		\node[player] (s6) at (8,0) {2,1} ;

% 		\node[player] (s7) at (6,4) {3,3} ;
% 		\node[player] (s8) at (8,4) {4,4} ;

		\node () at (6,1.3) {\large {\color{blue}$s$}} ;
	   % \node[subgraph] (Nc) at (5.75,1) {};

 \path[-latex]  
    %             (sinit) edge[bend left=20] node[right] {\small $a,1$} (s1)
				% (s1) edge[bend left=20] node[left] {\small $a,0.5$} (sinit)				
				% (s1) edge node[above] {\small $a,0.5$} (s2)
				(s2) edge[bend right=20] node[right] {\small $a,1$} (s4)
				(s4) edge[bend right=20] node[above] {\small $b,1$} (s3)
				(s3) edge node[left] {\small $a,0.5$} (s2)
				(s3) edge[bend right=20] node[below] {\small $a,0.5$} (s4)
				(s4) edge[bend left=20] node[above] {\small $a,1$} (s5)
				(s5) edge[bend left=20] node[below] {\small $a,0.5$} (s4)
				(s5) edge node[right] {\small $a,0.5$} (s6)
				(s6) edge[bend left=20] node[left] {\small $a,1$} (s4)
				
				% (s4) edge[bend left=20] node[left] {\small $c,1$} (s7)
				% (s7) edge[bend left=20] node[right] {\small $a,0.5$} (s4)
				% (s7) edge node[above] {\small $a,0.5$} (s8)
				% (s8) edge[bend left=20] node[above,xshift=.6cm] {\small $a,1$} (s4)
				
		;
	\end{tikzpicture}
}
\caption{\label{fig:UGEC}An example of a \typetwo EC}
\end{figure}
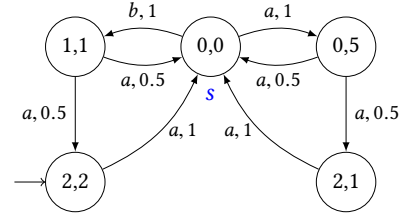
Now the strategy from $s$ to satisfy $\univ(\parobj_1) \wedge \qual(\parobj_2)$ is the following.
In \emph{round $1$}, action $b$ is chosen some $i_0$ times.
If the state with parity ($2,2$) is not visited when action $b$ is chosen $i_0$ times, then action $a$ is chosen until the state with parity ($2,1$) is visited.
Once state ($2,1$) is visited, we proceed to round $2$ in which action $b$ is chosen $i_1=i_0+1$ times.
% In round $2$, if the state with parity ($2,2$) is not visited when action $b$ is chosen $i_0+1$ times, then action $a$ is chosen until the state with parity ($2,1$) is visited.
From every round, we proceed to the next round if either the state with parity ($2,2$) is visited in the current round, or otherwise if the state with parity ($2,1$) is visited  and so on, resulting in action $b$ being chosen $i_j=i_0+j$ times at round $j$.
Now we compute the probability of not choosing action $a$ from $s$ during $n$ rounds.
The probability of not choosing $a$ from $s$ after the first round is $1-2^{-i_0}$.
The probability of not choosing $a$ from $s$ after the first and the second round is $(1-{2^{-i_0}}) \cdot (1-2^{-(i_0+1)})$,
and thus the probability of never choosing $a$ when $n$ rounds already happened is $p(n)=\prod_{j=n}^{\infty}(1-2^{-(i_0+j)})$.
% not choosing $a$ for $k$ rounds when $n$ rounds have already been played is $p(n,k)=\prod_{j=n}^{k+n}(1-2^{-(i_0+j)})$.
% Let $p(n) = \lim_{k \rightarrow \infty} p(k,n)$.
In~\cite{BRR17}, it has been shown that $\lim_{n \rightarrow \infty}p(n)=1$, implying that with probability $1$ action $a$ will eventually stop being played.
% \sgcomment{We are saying that the probability of never taking $a$ after $n$ rounds is $1$ for $\lim{n \to \infty}$. How does it say that the probability of never taking $a$ is $1$ when the game is played from the beginning?}
% for $\lim_{i \rightarrow \infty}$.
\end{example}

The strategies for conditions $\mathbf{(II_1)}$ and $\mathbf{(II_2)}$ in the example above are deterministic memoryless strategies.
However, in general, the strategy for $\mathbf{(II_1)}$ may require finite memory and the strategy for $\mathbf{(II_2)}$ may require memory or randomization.
% \sgcomment{Should the above be stated in a lemma?}
We illustrate this below. 
% examples in Appendix~\ref{app:typetwo}.
% We illustrate this with the following examples.
In Figure \ref{fig:1U}, note that starting from the state with parity ($1,1$), we need memory to satisfy condition $\mathbf{(II_1)}$.
A strategy from the state with parity ($1,1$) that alternates between actions $b$ and $c$ satisfies condition $\mathbf{(II_1)}$.

% \subsection{Examples of strategies for conditions $\mathbf{(II_1)}$ and $\mathbf{(II_2)}$.}
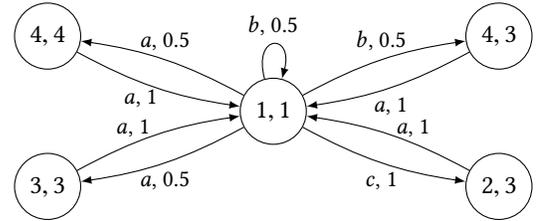
\begin{figure}[t]
\centering
\vspace{5pt}
\scalebox{1}{
	\begin{tikzpicture}
		\node[player] (s) at (0,0) {$1,1$} ;
		\node[player] (s1) at (-3,1) {$4,4$} ;
		\node[player] (s2) at (-3,-1) {$3,3$} ;
		\node[player] (s3) at (3,1) {$4,3$} ;
		\node[player] (s4) at (3,-1) {$2,3$} ;
		
		\path[-latex]  
			(s) edge[bend right=10] node[above] {\small $a$, 0.5} (s1)
			(s) edge[bend left=10] node[below] {\small $a$, 0.5} (s2)
			(s) edge[bend left=10] node[above] {\small $b$, 0.5} (s3)
			(s) edge[bend right=10] node[below] {\small $c$, 1} (s4)
			(s1) edge[bend right=10] node[below, xshift=-.2cm, yshift=.1cm] {\small $a$, 1} (s)
			(s2) edge[bend left=10] node[above, xshift=-.3cm, yshift=-.15cm] {\small $a$, 1} (s)
			(s3) edge[bend left=10] node[below] {\small $a$, 1} (s)
			(s4) edge[bend right=10] node[above, xshift=.3cm, yshift=-.15cm] {\small $a$, 1} (s)
			(s) edge [loop above] node[above] {\small $b$, 0.5} (s)
		;
	\end{tikzpicture}
}
\caption{\label{fig:1U}Finite memory may be needed for condition $\mathbf{(II_1)}$.}
\end{figure}

In Figure \ref{fig:2U}, starting from the state with parity ($0,0$), a randomized memoryless strategy that chooses action $a$ with probability $p>0$ and action $b$ with probability $1-p$ will visit both of the other two states infinitely often almost-surely, and hence a randomized memoryless strategy suffices here.
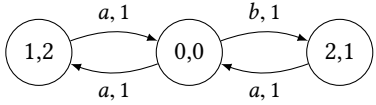
\begin{figure}[t]
\centering
\vspace{-15pt}
\hspace{-15pt}
\centering
\scalebox{1}{
	\begin{tikzpicture}
		\node[player] (sinit) at (0,0) {1,2} ;
		\node[player] (s1) at (2,0) {0,0} ;
		\node[player] (s2) at (4,0) {2,1} ;

 \path[-latex]  (sinit) edge[bend left=20] node[above] {\small $a,1$} (s1)
				(s1) edge[bend left=20] node[below] {\small $a,1$} (sinit)				
				(s1) edge[bend left=20] node[above] {\small $b,1$} (s2)
				(s2) edge[bend left=20] node[below] {\small $a,1$} (s1)				
		;
	\end{tikzpicture}
}
\caption{\label{fig:2U}Randomisation or finite-memory may be needed for condition $\mathbf{(II_2)}$.}
\end{figure}

We now state the main result of this section.

\begin{lemma}
\label{lem:SAS}
Given an MDP $\MDP$, and two sets of parity conditions $\{\parobj_{\induniv}\pipe \induniv\in\setuniv\}$ and $\{\parobj_{\indqual}\pipe \indqual\in\setqual\}$, for all states $s_0$, we have $\satisfies{s_0}{\MDP}\combi{\setuniv}{\setqual}$ iff $\satisfies{s_0}{\MDP}\limwedgeone{\induniv\in\setuniv}\univ(\parobj_{\induniv})\wedge \qual({\event \mathcal{T}_{\mathbf{II},\MDP,\setuniv,\setqual}})$.
% Such strategy may require infinite memory.
A strategy that enforces such conditions may require infinite memory.
\end{lemma}
% It is enough to consider formula of the form  $\combi{\setuniv}{\setqual}$.

We outline a sketch of proof of Lemma~\ref{lem:SAS}. We show each direction separately. We fix some $s_0\in S$ and use it throughout this section. The necessity of using infinite memory is proved in Theorem 18 of~\cite{BRR17}, in the subcase where $\setuniv$ and $\setqual$ are singletons.

We start with the right-to-left implication.
Since $\satisfies{s_0}{\MDP}\limwedgeone{\induniv\in\setuniv}\univ(\parobj_{\induniv})\wedge \qual({\event \mathcal{T}_{\mathbf{II},\MDP,\setuniv,\setqual}})$, there exists a witness strategy $\sigma_T$. 
By definition, paths in $\paths^{\MDPtoMC{\MDP}{\sigma_T}}(s_0)$ eventually reaches some \typetwo EC $C\in\ugec(\MDP,\setuniv,\setqual)$ with probability one. 
Since $C$ is a \typetwo EC, there exist two strategies, $\sigma_1$ and $\sigma_2$ respectively ensuring condition $\mathbf{(II_1)}$ and $\mathbf{(II_2)}$.
In what follows we define a strategy $\sigma_C$ from $\sigma_1$ and $\sigma_2$ such that for all $s\in C$ we have $\satisfy{s}{\sigma_C}{\MDP}\combi{\setuniv}{\setqual}$.
Finally, we compute a strategy $\sigma$ such that from $s_0$, we play $\sigma_T$ until reaching a \typetwo EC $C$, where we play $\sigma_C$.
This strategy $\sigma$ satisfies $\combi{\setuniv}{\setqual}$ from $s$.

We construct the infinite memory strategy $\sigma_C$ from $\sigma_1$ and $\sigma_2$. 
To do so, we first show some technical lemmas ensuring $\sigma_C$ can be effectively computed.

% \subsection{Additional Lemmas}

\begin{lemma}
\label{lemma_stratTwo1}
Given an EC $C$, it can be decided in polynomial time if condition $\mathbf{(II_2)}$ holds in $C$, and a randomized memoryless strategy suffices.
\end{lemma}
\begin{proof}

Given an EC $C$, we claim that (i) the existence of a sub-EC $D$ such that for all $i\in\setuniv\cup\setqual$ we have $D^{\max}_{\even}(\parobj_i) \neq \varnothing$ is not only necessary but also sufficient for $C$ to satisfy condition $\mathbf{(II_2)}$, and (ii) the existence of such a set can be decided in polynomial time.

For (i), it is obvious that the existence of such an EC is necessary. To prove it is also sufficient, we assume we have a sub-EC $D$ such that for all $i\in\setuniv\cup\setqual$ we have $D^{\max}_{\even}(\parobj_i) \neq \varnothing$. For all sub-EC of $C$, in particular for $D$, we can build a uniform randomized memoryless strategy $\sigma$ such that 
\[
\prob_{\MDPtoMC{\MDP_{\downharpoonright C}}{\sigma}}\left[ \{\pi \in \paths^{\MDPtoMC{\MDP_{\downharpoonright C}}{\sigma}}(s) \mid \infny(\pi) = D\}\right] = 1.
\]
Since for all $i\in\setuniv\cup\setqual$ we have $D^{\max}_{\even}(\parobj_i) \neq \varnothing$, we thus have that all $\parobj_{\induniv}$ and $\parobj_{\indqual}$ are almost-surely satisfied by $\sigma$, hence that $\sigma$ is a witness for $\satisfies{s}{\MDP_{\downharpoonright C}} \limwedgeone{\induniv\in\setuniv}\qual(\parobj_{\induniv})\wedge \limwedgeone{\indqual\in\setqual}\qual(\parobj_{\indqual})$, and that condition $\mathbf{(II_2)}$ holds in $C$.

It remains to check the existence of such a set $D \subseteq C$ in polynomial time. First, we check if all $C^{\max}_{\even}(\parobj_{\induniv}) \neq \varnothing$ and $C^{\max}_{\even}(\parobj_{\indqual}) \neq \varnothing$. If this holds, then $D = C$ and the answer is $\yes$ (it takes linear time obviously). If it does not hold, then we compute the sets
$$
	C^{\max}_{\odd}(\parobj_i) = \big\lbrace s\in C \mid (\parobj_i(s) \text{ is odd}) \wedge $$
	$$(\forall\, s'\in C,\, \parobj_i(s') \text{ is even } \implies \parobj_i(s') < \parobj_i(s))\big\rbrace
	$$
and we iterate this procedure in the sub-EC $C' \subset C$ defined as
% \[
% C' = C \setminus \mathtt{Attr}_{Env}\big(\bigcup_{i\in\setuniv\cup\setqual}C^{\max}_{\odd}(\parobj_i)\big).
% \]
\[
C' = \mathtt{Attr}_{1}\big(\bigcup_{i\in\setuniv\cup\setqual}C^{\max}_{\odd}(\parobj_i)\big).
\]
It is easy to see that a suitable $D$ exists if and only if this procedure stops before $C' = \varnothing$. In addition, this procedure takes at most $\vert C \vert$ iterations (as we remove at least one state at each step) and each iteration takes linear time. This implies our result and concludes our proof.
%\qed
\end{proof}

\begin{lemma}
\label{lemma_stratTwo2}
Let $C$ be an EC of $\MDP$. If condition $\mathbf{(II_2)}$ holds then there exists a randomized memoryless witness strategy $\sigma_{2}$ and a sub-EC $D \subseteq C$ such that for all $i\in\setuniv\cup\setqual$ we have $D^{\max}_{\even}(\parobj_{i}) \neq \varnothing$, and for all $s \in C$, we have that
\(
\prob_{\MDPtoMC{\MDP_{\downharpoonright C}}{\sigma_{II_2}}}\left[ \{\pi \in \paths^{\MDPtoMC{\MDP_{\downharpoonright C}}{\sigma_{II_2}}}(s) \mid \infny(\pi) = D\}\right] = 1.
\)
\end{lemma}

\begin{proof}[Proof of Lemma~\ref{lemma_stratTwo2}]
Let $C$ be an EC that satisfies condition $\mathbf{(II_2)}$, and a corresponding witness strategy $\sigma$ (whose existence is granted by Lemma~\ref{lemma_stratTwo1}). By Proposition~\ref{prop:long_run}, we know that for all states $s \in C$,
\[
\prob_{\MDPtoMC{\MDP_{\downharpoonright C}}{\sigma}}\left[ \{\pi \in \paths^{\MDPtoMC{\MDP_{\downharpoonright C}}{\sigma}}(s) \mid \infny(\pi) \in EC(\MDP_{\downharpoonright C})\}\right] = 1,
\]
where $EC(\MDP_{\downharpoonright C})$ is the set of ECs of $\MDP_{\downharpoonright C}$. We claim that for all $D \in \MDP_{\downharpoonright C}$ for which 
\[
\prob_{\MDPtoMC{\MDP_{\downharpoonright C}}{\sigma}}\left[ \{\pi \in \paths^{\MDPtoMC{\MDP_{\downharpoonright C}}{\sigma}} \mid \infny(\pi) = D\}\right] > 0,
\]
% \sgcomment{The above sentence is incomplete.}
and for all $i\in\setuniv\cup\setqual$, we necessarily have that $D^{\max}_{\even}(\parobj_i) \neq \varnothing$. Indeed, assume this is false for $\parobj_i$. Then, with probability strictly greater than zero, $\sigma$ induces a path $\pi$ such that $\max_{s' \in \infny(\pi)} \parobj_i(s')$ is odd (as the maximal priority in $\infny(\pi) = D$ is odd). This contradicts the fact that $\sigma$ is a witness strategy for $\satisfies{s}{\MDP_{\downharpoonright C}} \limwedgeone{\induniv\in\setuniv}\qual(\parobj_{\induniv})\wedge \limwedgeone{\indqual\in\setqual}\qual(\parobj_{\indqual})$.

We know that such a $D$ exists. As every state of an EC can almost-surely be visited using a uniform randomized memoryless strategy~\cite{BK08}, we can conceive a witness strategy $\sigma_{2}$ that reaches and stays only in D with probability $1$.
%\qed
\end{proof}

% We recall the following proposition: 

\begin{proposition}[Optimal reachability~\cite{BK08}]
\label{prop:opti_reach}
Given an MDP $\MDP=( S,E,Act,\prob)$, and a target set $T \subseteq S$, we can compute for each state $s \in S$ the maximal probability $v^{\ast}_{s}$ to reach $T$, in polynomial time. There is an optimal deterministic memoryless strategy $\sigma^{\ast}$ that enforces $v^{\ast}_{s}$ from all $s \in S$.
Now, fix $s \in S$ and $c \in \mathbb{Q}$ such that $c < v^{\ast}_{s}$. Then there exists $k \in \mathbb{N}$ such that by playing $\sigma^{\ast}$ from $s$ for $k$ steps, we reach $T$ with probability larger than $c$.
\end{proposition}

\begin{lemma}\label{lemStratTwo3}
Let $C$ be an EC of\ \ $\MDP$. Every strategy $\sigma_{2}$ satisfies the following property: for all $s \in C$, for all $\varepsilon > 0$, there exists $n \in \mathbb{N}$ such that:

$$
\prob_{\MDPtoMC{\MDP_{\downharpoonright C}}{\sigma_{2}}}[\lbrace \rho \in \fpaths^{\MDPtoMC{\MDP_{\downharpoonright C}}{\sigma_{2}}}\mid \forall \induniv\in\setuniv ,\  \exists k_{\induniv}\leq n ,\ $$
$$\rho(k_{\induniv})\in C^{\max}_{\even}(\parobj_{a})\rbrace] \geq 1 - \varepsilon. 
$$
\end{lemma}

\begin{proof}
This result is a consequence of Proposition~\ref{prop:opti_reach}. 
See that for all $\induniv\in\setuniv$  we have that $\sigma_{\mathbf{2}}$ reaches $D^{\max}_{\even}(\parobj_{\induniv})$ almost-surely from all state $s \in C$. 
Thus we can apply Proposition~\ref{prop:opti_reach} to all $\induniv\in \setuniv$ and find some $k_A\in \mathbb{N}$ such that by playing $\sigma_{2}$ from any $s\in S$ for $k_a$ steps, we reach $D^{\max}_{\even}(\parobj_{\induniv})$ with probability larger than $(1-\varepsilon)^{\frac{1}{\abs{\setuniv}}}$ (there exists such a $k_a$ for every $s\in S$, and as $S$ is finite, taking the maximum of them works for every $s\in S$). We take $n = \Sigma_{\induniv\in\setuniv} k_{\induniv}$. 
We consider a disjoint episodes labelled by $i$ for every $\induniv\in\setuniv$, of a duration of $k_{\induniv}$ steps. 
For all $\induniv\in\setuniv$, during episode $i$ we have a probability of visiting $D^{\max}_{\even}(\parobj_{\induniv})$ greater than $(1-\varepsilon)^{\frac{1}{\abs{\setuniv}}}$. 
Since these episodes are independent, the probability that for all $\induniv\in \setuniv$, we visited $D^{\max}_{\even}(\parobj_{\induniv})$ during episode $i$ is greater than $((1-\varepsilon)^{\frac{1}{\abs{\setuniv}}})^{\abs{\setuniv}}=(1-\varepsilon)$, and hence the desired property.
%\qed
\end{proof}
\begin{definition}
\label{def:strat_ugec}
Let $C \in \ugec(\MDP)$. Let $(n_i)_{i \in \mathbb{N}}$ be a sequence of naturals $n_i$ such that for all  $i \in \mathbb{N}$ we have:

$$
\prob_{\MDPtoMC{\MDP_{\downharpoonright C}}{\sigma_{2}}}[\lbrace \rho\cdot \rho' \in \paths^{\MDPtoMC{\MDP_{\downharpoonright C}}{\sigma_2}}\mid \rho\in \fpaths^{\MDPtoMC{\MDP_{\downharpoonright C}}{\sigma_2}},\ $$ $$\forall \induniv\in\setuniv ,\  \exists k_{\induniv}< n_i ,\ 
\rho(k_{\induniv})\in C^{\max}_{\even}(\parobj_{a})\rbrace] \geq 1 - 2^{-i}$$
whose existence is granted by Lemma~\ref{lemStratTwo3}. 
% We build strategy $\sigma_C$ as follows, starting with $i = 0$.
The strategy $\sigma_C$ is defined as follows.
\begin{enumerate}
% [label=\textbf{\alph*)}]
\item Play $\sigma_{2}$ for $n_i$ steps. Then $i = i + 1$ and go to \textbf{2}.
\item If for all $\induniv\in\setuniv:$ $C^{\max}_{\even}(\parobj_a)$ was visited in phase \textbf{1}, then go to \textbf{1}.

Else, play $\sigma_{1}$ until all $C^{\max}_{\even}(\parobj_a)$ are reached, and then go to \textbf{1}.
\end{enumerate}
\end{definition}

Strategy $\sigma_C$ can be effectively constructed, as the construction of $\sigma_1$ comes from Section~\ref{sec:typeone} and the construction of $\sigma_2$ comes from Lemma~\ref{lemma_stratTwo2}. 
The following lemma states that for all $s\in C$, the strategy $\sigma_C$ indeed satisfies $\combi{\setuniv}{\setqual}$:

\begin{lemma}
\label{lem:suf}
Let $C \in \ugec(\MDP,\setuniv,\setqual)$. For all $s \in C$, it holds that $\satisfy{s}{\sigma_C}{\MDP}\combi{\setuniv}{\setqual}$.
\end{lemma}
\begin{proof}
Let us first look at the $\limwedgeone{\induniv\in\setuniv}\univ(\parobj_{\induniv})$ condition. Each path $\pi$ has to follow one of these three cases:
\begin{itemize}
    \item Strategy $\sigma_1$ is only played a finite number of times, and for a finite duration: This means that eventually for some $i_0$, in each round $i>i_0$, in episodes of $n_i$ steps, $C^{\max}_{\even}(\parobj_a)$ was visited for all $\induniv\in\setuniv$. This also means that eventually only $\sigma_2$ is played and $\pi$ stays in $C$, hence all $\parobj_a$ are satisfied on $\pi$.
    \item Strategy $\sigma_1$ is eventually played for an infinite duration without coming back to 
$\textbf{1}$: By definition of $\sigma_1$, path $\pi$ satisfies all $\parobj_a$.
    \item Strategy $\sigma_1$ and $\sigma_2$ are both played infinitely often: The only way to stop strategy $\sigma_1$ is to have visited all $C^{\max}_{\even}(\parobj_a)$. As $\sigma_2$ and $\sigma_1$ were both played infinitely often, $\sigma_1$ was stopped infinitely often, and so $C^{\max}_{\even}(\parobj_a)$ was visited infinitely often for all $\induniv\in\setuniv$. As $\pi$ has to stay in $C$, it implies that $\pi$ satisfies all $\parobj_a$.
\end{itemize}

For the $\limwedgeone{\indqual\in\setqual}\univ(\parobj_{\indqual})$ conditions, we can prove that with probability $1$, eventually only $\sigma_2$ is played. As $\sigma_2$ has itself probability $1$ of ensuring all  $\parobj_{\indqual}$, we get that $\sigma_C$ satisfies $\limwedgeone{\indqual\in\setqual}\univ(\parobj_{\indqual})$.
\end{proof}

Now we construct a strategy $\sigma$ from $\sigma_T$ and $\sigma_C$.

\begin{definition}
\label{def:global}
Based on strategies $\sigma_T$ and $\sigma_{C}$ for all $C \in \ugec(\MDP)$, we build the global strategy $\sigma$ as follows.
\begin{enumerate}
\item Play $\sigma_T$ until a \typetwo EC $C$ is reached, then go to 2.
\item Play $\sigma_{C}$ forever.
\end{enumerate}
\end{definition} 

The following lemma concludes this direction of the proof of Lemma~\ref{lem:SAS}:
\begin{lemma}\label{lem:LeftRight}
\label{lemma_leftright}
It holds that $\satisfy{s_0}{\sigma}{\MDP}\combi{\setuniv}{\setqual}$.
\end{lemma}
\begin{proof}
First, we consider for all $\induniv\in\setuniv$ the objective $\univ(\parobj_a)$. Along each run $\pi$ consistent with $\sigma$, either (i) a \typetwo EC $C$ is eventually reached and $\sigma$ switches to $\sigma_C$, or (ii) $\sigma$ behaves as $\sigma_{T}$ forever. Since all strategies $\sigma_C$ and strategy $\sigma_{T}$ ensure $\univ(\parobj_a)$ and the parity condition is prefix-independent, we have that $\satisfy{s_0}{\sigma}{\MDP}\univ(\parobj_a)$.

Second, since with probability one, $\sigma_{T}$ reaches some {\sf Type\:{\rm II}} EC $C$, in which $\sigma_C$ ensures $\qual(\parobj_{\indqual})$ for all $\indqual\in\setqual$. Again invoking prefix-independence, we conclude that $\satisfy{s_0}{\sigma}{\MDP}\qual(\parobj_{\indqual})$, which ends our proof.
%\qed
\end{proof}

Now we sketch the proof of the left-to-right implication of Lemma~\ref{lem:SAS}. 
% We consider the following set $\mathcal{S}$ of subsets of $S$:
% $$\mathcal{S}=\{ R \subseteq S \mid \exists\, s \in S,\, \exists\, \sigma\text{ a strategy, }$$
% $$(\satisfy{s}{\sigma}{\MDP} \combi{\setuniv}{\setqual})\wedge(R =\mathtt{States}(\paths^{\MDPtoMC{\MDP}{\sigma}}(s)))\}.
% $$
% \todo{Give details about set S here?}
% We first compute a specific set, that is the set $S$ such for some state $s$ of $S$ there exists a strategy $\sigma$ such that $\satisfy{s}{\sigma}{\MDP}\combi{\setuniv}{\setqual}$, and there exists a \typetwo EC $C$ with all the states of $C$ being exactly the states appearing in $\MDPtoMC{\MDP_{\downharpoonright C}}{\sigma}$.
We make use of the following lemma.

\begin{lemma}
\label{pr-closed-a-q}
Given an MDP $\MDP = ( S,E,Act,\prob )$,  and two sets of parity conditions $\{\parobj_{\induniv}\pipe \induniv\in\setuniv\}$ and $\{\parobj_{\indqual}\pipe \indqual\in\setqual\}$, for all states $s,s'\in S$, and for all strategies $\sigma$ the following holds: if $\satisfy{s}{\sigma}{\MDP}\combi{\setuniv}{\setqual}$ and $s'\nin \llbracket\combi{\setuniv}{\setqual}\rrbracket$, then $s'\nin\paths^{\MDPtoMC{\Gamma}{\sigma}}(s)$.
\end{lemma}
\begin{proof}
The proof follows from the fact that for a strategy $\sigma$ in $\MDP$ that satisfies $\univ(\parobj)$ (resp. $\qual(\parobj)$), for all finite paths  $\pi$ from $s$ in $\MDPtoMC{\MDP}{\sigma}$, if $\pi$ leads to a state $\tilde{s}$, then it holds that for the set of paths originating from $\tilde{s}$, we have that $\univ(\parobj)$ (resp. $\qual(\parobj)$) is satisfied.
% A formal proof can be found in Appendix~\ref{app:typetwo}. 
The details of the proof is as follows:
% Both $\univ(\parobj_1)$ and $\qual(\parobj_2)$ are prefix-closed conditions. Thus if $\satisfy{s}{\sigma}{\MDP}\univ(\parobj_1)\wedge \qual(\parobj_2)$ and in this play $s'$ can be reached, then $s'$ would satisfy $\univ(\parobj_1)\wedge\qual(\parobj_2)$.

Assume by contradiction that $s'\in\paths^{\MDPtoMC{\Gamma}{\sigma}}(s)$.
Since it holds that  $s'\not\models_{\MDP}\combi{\setuniv}{\setqual}$, there can be two possibilities: either (i) for some $\induniv\in\setuniv$, we have $s'\not \models \univ(\parobj_{\induniv})$ or (ii) for some $\indqual\in\setqual$, we have $s'\not \models \qual(\parobj_{\indqual})$.

For (i), since the condition $\parobj_{\induniv}$ is prefix-independent, we have that any path going through 
$s'$ does not satisfy $\parobj_{\induniv}$, thus $s\not \models \univ(\parobj_{\induniv})$, and hence the contradiction.
For (ii), if $s' \not \models \qual(\parobj_{\indqual})$, and since by assumption $s'\in\paths^{\MDPtoMC{\Gamma}{\sigma}}(s)$, we have that $s'$ can be reached from $s$ with non-zero probability in the MC $\MDPtoMC{\MDP}{\sigma}$, and thus $\satisfy{s}{\sigma}{\MDP} \qual(\parobj_{\indqual})$ does not hold true, and hence the contradiction.
%\qed
\end{proof}

We now introduce the following definition.

\begin{definition}[Density]\label{def:dens} Let $\MDP=( S,E,Act,\prob)$ be an MDP, $s\in S$ an initial state, $\sigma$ a strategy, and $R\subseteq S$. We say that $R$ is \textit{dense} in $\sigma$ from $s$ if and only if for all $\rho \in \fpaths^{\MDPtoMC{\MDP}{\sigma}}(s)$, there exists $\rho'$ such that $\rho \cdot \rho' \in \fpaths^{\MDPtoMC{\MDP}{\sigma}}(s)$ and $\mathtt{Last}(\rho') \in R$. That is, after all prefixes in the tree $\paths^{\MDPtoMC{\MDP}{\sigma}}(s)$, there is a continuation that visits $R$.
\end{definition}

Now we state the following lemma that uses the above definition.

\begin{lemma}
\label{lem:bp}
Given an MDP $\MDP = ( S,E,Act,\prob)$, a state $s \in S$, a set of parity conditions $\{\parobj_{\induniv}\pipe \induniv\in\setuniv\}$, a set $R \subseteq S$, if there exists a strategy $\sigma$ such that $\satisfy{s}{\sigma}{\MDP} \limwedgeone{\induniv\in\setuniv}\univ(\parobj_a)$, and $R$ is dense in $\sigma$ from $s$, then $\satisfies{s}{G^{\MDP}_{R,\{\parobj_{\induniv}\pipe \induniv\in\setuniv\}}}\limwedgeone{\induniv\in\setuniv}\univ(\parobj_a)\wedge \univ(\Box\event B)$, with $G^{\MDP}_{R,\{\parobj_{\induniv}\pipe \induniv\in\setuniv\}}$ the Streett-B\"uchi game defined in Section~\ref{sec:typeone} where the B\"uchi condition is $B$.
\end{lemma}
% Lemma~\ref{lem:bp} is proved in Appendix~\ref{app:typetwo}.
\begin{proof}
% We construct a strategy $\sigma'$ to play in $G^{\MDP}_{R,\{\parobj_{\induniv}\pipe\induniv\in\setuniv\}}$ from strategy $\sigma$, and we show that $\sigma'$ is winning for the B\"uchi parity condition.

% \begin{proof}
We construct a strategy $\sigma'$ to play in $G^{\MDP}_{R,\{\parobj_{\induniv}\pipe\induniv\in\setuniv\}}$ from strategy $\sigma$, and we show that $\sigma'$ is winning for the B\"uchi parity condition. 
% The construction of $G^{\MDP}_{R,\{\parobj_{\induniv}\pipe\induniv\in\setuniv\}}$ is detailed in Appendix~\ref{app:typeone}.

Let $\alpha\colon S' \rightarrow S\cup \{\varepsilon\}$ be a mapping from the states in $G^{\MDP}_{R,\{\parobj_{\induniv}\pipe\induniv\in\setuniv\}}$ to the states in $\MDP$ such that $\alpha(s)=s$ for all $s \in S$, and $\alpha(s,a,i)=\varepsilon$ for all $(s,a,i) \in (S\setminus R) \times Act \times \{0,1\}$. Then we extend $\alpha$ to map prefixes in $G^{\MDP}_{R,\{\parobj_{\induniv}\pipe\induniv\in\setuniv\}}$ to prefixes in $\MDP$.
% and to prefixes in the strategy tree $\paths_s^{\MDP}(\sigma)$. 

Now, we label each prefix $\rho \in \fpaths^{\MDPtoMC{\MDP}{\sigma}}(s)$ by a finite continuation $\mathtt{lab}(\rho)=\rho'$ such that (i) $\rho \cdot \rho' \in \fpaths^{\MDPtoMC{\MDP}{\sigma}}(s)$ and $\last{\rho'} \in R$; and (ii) if $\rho'=\rho'_1 \cdot s \cdot \rho'_2$ is the label of $\rho$ then $\mathtt{lab}(\rho \cdot \rho'_1 \cdot s)=\rho'_2$. 	It should be clear that due to the density of $R$ (as given in Definition~\ref{def:dens}), such a labelling is always possible.

Now, we define the strategy $\sigma'$ from $\sigma$ and this labelling:
  \begin{itemize}
	\item (i) For a prefix $\rho$ in the game such that $\mathtt{Last}(\rho) \in R$, the only possible choice is the only action available as $R$ is absorbing in the game.
	\item (ii) For a prefix $\rho$ such that $\mathtt{Last}(\rho) \in S \cup (S\setminus R) \times \{Act\} \times \{0,1\}$ and such that $\alpha(\rho)$ is not consistent with $\sigma$, we choose any $a\in Act$. 
% 	such that $(\mathtt{Last}(\rho),s')\in E'$. 
	\item (iii) For a prefix $\rho$ such that $\mathtt{Last}(\rho) \in S$ and such that $\alpha(\rho)$ is consistent with $\sigma$, we define $\sigma'(\rho)=\sigma(\alpha(\rho))$.
	\item (iv) For a prefix $\rho$ such that $\mathtt{Last}(\rho) \in (S\setminus R) \times \{ Act\}\times\{1\}$ and such that $\alpha(\rho)$ is consistent with $\sigma$, then $\alpha(\rho)$ in the tree $\paths^{\MDPtoMC{\MDP}{\sigma}}(s)$ is labelled with a finite path $\rho'$ that leads to a state in $R$, i.e., $\mathtt{lab}(\alpha(\rho))=\rho'$, then we define $\sigma'(\rho)=\mathtt{First}(\rho')$.	\item (iv) For a prefix $\rho$ such that $\mathtt{Last}(\rho) \in (S\setminus R) \times \{ Act\}\times\{0\}$ and such that $\alpha(\rho)$ is consistent with $\sigma$, then we play the only action available.
\end{itemize}	
	
We now show that $\sigma'$ wins the B\"uchi $B$ and parity $\parobj'_a$ conditions for all $\induniv\in\setuniv$ defined in the game $G^{\MDP}_{R,\{\parobj_{\induniv}\pipe\induniv\in\setuniv\}}$ from state $s$. First, for $\induniv\in\setuniv$ as $\sigma$ is enforcing $\parobj_a$ in $\MDP$, we have that $\parobj'_a$ replicates the priorities given by $\parobj_a$, and absorbing states in $R$ have even priorities for $\parobj'_a$, it is clear that $\sigma'$ enforces $\parobj'_a$ in the game $G^{\MDP}_{R,\{\parobj_{\induniv}\pipe\induniv\in\setuniv\}}$. Now, for the B\"uchi condition, we consider the following case study on the paths $\pi$ consistent with $\sigma'$ in the game.
	\begin{itemize}
		\item If $\pi$ ends up in $R$, then it is clearly B\"uchi accepting.
		\item If $\pi$ is such that after a finite prefix $\rho$, we always visit states $s' \in S\times Act\times \{1\}$, then according to the definition of $\sigma'$, the play will follow a finite path $\rho'$ to a state in $R$, and so $\pi$ reaches $R$ and as a consequence $\pi$ is B\"uchi accepting.
		\item Finally, if $\pi$ is such that it is eventually never the case that we visit any $s' \in S\times Act\times \{1\}$, then we visit infinitely often the copies $s' \in S\times Act\times \{0\}$, and because all $(s,a,0) \in B$, we have that $\pi$ is also B\"uchi accepting.
	\end{itemize}
So in all cases, play $\pi$ satisfies the B\"uchi condition, and we are done.
% %\qed 
% \end{proof}

Let $\alpha\colon S' \rightarrow S\cup \{\varepsilon\}$ be a mapping from the states in $G^{\MDP}_{R,\{\parobj_{\induniv}\pipe\induniv\in\setuniv\}}$ to the states in $\MDP$ such that $\alpha(s)=s$ for all $s \in S$, and $\alpha(s,a,i)=\varepsilon$ for all $(s,a,i) \in (S\setminus R) \times Act \times \{0,1\}$. Then we extend $\alpha$ to map prefixes in $G^{\MDP}_{R,\{\parobj_{\induniv}\pipe\induniv\in\setuniv\}}$ to prefixes in $\MDP$.
% and to prefixes in the strategy tree $\paths_s^{\MDP}(\sigma)$. 

Now, we label each prefix $\rho \in \fpaths^{\MDPtoMC{\MDP}{\sigma}}(s)$ by a finite continuation $\mathtt{lab}(\rho)=\rho'$ such that (i) $\rho \cdot \rho' \in \fpaths^{\MDPtoMC{\MDP}{\sigma}}(s)$ and $\last{\rho'} \in R$; and (ii) if $\rho'=\rho'_1 \cdot s \cdot \rho'_2$ is the label of $\rho$ then $\mathtt{lab}(\rho \cdot \rho'_1 \cdot s)=\rho'_2$. 	It should be clear that due to the density of $R$ (as given in Definition~\ref{def:dens}), such a labelling is always possible.

Now, we define the strategy $\sigma'$ from $\sigma$ and this labelling:
  \begin{itemize}
	\item (i) For a prefix $\rho$ in the game such that $\mathtt{Last}(\rho) \in R$, the only possible choice is the only action available as $R$ is absorbing in the game.
	\item (ii) For a prefix $\rho$ such that $\mathtt{Last}(\rho) \in S \cup (S\setminus R) \times \{Act\} \times \{0,1\}$ and such that $\alpha(\rho)$ is not consistent with $\sigma$, we choose any $a\in Act$. 
% 	such that $(\mathtt{Last}(\rho),s')\in E'$. 
	\item (iii) For a prefix $\rho$ such that $\mathtt{Last}(\rho) \in S$ and such that $\alpha(\rho)$ is consistent with $\sigma$, we define $\sigma'(\rho)=\sigma(\alpha(\rho))$.
	\item (iv) For a prefix $\rho$ such that $\mathtt{Last}(\rho) \in (S\setminus R) \times \{ Act\}\times\{1\}$ and such that $\alpha(\rho)$ is consistent with $\sigma$, then $\alpha(\rho)$ in the tree $\paths^{\MDPtoMC{\MDP}{\sigma}}(s)$ is labelled with a finite path $\rho'$ that leads to a state in $R$, i.e., $\mathtt{lab}(\alpha(\rho))=\rho'$, then we define $\sigma'(\rho)=\mathtt{First}(\rho')$.	\item (iv) For a prefix $\rho$ such that $\mathtt{Last}(\rho) \in (S\setminus R) \times \{ Act\}\times\{0\}$ and such that $\alpha(\rho)$ is consistent with $\sigma$, then we play the only action available.
\end{itemize}	
	
We now show that $\sigma'$ wins the B\"uchi $B$ and parity $\parobj'_a$ conditions for all $\induniv\in\setuniv$ defined in the game $G^{\MDP}_{R,\{\parobj_{\induniv}\pipe\induniv\in\setuniv\}}$ from state $s$. First, for $\induniv\in\setuniv$ as $\sigma$ is enforcing $\parobj_a$ in $\MDP$, we have that $\parobj'_a$ replicates the priorities given by $\parobj_a$, and absorbing states in $R$ have even priorities for $\parobj'_a$, it is clear that $\sigma'$ enforces $\parobj'_a$ in the game $G^{\MDP}_{R,\{\parobj_{\induniv}\pipe\induniv\in\setuniv\}}$. Now, for the B\"uchi condition, we consider the following case study on the paths $\pi$ consistent with $\sigma'$ in the game.
	\begin{itemize}
		\item If $\pi$ ends up in $R$, then it is clearly B\"uchi accepting.
		\item If $\pi$ is such that after a finite prefix $\rho$, we always visit states $s' \in S\times Act\times \{1\}$, then according to the definition of $\sigma'$, the play will follow a finite path $\rho'$ to a state in $R$, and so $\pi$ reaches $R$ and as a consequence $\pi$ is B\"uchi accepting.
		\item Finally, if $\pi$ is such that it is eventually never the case that we visit any $s' \in S\times Act\times \{1\}$, then we visit infinitely often the copies $s' \in S\times Act\times \{0\}$, and because all $(s,a,0) \in B$, we have that $\pi$ is also B\"uchi accepting.
	\end{itemize}
So in all cases, play $\pi$ satisfies the B\"uchi condition, and we are done.
%\qed 
\end{proof}

% \rbcomment{rewrote this}
Lemma~\ref{pr-closed-a-q} implies that for all initial state $s'$, in $\MDPtoMC{\MDP_{\downharpoonright C}}{\sigma}$, after all finite path $\rho$ beginning from some state $(q',s')$, and ending in $(q'',s'')$ it holds that $\satisfy{(q'',s'')}{\sigma}{\MDP}\combi{\setuniv}{\setqual}$. Thus after every finite path, there exists a continuation that visits $\mathcal{T}_{\mathbf{II},\MDP,\setuniv,\setqual} $, hence $\mathcal{T}_{\mathbf{II},\MDP,\setuniv,\setqual} $ is dense from $s'$ in $\sigma$, and so by Lemma~\ref{lem:bp}, Lemma~\ref{thm:reach-Almagor} and Lemma~\ref{lem:implies_is_univ}, we have that $\satisfies{s_0}{\MDP}\limwedgeone{\induniv\in\setuniv}\univ(\parobj_a)\wedge\qual(\event\mathcal{T}_{\mathbf{II},\MDP,\setuniv,\setqual})$. Recall that $s_0$ is the initial state in Lemma~\ref{lem:SAS}. We detail below why $\mathcal{T}_{\mathbf{II},\MDP,\setuniv,\setqual} $ is non-empty.

\begin{lemma}%\label{lem\typetwo EC}
\label{lem:ugec}
If $\satisfies{s_0}{\MDP} \combi{\setuniv}{\setqual}$ then $\ugec(\MDP) \neq \varnothing$.
\end{lemma}
% The proof can be sketched as follows. 
% We first give a sketch of the proof before giving the formal details.
% Given an MDP $\MDP = ( S,E,Act,\prob)$, an initial state $s$, a strategy $\sigma$ and a set of paths $\Pi \subseteq \paths^{\MDPtoMC{\MDP}{\sigma}}(s)$, we define $\mathtt{States}(\Pi) = \{s \in S \mid \exists\, \pi \in \Pi,\, \exists\, n \in \mathbb{N}_0,\, \pi(n) = s\}.$
% To prove this lemma, we first study the following set $\mathcal{S}$ of subsets of $S$:
% $$\mathcal{S}=\{ R \subseteq S \mid \exists\, s \in S,\, \exists\, \sigma\text{ a strategy, }$$
% $$(\satisfy{s}{\sigma}{\MDP} \combi{\setuniv}{\setqual})\wedge(R =\mathtt{States}(\paths^{\MDPtoMC{\MDP}{\sigma}}(s)))\}.
% $$
% \sgcomment{$\sigma$ not defined above.}
% Intuitively, this set contains every subset of $S$ that contains all states reachable by some witness strategy $\sigma$ for $\combi{\setuniv}{\setqual}$, from some state $s \in S$.
% First note that $\satisfies{s_0}{\MDP} \combi{\setuniv}{\setqual}$ implies that $\mathcal{S}$ is non-empty, as for a witness strategy $\sigma$, we have $R =  \mathtt{States}(\paths^{\MDPtoMC{\MDP}{\sigma}}{s_0})) \in \mathcal{S}$, by definition. 

% Second, we show that all minimal elements of $\mathcal{S}$ under set inclusion $\subseteq$ are \typetwo ECs, i.e., for all $R\in\min_{\subseteq}(\mathcal{S})$, it holds that $R \in \ugec(\MDP)$.
% The details can be found in Appendix~\ref{app:typetwo}.

\begin{proof}
We recall that given an MDP $\MDP$, an initial state $s$, a strategy $\sigma$ and a set of paths $\Pi \subseteq \paths^{\MDPtoMC{\MDP}{\sigma}}(s)$, we define
\[
\mathtt{States}(\Pi) = \{s \in S \mid \exists\, \pi \in \Pi,\, \exists\, n \in \mathbb{N}_0,\, \pi(n) = s\}.
\]

Now consider the following set $\mathcal{S}$ of subsets of $S$:
$$\mathcal{S}=\{ R \subseteq S \mid \exists\, s \in S,\, \exists\, \sigma\text{ a strategy, }$$
$$(\satisfy{s}{\sigma}{\MDP} \combi{\setuniv}{\setqual})\wedge(R =\mathtt{States}(\paths^{\MDPtoMC{\MDP}{\sigma}}(s)))\}.
$$
% \begin{align*}
% \mathcal{S}=\{ R \subseteq S \mid \exists\, s \in S,\,
% (\satisfies{s}{\MDP} \combi{\setuniv}{\setqual}) \wedge
% \end{align*}
% \begin{align*}
% (R =  \mathtt{States}(\paths^{\MDPtoMC{\MDP}{\sigma}}(s)))\}.
% \end{align*}
% Intuitively, this set contains every subset of $S$ that contains all states reachable by some witness strategy $\sigma$, from some state $s \in S$.
% First note that $\satisfies{s_0}{\MDP} \combi{\setuniv}{\setqual}$ implies that $\mathcal{S}$ is non-empty, as for a witness strategy $\sigma$, we have $R =  \mathtt{States}(\paths^{\MDPtoMC{\MDP}{\sigma}}{s_0})) \in \mathcal{S}$, by definition.

% Second, we show that all minimal elements of $\mathcal{S}$ for set inclusion $\subseteq$ are \typetwo ECs, i.e., for all $R\in\min_{\subseteq}(\mathcal{S})$, it holds that $R \in \ugec(\MDP)$. This will establish our lemma. 
We have that for each $R \in\min_{\subseteq}(\mathcal{S})$, the following properties hold:
\begin{enumerate}
\item $R$ is an EC in $\MDP$,
\item $\forall\, s\in R,\, \satisfies{s}{\MDP_{\downharpoonright R}} \limwedgeone{\induniv\in\setuniv}(\univ(\parobj_a) \wedge \qual(\event R^{\max}_{\even}(\parobj_a)))$,
\item $\forall\, s\in R,\, \satisfies{s}{\MDP_{\downharpoonright R}} \limwedgeone{\induniv\in\setuniv}\qual(\parobj_a) \wedge \qual (\parobj_{\indqual})$.
\end{enumerate}

Before proving these three items, we claim that for all $R\in \min_{\subseteq}(\mathcal{S})$, and $s \in R$, there exists a strategy $\sigma_R$ such that $\satisfy{s}{\sigma_R}{\MDP_{\downharpoonright R}} \combi{\setuniv}{\setqual}$, i.e., $\sigma_R$ satisfies the property without leaving $R$. This is a direct consequence of Lemma~\ref{pr-closed-a-q} and the minimality of $R$ in $\mathcal{S}$ for the $\subseteq$ order. We use strategy $\sigma_R$ in the rest of the proof.

Item \textbf{1)}. We first prove that $R$ is strongly connected. By contradiction, assume it is not the case, i.e., that there exist $s, s' \in R$ such that there is no path in $R$ from $s$ to $s'$. Then, let $R'$ be the set of states reachable with strategy $\sigma_R$ from a prefix $\rho$ ending in $s$. By Lemma~\ref{pr-closed-a-q}, we have that $R' \in \mathcal{S}$. But as there is no path from $s$ to $s'$ in $R$, we have that $s' \not\in R'$ and $R' \subsetneq R$. This contradicts the minimality of $R$, hence we conclude that $R$ is strongly connected. Then, for all state $s$ of $R$, and $a\in Act(s)$, we have that $\post(s,a)\neq \varnothing$, as $R$ is strongly connected. Hence, it remains to show that for all states $s\in R$, there exists an action $a$ such that we have $\post(s,a) \subseteq R$. By contradiction, fix some $s \in R$, and assume that for all action $a$ there exists $s_a \not\in R$ such that $(s, s_a) \in E$. As $R$ belongs to $\mathcal{S}$, recall that $R =\mathtt{States}(\paths^{\MDPtoMC{\MDP}{\sigma}}(s''))$ for some strategy $\sigma$ and state $s''$. Since $s \in R$, there exists a prefix $\rho \in \fpaths^{\MDPtoMC{\MDP}{\sigma}}(s'')$ such that $\mathtt{Last(\rho)} = s$. But then, for some $a$, prefix $\rho' = \rho \cdot s_a$ also belongs to $\fpaths^{\MDPtoMC{\MDP}{\sigma}}(s'')$, and $s_a \in R$. Thus, we conclude that $R$ is indeed an EC in $\MDP$.

Item \textbf{2)}. We fix any $s \in R$, and we prove that $\satisfies{s}{\MDP_{\downharpoonright R}} \limwedgeone{\induniv\in\setuniv}(\univ(\parobj_a) \wedge \qual(\event R^{\max}_{\even}(\parobj_a)))$. As seen above, from the minimality of $R$ and Lemma~\ref{pr-closed-a-q}, we know that $\satisfy{s}{\sigma_R}{\MDP_{\downharpoonright R}} \\ \combi{\setuniv}{\setqual}$. Now, again from the minimality of $R$ in $\mathcal{S}$, we know that in the subtree induced by $\paths^{\MDPtoMC{\MDP}{\sigma_R}}(s)$, every non-empty subset $R' \subseteq R$ is dense, otherwise, there would exists some $r\in R$ such that a subtree of $\paths^{\MDPtoMC{\MDP}{\sigma_R}}(s)$ only visits $R\backslash r$, and this subtree would itself define a winning strategy for $\wedge{\setuniv}{\setqual}$, breaking the minimality assumption. As a consequence, for all prefixes $\rho$, there exists a path beginning with $\rho$ that eventually reaches a state of $R'$. Using the reduction to a parity-B\"uchi game which underlies Lemma~\ref{thm:reach-Almagor}, we deduce from the density argument presented in Lemma~\ref{lem:bp} that we can build $\sigma$ from $\sigma_R$ such that $\satisfy{s}{\sigma}{\MDP_{\downharpoonright R}} \limwedgeone{\induniv\in\setuniv}(\qual(\parobj_a) \wedge \qual(\event R^{\max}_{\even}(\parobj_a)))$.
It remains to argue that $R^{\max}_{\even}(\parobj_a)$ is non-empty to prove this item. This is necessarily true, otherwise $R^{\max}_{\odd}(\parobj_a)$ would be non-empty, and $\sigma_R$ would not ensure $\univ(\parobj_a)$ in $R$ (as $R^{\max}_{\odd}(\parobj_a)$ would be a dense subset).

Item \textbf{3)}. This is trivial as, for $s \in R$, we have that  $\sigma_R$ enforces $\satisfy{s}{\sigma_R}{\MDP_{\downharpoonright R}} \combi{\setuniv}{\setqual}$, a stronger property.
%\qed
\end{proof}
Finally, we state the following lemma. It refines Lemma~\ref{lem:ugec}
% , that shows  $\ugec(\MDP)$ is non-empty, 
by showing that some \typetwo EC of $\ugec(\MDP)$ can be reached almost-surely while satisfying $\limwedgeone{\induniv\in\setuniv}(\univ(\parobj_a)$. We define $\\ \mathcal{T}_{\mathbf{II},\MDP,\setuniv,\setqual}^{\min}= {\displaystyle \cup_{R \in\min_{\subseteq}(\mathcal{S})} R}$, that is the set of all states that belong to a minimal set $R$ of $\mathcal{S}$.
% , and establish that $\mathcal{T}_{\mathbf{II},\MDP,\setuniv,\setqual}^{\min}$ can be reached almost-surely if it holds that $\satisfies{s_0}{\MDP} \combi{\setuniv}{\setqual}$. 

% \rbchanged{two "show" here, do you see a way to rephrase?}

\begin{lemma}
%\label{lemRightLeft}
\label{lemma_rightleft}
If $\satisfies{s_0}{\MDP} \combi{\setuniv}{\setqual}$ then $\satisfies{s_0}{\MDP} \limwedgeone{\induniv\in\setuniv}(\univ(\parobj_a) \wedge \qual(\event \mathcal{T}_{\mathbf{II},\MDP,\setuniv,\setqual}^{\min}))$.
\end{lemma}
% \sgcomment{Is $\mathcal{T}_{\mathbf{II},\MDP,\setuniv,\setqual}^{\min}$ defined?}
% The details of the proof can be found in Appendix~\ref{app:typetwo}.
\begin{proof}
Let $\sigma$ be a witness for $\satisfies{s_0}{\MDP} \combi{\setuniv}{\setqual}$. By Lemma~\ref{lem:ugec}, we know that $\ugec(\MDP)$ is non-empty, and so is $\mathcal{T}_{\mathbf{II},\MDP,\setuniv,\setqual}^{\min}$. Furthermore, we claim that $\mathcal{T}_{\mathbf{II},\MDP,\setuniv,\setqual}^{\min}$ is dense in the tree induced by $\paths^{\MDPtoMC{\MDP}{\sigma}}(s_0)$. Indeed, by Lemma~\ref{pr-closed-a-q}, after every prefix $\rho \in \fpaths^{\MDPtoMC{\MDP}{\sigma}}(s_0)$, the following property holds: $\satisfy{\mathtt{Last}(\rho)}{\sigma[\rho]}{\MDP} \combi{\setuniv}{\setqual}$ (hence $\mathcal{T}_{\mathbf{II},\MDP,\setuniv,\setqual}^{\min}$ is reached, repeating the previous arguments). Since this holds for all prefixes of $\fpaths^{\MDPtoMC{\MDP}{\sigma}}(s_0)$, we have that
$\mathcal{T}_{\mathbf{II},\MDP,\setuniv,\setqual}^{\min}$ is indeed dense in the tree of $\sigma$.  Hence, again using the density argument presented in Lemma~\ref{lem:bp}, we can build $\sigma'$ from $\sigma$ such that $\satisfy{s_0}{\sigma'}{\MDP} \limwedgeone{\induniv\in\setuniv}(\univ(\parobj_a) \wedge \qual(\event \mathcal{T}_{\mathbf{II},\MDP,\setuniv,\setqual}^{\min}))$.
%\qed
\end{proof}

We end this section with the following observation. 
Lemma \ref{pr-closed-a-q} implies that winning strategies only visit states belonging to $\llbracket\combi{\setuniv}{\setqual}\rrbracket$. 
% As a consequence, in the following sections, we assume that for every state $s$ of $\MDP$, we have $s\in \llbracket\combi{\setuniv}{\setqual}\rrbracket$. 
As a consequence, pruning the states that do not satisfy $\combi{\setuniv}{\setqual}$ does not affect correctness. We use this pruned MDP in the rest of the paper. We state this formally:
\begin{assumption}
\label{as:pruned}
For every state $s$ of\ \ $\MDP$, we have: $\\ s\!\in\!\llbracket\combi{\setuniv}{\setqual}\rrbracket$.
\end{assumption}
We detail how to do this pruning in Section~\ref{sec:complex}, we use Lemma~\ref{lem:SAS} to find the states that do not satisfy the objective.
% explicitly prune the states that do not satisfy this objective, using Lemma~\ref{lem:SAS}.

%Given an MDP $\MDP$, we define the following properties:

%Given an end-component $C$, parity functions $\parobj_1,\ldots, \parobj_n$, and a target set of states $T$, 
%a strategy $\sigma$ \emph{safely reaches} $T$ with respect towards $\parobj_1,\ldots, \parobj_n$ if 
%$\forall\, s\in C,\, \satisfies{s}{\MDP{\downharpoonright C}} A(\parobj_1) \wedge\ldots \wedge A(\parobj_n) \wedge Q(\event T)$.
%We denote this property $SR(C,(\parobj_1,\ldots,\parobj_n),T)$.

%Given an end-component $C$, parity functions $\parobj_1,\ldots, \parobj_n$,
%a strategy $\sigma$ almost-surely ensures $\parobj_1,\ldots, \parobj_n$ on $C$ if 
%$\forall\, s\in C,\, \satisfies{s}{\MDP{\downharpoonright C}} Q(\parobj_1) \wedge\ldots \wedge Q(\parobj_n)$ or equivalently $\satisfies{s}{\MDP{\downharpoonright C}} Q(\parobj_1 \wedge\ldots \wedge \parobj_n).$
%We denote this property $Q(C,(\parobj_1,\ldots,\parobj_n))$.

\section{\typethree end-components}
\label{sec:typethree}
In this section, we define \typethree ECs.
% , and give an algorithm to compute the set of maximal \typethree ECs.
These end-components are used to characterize the winning strategies for formulas of the form $\combi{\setuniv}{\setqual}\wedge \posi(\parobj_{\indposi})$.
To do so, we show in Lemma~\ref{lem:aqp-ec} that all \typethree ECs satisfy such a formula. 
In Lemma~\ref{lem:SASPr}, we use the previous lemma and the technical Lemma~\ref{lem:probToAS} to relate the formula $\combi{\setuniv}{\setqual}\wedge \posi(\parobj_{\indposi})$ to the almost-sure reachability of the set of \typethree ECs under the constraint $\combi{\setuniv}{\setqual}$. 
We explain in Section~\ref{sec:nz_e} how to compute reachability of a set of states under the constraint $\combi{\setuniv}{\setqual}$, and in Section~\ref{sec:complex} we explain how to compute the set of \typethree ECs.
% In the example in Figure \ref{fig:GEC}, we illustrated that we may need randomized strategies for non-zero objectives.

% \todo{roadmap}

Given two sets of parity conditions $\{\parobj_{\induniv}\pipe \induniv\in\setuniv\}$, $\{\parobj_{\indqual}\pipe \indqual\in\setqual\}$ and another parity condition $\parobj_{\indposi}$, an end-component $C$ of $\MDP$ is \typethree($\setuniv,\setqual,\parobj_{\indposi}$)
if the following two properties hold:
\begin{itemize}
	\item $\mathbf{(III_1)}$ $\forall\, s\in C,\, \satisfies{s}{\MDP} \combi{\setuniv}{\setqual}$;

% 	density argument presented in Appendix B, but we rely on Streett-B\"uchi games
% 	\sgcomment{Removed in an earlier version.}
	\item $\mathbf{(III_2)}$ $\forall\, s\in C,\, \satisfies{s}{\MDP{\downharpoonright C}} \limwedgeone{\induniv\in\setuniv}\qual(\parobj_{\induniv})\wedge \limwedgeone{\indqual\in\setqual}\qual(\parobj_{\indqual})\wedge\qual(\parobj_{\indposi})$
\end{itemize}
We write \typethree($\setuniv,\setqual,\parobj_{\indposi}$) EC as \typethree ECs when the parity sets are clear from the context.
We note that condition $\mathbf{(III_1)}$ may require an infinite memory strategy (see Lemma~\ref{lem:SAS}), and can always be satisfied in the pruned MDP, due to Assumption~\ref{as:pruned}. Condition $\mathbf{(III_1)}$ can be checked using~\cite{etessami2007multi}.
% For condition $\mathbf{(III_2)}$, one may need a randomized memoryless strategy\todo{Should we already talk about this?} \sgcomment{Or should we state this in a lemma?}.
Note that condition $\mathbf{(III_2)}$ is only about the parity conditions indexed by $\setuniv$ and $\setqual$, and must hold while staying inside the end-component $C$, but the witness strategy for $\mathbf{(III_1)}$ may leave $C$.
This notion strengthens in a non-trivial way the notion of \textit{very-good end-component} (VGEC in \cite{BRR17}) which are \typethree ECs where $\setuniv$ is a singleton, $\setqual=\varnothing$, and the $\parobj_{\indposi}$ stays as it is.
From condition $\mathbf{(III_1)}$ and Lemma~\ref{lem:SAS} we know that if there exists a \typethree EC in the MDP $\MDP$ then there also exists a  \typetwo EC.
We introduce the following notations: $\nec(\MDP,\setuniv,\setqual,\parobj_{\indposi})$ is the set of all \typethree($\setuniv,\setqual,\\ \parobj_{\indposi}$) ECs of $\MDP$, and $\mathcal{T}_{\mathbf{III},\setuniv,\setqual,\parobj_{\indposi}} = {\displaystyle \cup_{C \in \nec(\MDP,\setuniv,\setqual,\parobj_{\indposi})} C}$ is the set of states belonging to some \typethree EC in $\MDP$. 

In the sequel, we relate \typethree ECs to the formula $\\ \combi{\setuniv}{\setqual}\\ \wedge \posi(\parobj_{\indposi})$. In particular, from all states belonging to a {\sf Type\:{\rm III}} ECs there exists a strategy satisfying the formula $\combi{\setuniv}{\setqual}\wedge \posi(\parobj_{\indposi})$. We illustrate this by the following example.

\begin{example}
Consider the example in Figure~\ref{fig:VGEC}.
We show a strategy satisfying $\univ(\parobj_1) \wedge \qual(\parobj_2)\wedge \posi(\parobj_3)$ in an MDP that is also a \typethree EC:
% We restrict our attention to the end-component inside the dashed box that is a \typetwo EC.
All the states satisfy condition  $\mathbf{(III_1)}$ when action $a$ is chosen from state $s$, and all the states satisfy condition $\mathbf{(III_2)}$ when action $b$ is chosen from state $s$.

\begin{figure}[t]
\centering
\vspace{-15pt}
\hspace{-15pt}
\centering
\scalebox{0.9}{
	\begin{tikzpicture}
% 		\node[player] (sinit) at (2,2) {0,1} ;
% 		\node[player] (s1) at (2,0) {0,1} ;
		
		\node[player] (s2) at (4,1) {2,2,1} ;
% 		\node[player] (s3) at (4,2) {1,1} ;
		\node[player,initial,initial above,initial text={}] (s4) at (6,1) {1,1,1} ;
		\node[player] (s5) at (8,2) {1,1,1} ;
		\node[player] (s6) at (8,0) {2,2,2} ;

% 		\node[player] (s7) at (6,4) {3,3} ;
% 		\node[player] (s8) at (8,4) {4,4} ;

		\node () at (6,0.3) {\large {\color{blue}$s$}} ;
	   % \node[subgraph] (Nc) at (5.75,1) {};

 \path[-latex]  
    %             (sinit) edge[bend left=20] node[right] {\small $a,1$} (s1)
				% (s1) edge[bend left=20] node[left] {\small $a,0.5$} (sinit)				
				% (s1) edge node[above] {\small $a,0.5$} (s2)
				% (s2) edge[bend right=20] node[right] {\small $a,1$} (s4)
				(s4) edge node[above] {\small $a,1$} (s2)
				% (s3) edge node[left] {\small $a,0.5$} (s2)
				% (s3) edge[bend right=20] node[below] {\small $a,0.5$} (s4)
				(s4) edge[bend left=20] node[above] {\small $b,1$} (s5)
				(s5) edge[bend left=20] node[below] {\small $a,0.5$} (s4)
				(s5) edge node[right] {\small $a,0.5$} (s6)
				(s6) edge[bend left=20] node[below] {\small $a,1$} (s4)
		    	(s2) edge[loop above] node[above] {\small $a,1$} (s2)
				% (s4) edge[bend left=20] node[left] {\small $c,1$} (s7)
				% (s7) edge[bend left=20] node[right] {\small $a,0.5$} (s4)
				% (s7) edge node[above] {\small $a,0.5$} (s8)
				% (s8) edge[bend left=20] node[above,xshift=.6cm] {\small $a,1$} (s4)
				
		;
	\end{tikzpicture}
}
\caption{\label{fig:VGEC}An example of a \typethree EC}
\end{figure}
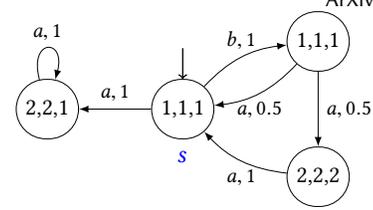
Now the strategy from $s$ to satisfy  $\univ(\parobj_1) \wedge \qual(\parobj_2)\wedge \posi(\parobj_3)$  is the following.
In \emph{round $1$}, action $b$ is chosen some $i_0$ times.
If the state with parity value ($2,2,2$) is visited 
% when action $b$ is chosen $i_0$ times, 
in round $1$,
then proceed to round $2$. Otherwise action $a$ is chosen from state $s$ at the end of round $1$.
In round $2$, action $b$ is chosen $i_1=i_0+1$ times and so on, resulting in action $b$ being chosen $i_j=i_0+j$ times in round $j$.
% In round $2$, if the state with parity ($2,2$) is not visited when action $b$ is chosen $i_0+1$ times, then action $a$ is chosen until the state with parity ($2,1$) is visited.
From every round, we proceed to the next round if the state with parity value ($2,2,2$) is visited in the current round, or otherwise we switch to playing action $a$ from state $s$.
Now we compute the probability of never switching to action $a$ in state $s$.
The probability of not choosing action $a$ from $s$ after the first round is $1-2^{-i_0}$.
The probability of not choosing $a$ from $s$ after the first round and as well as after the second round is $(1-2^{-i_0})\cdot (1-2^{-(i_0+1)})$,
and thus the probability of not choosing $a$ from $s$ after each of the first $n$ rounds is $p(n)=\Pi _{j=0}^{n-1}(1-2^{-(i_0+j)})$.
% not choosing $a$ for $k$ rounds when $n$ rounds have already been played is $p(n,k)=\prod_{j=n}^{k+n}(1-2^{-(i_0+j)})$.
% Let $p(n) = \lim_{k \rightarrow \infty} p(k,n)$.
In~\cite{BRR17}, it has been shown that $\lim_{n \rightarrow \infty}p(n)>0$, implying that with non-zero probability, action $a$ will never be played in $s$.
Hence with non-zero probability $\parobj_3$ holds.
% for $\lim_{i \rightarrow \infty}$.
\end{example}

% We now give the intuitions behind these properties through some examples.

\begin{lemma}
\label{lem:aqp-ec}
Given an MDP $\MDP$, and a \typethree EC $C$ in $\MDP$, for all states $s\in C$, we have $\satisfies{s}{\MDP}\combi{\setuniv}{\setqual}\wedge \posi(\parobj_{\indposi})$.
\end{lemma} 
\begin{proof}
Consider some $\varepsilon>0$, and $\fun{f}{\mathbb{N}_0}{\mathbb{Q}\cap (0,1]}$ a series of probabilities such that the infinite product $\prod_{i\in\mathbb{N}_0}{} f(i)>1-\varepsilon$. 
Let $\sigma_1$ be a strategy satisfying ($\mathbf{III_1}$), and $\sigma_2$ be a strategy satisfying ($\mathbf{III_2}$).
Using Proposition~\ref{prop:opti_reach}, we associate with $\sigma_2$ a sequence of numbers $\fun{g}{\mathbb{N}_0}{\mathbb{N}_0}$ such that if $\sigma_2$ is played for $g(i)$ steps then for all $\induniv\in\setuniv$ we have that $C^{\max}_{\even}(\parobj_{\induniv})$ is visited with probability at least $f(i)$. 
We consider the following infinite-memory strategy $\sigma$ by combining strategies $\sigma_1$ and $\sigma_2$ as follows:
\begin{itemize}
\item a) Let $i=0$
\item b) Play $\sigma_2$ for $g(i)$ steps. Let $i=i+1$.
\item c) if for all $\induniv\in\setuniv$ we have $C^{\max}_{\even}(\parobj_{\induniv})$ was visited, then go to b), else play $\sigma_1$ forever.
\end{itemize}
When following $\sigma$, in each round $i$, we have probability at least $f(i)$ of continuing to play $\sigma_2$. The probability of playing $\sigma_2$ forever is thus the same as the probability of visiting all $C^{\max}_{\even}(\parobj_{\induniv})$ in each round, that is,
at least $1-\varepsilon$, and thus satisfying the parity conditions $\parobj_{\indqual}$ for $\indqual\in\setqual$ and $\parobj_{\indposi}$ with probability $1-
\varepsilon$. In all the paths where $\sigma_2$ keeps being played, for all $\induniv\in\setuniv$ we have that $\parobj_{\induniv}$ is satisfied since $C^{max}_{even}(\parobj_{\induniv})$ is visited infinitely often. 
For the plays switching to $\sigma_1$ at some point (we have probability $\varepsilon$ to switch to one of these plays at some point), we have $\combi{\setuniv}{\setqual}$. 
This implies that considering all possibilities, we have that all conditions $\parobj_{\indqual}$ are satisfied with probability $1-\varepsilon+\varepsilon=1$, and that all possible plays satisfy all $\parobj_{\induniv}$.
Finally, $\parobj_{\indposi}$ is satisfied with probability $1-\varepsilon$, and hence with probability greater than $0$.
Thus for $s\in C$, we have $\satisfy{s}{\sigma}{\MDP}\combi{\setuniv}{\setqual}\wedge \posi(\parobj_{\indposi})$.
%\qed
\end{proof}

% In the following, we use $\satisfy{s}{\sigma}{\MDP}\posi({\parobj_1}\wedge {\parobj_2})$ (or more parity conditions) to denote that there exists a strategy such that $\prob(\{\pi \in \paths^{\MDPtoMC{\MDP}{\sigma}}(s)\nd \pi$ satisfies ${\parobj_1}\wedge {\parobj_2}\})>0$. This is different from $\satisfy{s}{\sigma}{\MDP}\posi({\parobj_1})\wedge \posi({\parobj_2})$.  
% The formula $\posi({\parobj_1}\wedge {\parobj_2})$ cannot be expressed in \QPL, but is useful in the proof.

We now relate the existence of a \typethree end component with the $\limwedgeone{\induniv\in\setuniv}\univ(\parobj_{\induniv})$ objective.
% Note that in an end-component $C$ in an MDP $\Gamma$, for all states $s \in C$, if $\satisfy{s}{\sigma}{}\posi(\parobj_1 \wedge \parobj_2 \wedge \parobj_3)$, then an infinite path in the MC $\MDPtoMC{(\MDP{\downharpoonright C})}{\sigma}$ that visits a state $s'$ satisfying a parity condition $p_i$ for some $i \in [1,3]$ with non-zero probability, then the path indeed visits $s'$ infinitely often almost surely.
We begin with the following observation.
\begin{lemma} \label{lem:probToAS}
Given an end-component $C$ in an MDP $\MDP$, a set of parity conditions $\{\parobj_{x}\pipe x\in X\}$, for all states $s \in C$, we have $\satisfies{s}{\MDP{\downharpoonright C}}\qual(\limwedgeone{x\in X}\parobj_{x})$ iff $\satisfies{s}{\MDP{\downharpoonright C}}\posi(\limwedgeone{x\in X}\parobj_{x})$ .
\end{lemma}
% 
% The difficult part of the proof is the right to left impliciation. It uses the fact that Proposition~\ref{prop:long_run} implies that the set of states visited infinitely often forms an end-component with probability $1$. From there, we find a sub-EC of $C$ such that in this sub-EC for all $x\in X$, condition $\parobj_x$ has even maximum priority. As there exists a strategy to visit all the states of this sub-EC almost-surely, we have $\satisfies{s}{\MDP{\downharpoonright C}}\qual(\limwedgeone{x\in X}\parobj_{x})$. The detailed proof of the lemma can be found in Appendix~\ref{app:typethree}.
\begin{proof}
The left to right implication is obvious. For the right to left implication, consider $\sigma$ such that $\satisfy{s}{\sigma}{\MDP{\downharpoonright C}}\posi(\limwedgeone{x\in X}\parobj_{x})$. By Proposition~\ref{prop:long_run}, when playing strategy $\sigma$ from $s$ there is probability $1$ that the set of states visited infinitely often is an end-component. We consider sets of states, such that every set forms an end-component with non-zero probability when playing  $\sigma$ from $s$. We call this set of sets of states $\mathcal{C}$. Formally $\mathcal{C}= \{D\mid \prob(\{\pi \in \paths^{\MDPtoMC{\MDP}{\sigma}}(s) \mid  \inf(\proj{S}(\pi))=D\})> 0$. There exists at least one sub-EC $D_X\in\mathcal{C}$ of $C$ such that for all $x\in X$, we have $\max(\{\parobj_x(s)\mid s\in D_X\})$ is even. We define a strategy $\sigma'$ as follows: since for reachability objective in an $MDP$ a deterministic memoryless strategy suffices, we use such a deterministic memoryless strategy inside $C$ until we reach a state $s\in D_X$; then we play a uniform randomized strategy that has probability $1$ of visiting all states of $D_X$ while staying inside $D_X$. We have $\satisfy{s}{\sigma'}{\MDP{\downharpoonright C}}\qual(\limwedgeone{x\in X}\parobj_{x})$, and the result follows.
%\qed
\end{proof}

Now we state the main result of this section that relates the existence of a \typethree end component to the objective $\combi{\setuniv}{\setqual}\wedge \posi(\parobj_{\indposi})$.

\begin{lemma}
\label{lem:SASPr}
Given an MDP $\MDP$, two sets of parity conditions $\{\parobj_{\induniv}\pipe \induniv\in\setuniv\}$, $\{\parobj_{\indqual}\pipe \indqual\in\setqual\}$, and another parity condition $\parobj_{\indposi}$, for all states $s$, we have $\satisfies{s}{\MDP}\combi{\setuniv}{\setqual}\wedge \posi(\parobj_{\indposi})$ iff $\satisfies{s}{\MDP}\combi{\setuniv}{\setqual}\wedge \posi({\event \mathcal{T}_{\mathbf{III},\setuniv,\setqual,\parobj_{\indposi}}})$ .
\end{lemma}
\begin{proof}
We begin with the right to left implication.
Consider a strategy $\sigma_s$ such that $\satisfies{s, \sigma_s}{\MDP}\combi{\setuniv}{\setqual}\wedge \posi({\event \mathcal{T}_{\mathbf{III},\setuniv,\setqual,\parobj_{\indposi}}})$. We show below that $\satisfies{s}{\MDP}\combi{\setuniv}{\setqual}\wedge\posi(\parobj_{\indposi})$. First note that for all states $q\in \mathcal{T}_{\mathbf{III},\setuniv,\setqual,\parobj_{\indposi}}$, we consider a strategy $\sigma_{q}$ such that $\satisfy{q}{\sigma_q}{\MDP}\combi{\setuniv}{\setqual}\wedge \posi(\parobj_{\indposi})$ (we know that such a strategy always exists in a \typethree EC, thanks to Lemma~\ref{lem:aqp-ec}).
Now we construct a strategy $\sigma$, such that $\satisfy{s}{\sigma}{\MDP}\combi{\setuniv}{\setqual}\wedge \posi(\parobj_{\indposi})$. Strategy $\sigma$ is defined from $\sigma_s$ and $\sigma_q$ as follows: 
We play $\sigma_s$ from $s$.
If we reach a state $q$ belonging to a \typethree EC, we play $\sigma_q$ from $q$ forever. 
Since such a state $q$ can be reached with non-zero probability, and strategy $\sigma_q$ satisfies $\parobj_{\indposi}$ with non-zero probability, we have that $\sigma$ satisfies $\parobj_{\indposi}$ with non-zero probability.
As both $\sigma_s$ and $\sigma_q$ satisfy $\combi{\setuniv}{\setqual}$, and thus, while following $\sigma$, since every path in the corresponding MC ends up either following $\sigma_s$ forever or following $\sigma_q$ forever, we have that $\sigma$ also satisfies $\combi{\setuniv}{\setqual}$.

% F, assume that there is a strategy $\sigma$ from $s$ that satisfies  $\combi{\setuniv}{\setqual}\wedge \posi(\parobj_{\indposi})$ in $\MDP$.

Now for the left to right implication, let $\sigma$ be a strategy such that $\satisfy{s}{\sigma}{\MDP}\combi{\setuniv}{\setqual}\wedge \posi(\parobj_{\indposi})
$.
It is easy to see that $\satisfy{s}{\sigma}{\MDP}\posi(\limwedgeone{\induniv\in\setuniv}(\parobj_{\induniv})\wedge \limwedgeone{\indqual\in\setqual}(\parobj_{\indqual})\wedge (\parobj_{\indposi}))
$.
From Proposition~\ref{prop:long_run}, there is probability $1$ that an  infinite path ends up in an end-component. 
Hence in the Markov chain $\MDP^{[\sigma]}$ there is a non-zero probability that an infinite path will reach an end-component $C$ such that for all states $s'\in C$, we have $\satisfy{s'}{\sigma}{\MDP{\downharpoonright C}}\posi(\limwedgeone{\induniv\in\setuniv}(\parobj_{\induniv})\wedge \limwedgeone{\indqual\in\setqual}(\parobj_{\indqual})\wedge (\parobj_{\indposi}))$.
From Lemma~\ref{lem:probToAS}, we thus have that for all $s'\in C$, there exists $\sigma'$ such that $\satisfy{s'}{\sigma'}{\MDP{\downharpoonright C}}\qual(\limwedgeone{\induniv\in\setuniv}(\parobj_{\induniv})\wedge \limwedgeone{\indqual\in\setqual}(\parobj_{\indqual})\wedge (\parobj_{\indposi}))$.
Thus condition $\mathbf{(III_2)}$ is satisfied.
As we consider a pruned MDP thanks to Assumption~\ref{as:pruned}, for all $s'\in C$ we have that $\satisfies{s'}{\MDP} \combi{\setuniv}{\setqual}$.

Thus $C$ is a \typethree EC that can be reached from $s$ with non-zero probability, and thus $\satisfy{s}{\sigma}{\MDP}\combi{\setuniv}{\setqual}\wedge \\ \posi({\event \mathcal{T}_{\mathbf{III},\setuniv,\setqual,\parobj_{\indposi}}})$.
\end{proof}

\section{Formulas with multiple Non-Zero and multiple Exists}
\label{sec:nz_e}
In this section, we discuss how to compute strategies for formulas that consist of several sure parity objectives, several almost-sure parity objectives, several non-zero parity objectives, and several existential parity objectives.
We show in Lemma~\ref{lem:2ep} that such a formula can be split into sub-formulas having a single non-zero or a single existential parity objective.
Further, we show in Lemma~\ref{lem:r_to_par} that a single non-zero parity objective can be transformed into an existential parity objective. 
We finally show in Lemma~\ref{lem:allconstraints2} how to check the satisfiability of a formula that consists of several sure parity objectives, several almost-sure parity objectives, and one existential parity objective.

% As far as we know these results are new.\todo{a bit more details here}
\begin{lemma} \label{lem:2ep}
Given an MDP $\MDP$, a state $s$, four sets of parity conditions 
$\{\parobj_{\induniv}\pipe \induniv\in\setuniv\}$, $\{\parobj_{\indqual}\pipe \indqual\in\setqual\}$, 
$\{\parobj_{\indposi}\pipe \indposi\in\setposi\}$, and $\{\parobj_{\indexis}\pipe \indexis\in\setexis\}$, the following holds: 
$\satisfies{s}{\MDP}\combi{\setuniv}{\setqual}
\wedge 
\limwedgeone{\indposi\in\setposi}\posi(\parobj_{\indposi})
\wedge 
\limwedgeone{\indexis\in\setexis}\exis(\parobj_{\indexis})$ 
iff for all $\indposi\in\setposi$ we have $\satisfies{s}{\MDP}\combi{\setuniv}{\setqual}
\wedge 
\posi(\parobj_{\indposi})$, and for all $\indexis\in\setexis$, we have $\satisfies{s}{\MDP}\combi{\setuniv}{\setqual}
\wedge 
\exis(\parobj_{\indexis})$
\end{lemma}
% e same as that of Lemma~\ref{lem:allconstraints}. 
\begin{proof}
As the left to right implication is obvious, we prove here the other direction.
For $i\in\setposi\cup\setexis$, we consider a strategy $\sigma_i$ such that $\satisfy{s}{\sigma_i}{\MDP}\combi{\setuniv}{\setqual}\wedge Q_i(\parobj_{i})$ for the appropriate $Q_i\in \{\posi,\exis\}$. Now we construct a randomized strategy $\sigma$ given all $\sigma_i$ in which each $\sigma_i$ is chosen uniformly, that is with equal probability. Clearly, $\satisfies{s}{\MDP}\combi{\setuniv}{\setqual}
\wedge 
\limwedgeone{\indposi\in\setposi}\posi(\parobj_{\indposi})
\wedge 
\limwedgeone{\indexis\in\setexis}\exis(\parobj_{\indexis})$, 
and hence the result.
%\qed
\end{proof}

We now show that a non-zero objective can be replaced with an existential parity objective.
Towards this, we first observe the following.
\begin{proposition}
\label{prop:reachpar}
Every reachability condition can be translated to a parity condition.
\end{proposition}
\begin{proof}
Given an MDP $\MDP=(S,E,Act,\prob)$, we construct an MDP $\MDP'=(S',E',Act,\prob')$ such that $S'= S \times \{1,2\}$ where intuitively $\MDP'$ consists of two copies of $\MDP$ with state space $S \times \{1\}$ and $S \times \{2\}$ respectively, and the reachability condition being satisfied corresponds to moving from the first copy to the second copy, and staying there forever. Formally, we consider the parity condition $\parobj$ such that for $s'\in S'$, with $s'=(s,i)$ for $s\in S$ and $i\in\{1,2\}$, we have $\parobj(s')=i$.
\end{proof}
We use this $\MDP'$ in the following two lemmas.

\begin{lemma}
\label{lem:r_to_par}
Given an MDP $\MDP$, two sets of parity conditions $\{\parobj_{\induniv}\pipe \induniv\in\setuniv\}$ and $\{\parobj_{\indqual}\pipe \indqual\in\setqual\}$, and a reachability set $R$, there exists a state $s'$ in $\MDP'$, and a parity condition denoted $\parobj_{\event R}$ such that for all states $s$, we have $\satisfies{s}{\MDP}\combi{\setuniv}{\setqual}\wedge \posi(\event R)$ iff $\satisfies{s'}{\MDP'}\combi{\setuniv}{\setqual}\wedge \exis(\parobj_{\event R})$.
\end{lemma}
% We now give the proof of Lemma~\ref{lem:r_to_par}.
% The result comes from noticing that a non-zero reachability objective and an existential reachability objective are equivalent. It is then enough to use Proposition~\ref{prop:reachpar} to get Lemma~\ref{lem:r_to_par}. The detailed proof of the lemma appears in Appendix~\ref{app:nz_e}.
\begin{proof}
Let $\satisfies{s}{\MDP}\combi{\setuniv}{\setqual}\wedge \posi(\event R)$. As a reachability condition is satisfied for $\posi$ iff it is satisfied for $\exis$, we want to satisfy $\satisfies{s}{\MDP}\combi{\setuniv}{\setqual}\wedge \exis(\event R)$. We now use Proposition~\ref{prop:reachpar} and denote $\parobj_{\event R}$ the parity condition associated to $\event R$, we get that it is necessary and sufficient to find a strategy for $\satisfies{s'}{\MDP'}\combi{\setuniv}{\setqual}\wedge \exis(\parobj_{\event R})$
%\qed
\end{proof}

We get from Lemma~\ref{lem:SASPr} and Lemma~\ref{lem:r_to_par} the following:
\begin{lemma}
\label{lem:SASP_to_SASE}
Given an MDP $\MDP$, and two sets of parity conditions $\{\parobj_{\induniv}\pipe \induniv\in\setuniv\}$ and $\{\parobj_{\indqual}\pipe \indqual\in\setqual\}$, and a parity condition $\parobj_{\indposi}$, for all states $s$, we have $\satisfies{s}{\MDP}\combi{\setuniv}{\setqual}\wedge \posi(\parobj_{\indposi})$ iff there exists $s'$ in $\MDP'$ such that $\satisfies{s'}{\MDP'}\combi{\setuniv}{\setqual}\wedge \exis(\parobj_{\event \mathcal{T}_{\mathbf{III},\indposi}})$, where $\mathcal{T}_{\mathbf{III},\indposi}$ is defined w.r.t. $\MDP'$.
\end{lemma}

Now, since by Assumption~\ref{as:pruned} we have removed all the states that do not satisfy $\combi{\setuniv}{\setqual}$, we have the following:
\begin{lemma} \label{lem:allconstraints2}
Given an MDP $\MDP$, a state $s$, two sets of parity conditions $\{\parobj_{\induniv}\pipe \induniv\in\setuniv\}$ and $\{\parobj_{\indqual}\pipe \indqual\in\setqual\}$, and another parity condition $\parobj$, we have that $\satisfies{s}{\MDP}\combi{\setuniv}{\setqual} \wedge \exis(\parobj)$ iff $\satisfies{s}{\MDP} \exis(\limwedgeone{\induniv\in\setuniv}\parobj_{\induniv} \wedge \parobj)$.
\end{lemma}
\begin{proof}
The left to right result is obvious as $\limwedgeone{\induniv\in\setuniv}\univ(\parobj_{\induniv})\wedge\exis(\parobj)$ implies $\exis(\limwedgeone{\induniv\in\setuniv}\parobj_{\induniv}\wedge \parobj)$.

For the other direction, consider a strategy $\widehat{\sigma}$ such that $\satisfy{s}{\widehat{\sigma}}{\MDP} \exis(\limwedgeone{\induniv\in\setuniv}\parobj_{\induniv}\wedge \parobj)$. 
% Note that for all MDP $\MDP'$ checking if $\satisfies{s}{\MDP'} \exis(\parobj_1 \wedge \parobj_4)$ holds can be solved by considering $\MDP'$ as a word automaton $A$ with two parity conditions $\parobj_1$ and $\parobj_4$ that are in conjunction and checking non-emptiness of this automaton.
Now, a conjunction of parity conditions is a Streett condition \cite{CHP07},
and non-emptiness problem of Streett automaton is decidable.

We noted in Section~\ref{sec:prelims} that satisfying existentially is the same as finding a satisfying path in the nondeterministic automaton. This means that if the Streett automaton is non-empty, then there exists a finite-memory strategy $\sigma$ (linear in the indices of $\parobj_{\induniv}$ and $\parobj$) in $\MDP$ such that there exists a path $\pi$ in the MC $\MDP^{[\sigma]}$ satisfying both $\parobj$ and all $\parobj_{\induniv}$ for $\induniv\in\setuniv$.
% and solving the non-emptiness problem of a Streett automaton is polynomial in the indices of $\parobj_1$ and $\parobj_2$. If this automaton $A$ is nonempty, there exists a finite Moore machine $M$ such that the path given by $A\times M$ is exactly an accepting run, where $A$ gives the current state and $M$ the action that must be taken. It is possible to see whether a run of the initial MDP $\MDP'$ is following $A\times M$. We can assume we took $\sigma_1$ such that a Moore machine $M$ exists. It means $\sigma_1$ is a finite memory strategy (linear in the indices of $\parobj_1$ and $\parobj_2$). 

Now by assumption, there exists a strategy $\sigma_{\wedge}$ such that $\satisfy{s}{\sigma_{\wedge}}{\MDP}\combi{\setuniv}{\setqual}$. 
A strategy $\sigma$ such that $\satisfy{s}{\sigma}{\MDP}\combi{\setuniv}{\setqual} \\ \wedge\exis(\parobj)$ is obtained below by combining $\widehat{\sigma}$ and $\sigma_{\wedge}$ as follows.
% The strategy $\sigma$ is constructed as follows. 
At each step a coin is tossed. 
If it gives head, then we play $\sigma_{\wedge}$ forever. 
Otherwise, we play this step as specified by the strategy $\widehat{\sigma}$.
% described as a product $A\times M$. 
If this results in deviating from the path $\pi$,
% visiting a state different from what was expected on $A\times M$ 
then we play $\sigma_{\wedge}$ forever, else repeat the same process.

Strategy $\sigma$ has a path ensuring $\parobj$: the one where we always follow $\widehat{\sigma}$, that happens when all the coin tosses give tails, and the state randomly taken in the MDP is always in $\pi$.
% the same as what is expected on $A\times M$). 
The probability of switching to $\sigma_{\wedge}$ some time is $1$, thus satisfying $\limwedgeone{\indqual\in\setqual}\qual(\parobj_{\indqual})$. 
If we follow the path
% defined by $A\times M$
$\pi$, then for all ${\induniv\in\setuniv}$ we have that $\parobj_{\induniv}$ is satisfied, otherwise we switch to $\sigma_{\wedge}$ at some point, and for all ${\induniv\in\setuniv}$ we again have that $\parobj_{\induniv}$ is satisfied, thus ensuring $\limwedgeone{\induniv\in\setuniv}\univ(\parobj_{\induniv})$.
%\qed
\end{proof}

This concludes the decidability proof of the realizability of the negation and disjunction-free fragment of \QPL. Indeed, given a formula 
$\satisfies{s}{\MDP}\combi{\setuniv}{\setqual}
\wedge 
\limwedgeone{\indposi\in\setposi}\posi(\parobj_{\indposi})
\wedge 
\limwedgeone{\indexis\in\setexis}\exis(\parobj_{\indexis})$, we use Lemma~\ref{lem:2ep} to split it into formulas of the form $\satisfies{s}{\MDP}\combi{\setuniv}{\setqual}
\wedge 
\posi(\parobj_{\indposi})$, and of the form $\satisfies{s}{\MDP}\combi{\setuniv}{\setqual}
\wedge 
\exis(\parobj_{\indexis})$. We then apply Lemma~\ref{lem:SASPr} on formulas of the form $\satisfies{s}{\MDP}\combi{\setuniv}{\setqual}
\wedge 
\posi(\parobj_{\indposi})$, and then use Lemma~\ref{lem:SASP_to_SASE} to transform the non-zero objective into an existential objective. For formulas of the form  $\satisfies{s}{\MDP}\combi{\setuniv}{\setqual}
\wedge 
\exis(\parobj_{\indexis})$, we use Lemma~\ref{lem:allconstraints2}.

By Remark~\ref{rem:negation} it shows the decidability of the \QPL-realizability problem. 
We state the complexity of this realizability problem in the next Section.

% \stam{
% \begin{remark}\label{rmk:exis}
% Note that instead of checking whether $\satisfies{s}{\MDP} \exis(\parobj_1 \wedge \parobj_4)$ in Lemma \ref{lem:allconstraints2}, checking if $\satisfies{s}{\MDP} \exis(\parobj_1) \wedge \exis(\parobj_4)$ is not sufficient, as seen again on Figure~\ref{fig:CEposiexis} where $\parobj_1$ is the first value of each node, and $\parobj_2$ the second one. Tossing a coin and having probability half to go right and half left satisfies $\exis(\parobj_1) \wedge \exis(\parobj_4)$ but no strategy can satisfy $\exis(\parobj_1 \wedge \parobj_4)$.
% \end{remark}

% Checking if $\satisfies{s}{\MDP} \exis(\parobj_1 \wedge \parobj_4)$ holds can be solved in $O()$ time. \track{complete}
% }

% Now, we construct a strategy $\sigma$ for the condition $\univ(\parobj_1)\wedge \qual(\parobj_2)\wedge \posi(\parobj_3) \wedge \exis(\parobj_4)$ from the strategies $\sigma_1$ and $\sigma_2$ for the conditions $\univ(\parobj_1)\wedge \qual(\parobj_2)\wedge \posi(\parobj_3)$ and $\univ(\parobj_1)\wedge \qual(\parobj_2) \wedge \exis(\parobj_4)$ respectively.

\section{Complexity results}
\label{sec:complex}
In this section, we analyze the complexity of deciding the existence of strategies to satisfy \QPL formulas.
Recall from the results of Sections  \ref{sec:typethree} and \ref{sec:nz_e}, that in order to find strategies for formulas of the form $\limwedgeone{\induniv\in\setuniv}\univ(\parobj_{\induniv}) \wedge \limwedgeone{\indqual\in\setqual}\qual(\parobj_{\indqual})\wedge \limwedgeone{\indposi\in\setposi}\posi(\parobj_{\indposi}) \wedge \limwedgeone{\indexis\in\setexis}\exis(\parobj_{\indexis})$, we need to compute the set of maximal \typethree ECs. 
In particular, we need these ECs for subformulas of the form $\limwedgeone{\induniv\in\setuniv}\univ(\parobj_{\induniv}) \wedge \limwedgeone{\indqual\in\setqual}\qual(\parobj_{\indqual})\wedge \limwedgeone{\indposi\in\setposi}\posi(\parobj_{\indposi})$.
\begin{algorithm}
\caption{}
\label{proc:all}
\begin{flushleft}
\textbf{Input} : An MDP $\MDP$, $s$ a state of $\MDP$, parity conditions $\{\parobj_{\induniv}\pipe \induniv\in\setuniv\}$, $\{\parobj_{\indqual}\pipe \indqual\in\setqual\}$, $\{\parobj_{\indposi}\pipe \indposi\in\setposi\}$, and $\{\parobj_{\indexis}\pipe \indexis\in\setexis\}$. \\
\textbf{Output} : 
% Do we have
true if 
$\satisfies{s}{\MDP}\limwedgeone{\induniv\in\setuniv}\univ(\parobj_{\induniv})\limwedgeone{\indqual\in\setqual}\qual(\parobj_{\indqual})\wedge \limwedgeone{\indposi\in\setposi}\posi(\parobj_{\indposi}) \wedge \limwedgeone{\indexis\in\setexis}\exis(\parobj_{\indexis})$\\
else false.
\end{flushleft}
\begin{algorithmic}[1]
% \Procedure{}{}
    \State \label{proc-sat-1}Compute the set of maximal \typeone($\setuniv,\setuniv$) ECs $C$ such that $\forall\, s'\in C,\, \satisfies{s'}{\MDP{\downharpoonright C}} \limwedgeone{\induniv\in\setuniv}\univ(\parobj_{\induniv}) \wedge \limwedgeone{\induniv\in\setuniv} \qual(\event C^{\max}_{\even}(\parobj_{\induniv}))$. \Comment{By Lemma~\ref{lem:alg-typeone}.}
    \State \label{proc-sat-2}Compute the set of maximal \typetwo($\setuniv,\setqual$) ECs $C$:  $C$ is a maximal \typeone($\setuniv,\setqual$) EC and $\forall\, s'\in C,\, \satisfies{s'}{\MDP{\downharpoonright C}} \limwedgeone{\induniv\in\setuniv}\qual(\parobj_{\induniv})\wedge \limwedgeone{\indqual\in\setqual}\qual(\parobj_{\indqual})$.
    \State \label{proc-sat-3}Compute the set $\mathcal{S}_1$ of states $s'$ such that $\satisfies{s'}{\MDP} \limwedgeone{\induniv\in\setuniv}\univ(\parobj_{\induniv}) \wedge \qual(\event \mathcal{T}_{\mathbf{II},\MDP,\setuniv,\setqual})$. \Comment{Correct by Lemma~\ref{lem:SAS}.}
    %  $\ugec$
    %  $\nec$
    % $\sgec$
    \newline
    If $s\nin \mathcal{S}_1$ then return false.
    \State \label{proc-sat-4}Compute $\MDP{\downharpoonright \mathcal{S}_1}$ where all the states that do not satisfy $\combi{\setuniv}{\setqual}$ have been pruned. \Comment{Correct by Lemma~\ref{pr-closed-a-q}.}
    \For {\texttt{all $\indposi\in\setposi$}} 
    \State \label{proc-sat-5} Compute the set of maximal \typethree($\setuniv,\setqual,\parobj_{\indposi}$) ECs $C$ of $\MDP{\downharpoonright \mathcal{S}_1}$: $\forall\, s\in C$, we have that $\satisfies{s}{(\MDP\downharpoonright \mathcal{S}_1){\downharpoonright C}} \limwedgeone{\induniv\in\setuniv}\qual(\parobj_{\induniv})\limwedgeone{\indqual\in\setqual}\qual(\parobj_{\indqual})\wedge \qual(\parobj_{\indposi})$.  \Comment{By Lemma~\ref{lem:SASPr}.}
    \State \label{proc-sat-6}Compute the parity condition $\parobj_{\event \mathcal{T}_{\mathbf{III},\setuniv,\setqual,\parobj_{\indposi}}}$ and the MDP $\MDP'$ with set $\mathcal{S}'_1$ of states 
    % (By Proposition~\ref{prop:reachpar}).
    % \rbchanged{These primed objects are defined in this Proposition}
    \Comment{$\MDP'$ and $\mathcal{S}'_1$ are defined in Proposition~\ref{prop:reachpar}.}
    \State \label{proc-sat-8} Check if $\satisfies{(s,1)}{\mathcal{S}'_1} \exis(\limwedgeone{\induniv\in\setuniv}\parobj_{\induniv}\wedge\parobj_{\event \mathcal{T}_{\mathbf{III},\setuniv,\setqual,\parobj_{\indposi}}})$. \Comment{By lemmas \ref{lem:SASP_to_SASE} and \ref{lem:allconstraints2}; ($s,1$) is defined in Proposition~\ref{prop:reachpar}.}
    \EndFor
    \For {\texttt{all $\indexis\in\setexis$}}
    \State \label{proc-sat-7}Check if $\satisfies{s}{\mathcal{S}_1} \exis(\limwedgeone{\induniv\in\setuniv}\parobj_{\induniv}\wedge\parobj_{\indexis})$. \Comment{By Lemma \ref{lem:allconstraints2}.}
    \EndFor
    \State If any of the checks in Steps~\ref{proc-sat-8} and ~\ref{proc-sat-7} fails then return \texttt{false}, else return \texttt{true}.
    % \State \label{proc-sat-5}For all $\indposi\in\setposi$, Compute the maximal new end-components $\mathcal{N}_{\indposi}$ of $\MDP{\downharpoonright \mathcal{S}_1}$: $C\in \mathcal{N}$ iff $\forall\, s\in C,\, \satisfies{s}{\MDP{\downharpoonright C}} \limwedgeone{\induniv\in\setuniv}\qual(\parobj_{\induniv})\limwedgeone{\indqual\in\setqual}\qual(\parobj_{\indqual})\wedge \qual(\parobj_{\indposi})$. 
    % \State\label{proc-sat-6}For all $\indposi\in\setposi$ compute the parity condition $\parobj_{\event \mathcal{T}_{\mathbf{III},\indposi}}$.
    % \State \label{proc-sat-7}For all $\indexis\in \setexis$ check if $\satisfies{s}{\mathcal{S}_1} \exis(\limwedgeone{\induniv\in\setuniv}\parobj_{\induniv}\wedge\parobj_{\indexis})$
    % \State \label{proc-sat-8}For all ${\indposi}\in \setposi$ check if $\satisfies{s}{\mathcal{S}_1} \exis(\limwedgeone{\induniv\in\setuniv}\parobj_{\induniv}\wedge\parobj_{Reach,\indposi})$.
    % \EndProcedure
\end{algorithmic}
\end{algorithm}
This in turn requires solving several Streett games in general (Lemma \ref{lem:alg-typeone}).
The procedure is described in Algorithm \ref{proc:all}.
We show that the algorithm runs in time $\Sigma_2^{\sf{P}}={\sf P}^{\sf{NP}}$ (Theorem \ref{thm:multi-parity}).
We also show that we have a polynomial algorithm for the special case where the set $\setuniv$ is empty (Theorem~\ref{thm:poly}), and that randomization and finite memory are required (Theorem~\ref{thm:poly_space}). The problem is in ${\sf NP} \cap {\sf CoNP}$ when $\setuniv$ is singleton (Theorem \ref{thm:NPcapCoNP}).
Finally, we show that finding a strategy for an arbitrary \QPL formula that may have combination of conjunctions and disjunctions is in ${\Sigma}^{\sf P}_2 (={\sf NP}^{\sf NP})$ (Theorem \ref{thm:Sigma2P}), and it is both ${\sf NP}$-hard and ${\sf coNP}$-hard (Theorem \ref{thm:hardness}).

\begin{theorem}
\label{thm:multi-parity}
% \sgcomment{We should move this after the description of Algorithm 1.}
Given an MDP $\MDP$, and a state $s$, Algorithm~\ref{proc:all} decides if $\satisfies{s}{\MDP}\limwedgeone{\induniv\in\setuniv}\univ(\parobj_{\induniv})\wedge\limwedgeone{\indqual\in\setqual}\qual(\parobj_{\indqual})\wedge \limwedgeone{\indposi\in\setposi}\posi(\parobj_{\indposi}) \wedge \limwedgeone{\indexis\in\setexis}\exis(\parobj_{\indexis})$, and solves the problem in $\sf{P}^{\sf{NP}}$ time.
\end{theorem}
\begin{proof}
For the correctness of Algorithm~\ref{proc:all}, consider a formula $\varphi=\limwedgeone{\induniv\in\setuniv}\univ(\parobj_{\induniv}) \wedge \limwedgeone{\indqual\in\setqual}\qual(\parobj_{\indqual})\wedge \limwedgeone{\indposi\in\setposi}\posi(\parobj_{\indposi}) \wedge \limwedgeone{\indexis\in\setexis}\exis(\parobj_{\indexis})$.
Given $\varphi$, we explain below how to check its satisfiability. 
% We then give a complete procedure associated to this check, doing it in a bottom-up way. We relate both of these approaches during the explanations.
% To find whether this holds, we use that by Theorem~\ref{lem:2ep} we only have to check that 

According to Lemma~\ref{lem:2ep}, to check if $s \models_{\Gamma} \varphi$, we need to check for all $\indposi\in\setposi$, if $\satisfies{s}{\MDP}\combi{\setuniv}{\setqual}\wedge \posi(\parobj_{\indposi})$ and that for all $\indexis\in\setexis$ if $\satisfies{s}{\MDP}\combi{\setuniv}{\setqual}\wedge \exis(\parobj_{\indexis})$.
% See Lemma~\ref{lem:2ep} in the previous section for a formal statement.

For the non-zero ($\posi$) part, recall from Section~\ref{sec:nz_e}, that for each $\parobj_{\indposi}$, we need to compute the set of all \typethree ECs for the objective $\combi{\setuniv}{\setqual}\wedge \posi(\parobj_{\indposi})$. 
Recall that each $\typethree$ EC $C$ is such that two properties hold. First $\forall\, s'\in C$, we have $\satisfies{s'}{\MDP} \combi{\setuniv}{\setqual}$;
% (this will be done at step \ref{proc-sat-3} of the procedure) 
we do this in Step \ref{proc-sat-3} of the algorithm.
Then we also check that $\forall\, s'\in C,\, \satisfies{s'}{\MDP{\downharpoonright C}} \limwedgeone{\induniv\in\setuniv}\qual(\parobj_{\induniv})\limwedgeone{\indqual\in\setqual}\qual(\parobj_{\indqual})\\ \wedge \qual(\parobj_{\indposi})$. This is done in Step \ref{proc-sat-5}. 
% which is the same as the following point:

Now for an existential parity objective $\exis(\parobj_{\indexis})$,
% , let's call $\parobj_{\indexis}$ the existential condition.
% we want to check (it may be optain by the set $\setexis$, or by some $\parobj_{Reach,\indposi}$), by Lemma~\ref{lem:allconstraints1} and 
by Lemma~\ref{lem:allconstraints2}, we only need to check if: $\satisfies{s}{\MDP{\downharpoonright S_{\llbracket\combi{\setuniv}{\setqual}\rrbracket}}} \exis(\limwedgeone{\induniv\in\setuniv}\parobj_{\induniv}\wedge\parobj_{\indexis})$
It remains to explain how the sub-MDP in which all states satisfy the formula $\limwedgeone{\induniv\in\setuniv}\univ(\parobj_{\induniv})\limwedgeone{\indqual\in\setqual}\qual(\parobj_{\indqual})$ is computed. This is done in Steps  \ref{proc-sat-1} to \ref{proc-sat-4}. 

To do so we first find the set of states $s'$ such that $\satisfies{s'}{\MDP}\limwedgeone{\induniv\in\setuniv}\univ(\parobj_{\induniv})\limwedgeone{\indqual\in\setqual}\qual(\parobj_{\indqual})$. 
% this is the same problem we studied in Section~\ref{app:morecond}. 
% We split this problem into two.
We find the set of maximal \typetwo ECs.
By definition of \typeone and \typetwo ECs, a maximal \typeone EC that satisfies $\mathbf{(II_2)}$ is a maximal \typetwo EC. We now prove that if $C$ is a maximal \typetwo EC, it is included in $C_1$, a maximal \typeone EC. Since for all $s'\in C$, we have $\satisfies{s'}{\MDP{\downharpoonright C}} \limwedgeone{\induniv\in\setuniv}\qual(\parobj_{\induniv})\wedge \limwedgeone{\indqual\in\setqual}\qual(\parobj_{\indqual})$, we also have that for all $s'\in C_1,\, \satisfies{s'}{\MDP{\downharpoonright C_1}} \limwedgeone{\induniv\in\setuniv}\qual(\parobj_{\induniv})\wedge \limwedgeone{\indqual\in\setqual}\qual(\parobj_{\indqual})$ (it suffices to take a strategy that has probability $1$ of reaching some state $s''\in C$ and then play the strategy ensuring $\satisfies{s''}{\MDP{\downharpoonright C}} \limwedgeone{\induniv\in\setuniv}\qual(\parobj_{\induniv})\wedge \limwedgeone{\indqual\in\setqual}\qual(\parobj_{\indqual})$). Hence $C_1$ is a \typetwo EC. By maximality of $C$, we have $C=C_1$. This means that to find the maximal \typetwo EC, it is sufficient to compute the maximal \typeone EC and remove those maximal \typeone ECs that do not ensure condition $\mathbf{(II_2)}$. Finding the maximal \typeone EC is done in Step \ref{proc-sat-1} thanks to Lemma~\ref{lem:alg-typeone}. We then check condition  $\mathbf{(II_2)}$ in Step \ref{proc-sat-2}.

% The difficult part is the general case for step \ref{proc-sat-1}: computing the set $\mathcal{C}_u$ of maximal \typeone($\setuniv,\setuniv$) ECs, that is such that $\forall\, s'\in C$, we have $\satisfies{s'}{\MDP{\downharpoonright C}} \limwedgeone{\induniv\in\setuniv}\univ(\parobj_{\induniv}) \wedge \limwedgeone{\induniv\in\setuniv} \qual(\event C^{\max}_{\even}(\parobj_{\induniv}))$. 

It then remains to find the states $s'$ such that  $\satisfies{s'}{\MDP} \limwedgeone{\induniv\in\setuniv}\univ(\parobj_{\induniv}) \\ \wedge \qual(\event \mathcal{T}_{II,\MDP,\setuniv,\setuniv})$ (Step \ref{proc-sat-3}). 
It can be done by transforming the MDP into a Streett-B\"uchi game that in turn can be transformed into a Streett game since a B\"uchi condition is a special case of a parity condition, and a conjunction of parity conditions is a Streett condition.

For the complexity, 
Steps \ref{proc-sat-2}, \ref{proc-sat-4}, \ref{proc-sat-5}, \ref{proc-sat-6}, \ref{proc-sat-8} and \ref{proc-sat-7}  are polynomial. Step \ref{proc-sat-3} is parity-complete if there is only one parity condition, polynomial if there is none, and is co-NP complete in general (we have to solve a Streett game). Step \ref{proc-sat-1} is parity-complete if there is only one parity condition, polynomial if there is none. In the general case, Step~\ref{proc-sat-1} requires to iteratively solve a polynomial number of Streett games (which is in \sf{coNP}), and use the result of this computation to remove some of the states, resulting in a $\sf{P^{NP}}$ complexity. Further details of the procedure are given  in algorithms~\ref{proc:issg} to~\ref{proc:sgec} in the proof of Lemma~\ref{lem:alg-typeone}.
% in Appendix~\ref{app:typeone}
 This leads us to a $\sf{P^{NP}}$ complexity for Algorithm~\ref{proc:all}.
    %\qed
\end{proof}

% We discuss below some of the special cases of Proposition \ref{thm:multi-parity}.
\begin{theorem} \label{thm:poly}
Given an MDP $\MDP$, and a state $s$, it is decidable in polynomial time if $\satisfies{s}{\MDP}\limwedgeone{\indqual\in\setqual}\qual(\parobj_{\indqual})\wedge \limwedgeone{\indposi\in\setposi}\posi(\parobj_{\indposi}) \wedge \limwedgeone{\indexis\in\setexis}\exis(\parobj_{\indexis})$.
\end{theorem}
\begin{proof}[Proof]
The proof is similar to that of Theorem~\ref{thm:multi-parity}, but as $\setuniv=\varnothing$ we compute \typeone($\setuniv,\setuniv$) that is \typeone($\varnothing,\varnothing$) ECs, that is any EC, and so Step~\ref{proc-sat-1} returns the maximal ECs. In the same way, Step~\ref{proc-sat-3} polynomially computes one almost-sure reachability problem.
% , as the additional $A(\parobj_{a})$ objectives are removed.
As we consider \typeone($\varnothing,\varnothing$), the checks in Steps 2,3,6,8 and 10 can be done with a randomized finite-memory strategy. Further, a strategy for each of these steps can be computed in polynomial time~\cite{etessami2007multi}.
\end{proof}

In the above proof, we show that randomized finite-memory strategies are sufficient for formulas of the form $\\ \satisfies{s}{\MDP}\limwedgeone{\indqual\in\setqual}\qual(\parobj_{\indqual})\wedge \limwedgeone{\indposi\in\setposi}\posi(\parobj_{\indposi}) \wedge \limwedgeone{\indexis\in\setexis}\exis(\parobj_{\indexis})$. Below, we show that such strategies are indeed necessary. 
% We end this section with the following observation about \QPL formulas without sure parity atoms. 
\begin{theorem} \label{thm:poly_space}
Given an MDP $\MDP$, and a state in $\MDP$,
% and a state $s$, if it holds that
to solve the realizability problem for formulas of the form
% $\satisfies{s}{\MDP}
$\limwedgeone{\indqual\in\setqual}\qual(\parobj_{\indqual})\wedge \limwedgeone{\indposi\in\setposi}\posi(\parobj_{\indposi}) \wedge \limwedgeone{\indexis\in\setexis}\exis(\parobj_{\indexis})$, a randomized finite-memory winning strategy may be necessary.
\end{theorem}

\begin{proof}
% We prove that memory is required for condition $\qual(\parobj_1)\wedge\exis(\parobj_2)$.
We show that the lemma already holds for formulas of the form $\qual(\parobj_1)\wedge\exis(\parobj_2)$.
Consider the MDP in Figure~\ref{fig:mem_needed}. We want a strategy $\sigma$ such that in $\MDP^{[\sigma]}$, we have a subset of $(ab)^*a^{\omega}$ with measure $1$ to satisfy $\qual(\parobj_1)$, and we want at least one path in $(a^+b)^{\omega}$ to satisfy $\exis(\parobj_2)$. We can argue that no randomized memoryless strategy, as well as no deterministic memoryful strategy suffices here. However, we can use the following randomized strategy with one bit of memory: Toss a coin every time the token is in state $(0,1)$. If the coin gives heads, play $ba$, otherwise if the coin gives tails once, play $a$ forever from the state $(0,1)$ (one bit of memory is used to store if the coin has given tails).
\begin{figure}[t]
\centering
\vspace{5pt}
\scalebox{1}{
	\begin{tikzpicture}
		\node[player,,initial,initial text={}] (s0) at (0,0) {$1,2$} ;
		\node[player] (s1) at (3,0) {$0,1$} ;
		
		\path[-latex]  
			(s0) edge[bend left=20] node[above] {\small $a$, 1} (s1)
			(s1) edge[bend left=20] node[below] {\small $b$, 1} (s0)
			(s1) edge [loop above] node[above] {\small $a$, 1} (s1)
		;
	\end{tikzpicture}
}
\caption{An MDP requiring randomized finite-memory strategy for the objective $\qual(\parobj_1)\wedge\exis(\parobj_2)$.}
\label{fig:mem_needed}
\end{figure}
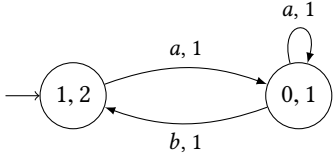
%\qed
\end{proof}
% This example further reiterates that randomization with memory is necessary in our setting as opposed to the previous related works where deterministic strategies suffice.

\begin{theorem}
\label{thm:NPcapCoNP}
Given an MDP $\MDP$, and a state $s$, we can decide in $\sf{NP}\cap\sf{coNP}$ if $\satisfies{s}{\MDP}\univ(\parobj_{\induniv})\wedge \limwedgeone{\indqual\in\setqual}\qual(\parobj_{\indqual})\wedge \limwedgeone{\indposi\in\setposi}\posi(\parobj_{\indposi}) \wedge \limwedgeone{\indexis\in\setexis}\exis(\parobj_{\indexis})$. This is done by solving a polynomial number of parity games.
\end{theorem}
\begin{proof}[Proof]
The proof is similar to that of Theorem~\ref{thm:multi-parity}, but Step~\ref{proc-sat-1} computes maximal \typeone($\{\ \induniv\},\{ \induniv\}$) ECs which can be done in $\sf{P}^{\sf{NP}\cap\sf{coNP}}$~\cite{AKV16}. Also  
% step \ref{proc-sat-3} solves a polynomial number of B\"uchi-parity games. 
% At Step~\ref{proc-sat-1} computing the maximal SGEC is the same problem as in~\cite{AKV16} where it is done 
% This special case of single parity condition corresponding to Step~\ref{proc-sat-1} in Algorithm~\ref{proc:issg} has been solved in~\cite{AKV16}.
% There the set of maximal \typeone($\{\parobj_{\induniv}\},\{\parobj_{\induniv}\}$) ECs has been computed by solving a polynomial number of parity games.
the set of states that ensure the objective in Step~\ref{proc-sat-3} with only one sure parity condition can be computed by solving a polynomial number of parity games~\cite{AKV16,BRR17}, and by pruning states after solving each parity game. The result follows since $\sf{P}^{\sf{NP}\cap\sf{coNP}}=\sf{NP}\cap\sf{coNP}$~\cite{brassard1979note}.
%\qed
\end{proof}

We now turn to the case of Boolean combinations of $\univ(\parobj)$, $\qual(\parobj)$, $\posi(\parobj)$, and $\exis(\parobj)$ atoms.

\begin{theorem}
\label{thm:Sigma2P}
% Given an MDP $\MDP$, a state $s$, and a \QPL formula $\varphi$, 
\QPL realizability is in 
% , it is decidable in 
$\sf{NP}^{\sf{NP}}$
% if $\satisfies{s}{\MDP}\varphi$
.
\end{theorem}
\begin{proof}[Proof]
Assume that we are given an MDP $\MDP$, a state $s$, and a \QPL formula $\varphi$.  
% $\satisfies{s}{\MDP}\varphi$.
We assume that $\varphi$ is in negation-free normal form. A negation-free equivalent formula can be found in polynomial time, it suffices to take the negation normal form (by using De Morgan's law to push negations inwards), and to take the dual of the negated atomic parity conditions.

% We show that given an $\sf{NP}$ and a $\sf{coNP}$ oracle, if the formula is satisfiable, then there exists a polynomial-size witness that can be checked to prove its satisfiability. 
We note that there exists a DNF formula equivalent to $\varphi$ with the same set of atomic objectives. 
If $\varphi$ is satisfiable, then there exists a conjunctive clause $\psi$ that is a conjunction of these atomic objectives, and $\psi$ can be guessed by an $\sf{NP}$ machine such that
% there exists a strategy $\sigma$, and $\satisfy{s}{\sigma}{\MDP}\psi$. 
$\satisfies{s}{\MDP}{\psi}$.
Given this witness $\psi$, we can use an $\sf{NP}$ oracle to check that $\psi$ implies $\varphi$ where the atomic objectives are regarded as classical propositional atoms. 
We can then use the $\sf{P}^{\sf{NP}}$ algorithm of Theorem~\ref{thm:multi-parity} to get a proof that 
% $\psi$ holds on $\MDP$ from state $s$, 
indeed $\satisfies{s}{\MDP}{\psi}$
implying that $\varphi$ is also satisfied by $s$.
The result follows since the class $\sf{NP}^{{\sf{P}^{\sf{NP}}}}$ is the same as $\sf{NP}^{\sf{NP}}$.
%\qed
\end{proof}

\begin{theorem}
\label{thm:hardness}
% Given an MDP $\MDP$, a state $s$, and a \QPL formula $\varphi$, 
\QPL realizability is both $\sf{NP}$-hard and $\sf{coNP}$-hard.
% if $\satisfies{s}{\MDP}\varphi$.
\end{theorem}
\begin{proof}[Proof]
% Let us assume we are given an MDP $\MDP$, a state $s$, and a \QPL formula $\varphi$.  
The $\sf{coNP}$-hardness follows from the fragment of $\QPL$ made of conjunctions of $\univ(\parobj)$ atoms that is powerful enough to encode Streett games, as proved in~\cite{CHP07}.

We prove $\sf{NP}$-hardness for the fragment of \QPL made of $\{\wedge,\vee,\univ(p)\}$. 
% We are not aware if this proof exists in the literature. 
Given a SAT formula $\varphi$ in negation normal form, we define a \QPL formula $\varphi_{obj}$, and an MDP $\MDP$ both polynomial in $\varphi$ such that there exists a state $s$ of $\MDP$ and there exists a winning strategy for $\varphi_{obj}$ from $s$ on $\MDP$ if and only if $\varphi$ is satisfiable. Let $Var=\{a_1,\dots a_n\}$ be the set of propositional variables of $\varphi$. 
Let $S=\{a_i,\bar{a}_i\pipe i\in[n]\}$ be both the set of states and the set of actions. We define the MDP $\MDP=\{S,\{(a,b,b)\pipe a,b\in S\},S,Pr\}$ where for all $a,b\in S$ we have
$Pr(a,b,b)=1$.
% $Pr(a,b,c)=0\texttt{ if }b\neq c\texttt{ and }Pr(a,b,b)=1\}$.
Note that the underlying graph is a complete graph.
For each $i\in [n]$, we define two parity conditions $\parobj_{a_i}$ and  $\parobj_{\bar{a}_i}$ such that $\parobj_{a_i}(a_i)=\parobj_{\bar{a}_i}(\bar{a}_i)=2$, such that $\parobj_{a_i}(\bar{a}_i)=\parobj_{\bar{a}_i}(a_i)=3$ and $\parobj_{a_i}(b)=\parobj_{\bar{a}_i}(b)=1$ if $b\nin\{a_i,\bar{a}_i\}$. 
We define $\psi=\limwedgeone{i\in[1,n]}\univ(\parobj_{a_i})\vee\univ(\parobj_{\bar{a}_i})$. Given the SAT formula $\varphi$ in NNF, we define a \QPL formula $\varphi'$ by transforming $\varphi$ in the following way: each $\neg{a_i}$ is replaced by $\univ(\parobj_{\bar{a}_i})$ and each (non-negated) $a_i$ is replaced by $\univ(\parobj_{a_i})$. 
We define $\varphi_{obj}=\psi\wedge \varphi'$. 
In the MDP $\MDP$, we consider an arbitrary state $s$, a strategy $\sigma$, and the paths in $\paths^{\MDPtoMC{\MDP}{\sigma}}(s)$.
We now show that $\satisfies{s}{\MDP}{\varphi_{obj}}$ iff $\varphi$ is satisfiable.

We first prove the right to left implication. We assume $\varphi$ is satisfiable, by some valuation $\fun{\rho}{[n]}{Var\cup \{\neg a_i\pipe a_i\in Var\}}$ such that $\rho(i)=a_i$ or $\rho(i)=\neg a_i$. 
In the MDP $\MDP$, we consider an arbitrary state $s$, and define a strategy $\sigma$ with a transition system that loops over $n$ modes. At mode $i$, it plays  $a_i$ if $\rho(i)=a_i$, and it plays  $\bar{a_i}$ if $\rho(i)=\neg a_i$. The only path in $\paths^{\MDPtoMC{\MDP}{\sigma}}(s)$ clearly satisfies both $\psi$ and $\varphi'$, hence  $\satisfies{s}{\MDP}{\varphi_{obj}}$.

We now prove the left to right implication. We assume that for some $s\in S$, there exists a strategy $\sigma$ such that $\satisfy{s}{\sigma}{\MDP}{\varphi_{obj}}$. We note that since $\sigma$ may be randomized,  $\paths^{\MDPtoMC{\MDP}{\sigma}}(s)$ may contain more than one path. We take $\pi\in \paths^{\MDPtoMC{\MDP}{\sigma}}(s)$, and consider the following valuation: $\fun{\rho}{[n]}{Var\cup \{\neg a_i\pipe a_i\in Var\}}$ such that for all $i\in[n]$, we have $\rho(i)=a_i$ if $a_i\in\inf(\pi)$ and $\rho(i)=\neg a_i$ if $\bar{a_i}\in\inf(\pi)$. We have that $\rho$ is a well-defined and complete valuation: Since $\varphi_{obj}=\varphi'\wedge \psi$ we have that $\satisfy{s}{\sigma}{\MDP}{\psi}$, and for all $i\in[n]$ either $\rho(i)=a_i$ holds or $\rho(i)=\neg a_i$ holds, but not both. Valuation $\rho$ satisfies $\varphi$ since $\satisfy{s}{\sigma}{\MDP}{\varphi'}$, hence $\varphi$ is satisfiable.
% \sgcomment{The following may go to the appendix due to lack of space.}
% Note that there may be more than one path if $\sigma$ is randomized. 
% Then $\paths^{\MDPtoMC{\MDP}{\sigma}}(s)$ satisfies $\psi$ if and only if for all $i\in[n]$, and for all paths $\pi\in \paths^{\MDPtoMC{\MDP}{\sigma}}(s)$ exactly one of state $a_i$ or state $\bar{a}_i$ belongs to $\inf(\pi)$ (as these are the only states with even valuation for $\parobj_{a_i}$ and $\parobj_{\bar{a}_i}$ respectively), and only one of those two states belongs to $\inf(\pi)$ (as they have high odd priority for the condition associated to the other one). Thus $\psi$ is satisfied if and only if in every path the states seen infinitely often yield a consistent valuation. 
% On the other hand, $\paths^{\MDPtoMC{\MDP}{\sigma}}(s)$ satisfies $\varphi'$ if and only if for all paths $\pi\in \paths^{\MDPtoMC{\MDP}{\sigma}}(s)$, we have that $\inf(\pi)$ defines a valuation for the $a_i$'s and $\neg a_i$'s seen as independent variables that makes formula $\varphi$ true. 
% Thus $\varphi_{obj}$ can be satisfied if and only if every $\pi\in \paths^{\MDPtoMC{\MDP}{\sigma}}(s)$ corresponds to a consistent valuation that makes $\varphi$ true, and since the underlying graph is complete, all sequences of states are possible, and hence $\varphi_{obj}$ is satisfiable iff $\varphi$ is satisfiable.
%\qed
\end{proof}
\begin{acks}                            %% acks environment is optional
                                        %% contents suppressed with 'anonymous'
  %% Commands \grantsponsor{<sponsorID>}{<name>}{<url>} and
  %% \grantnum[<url>]{<sponsorID>}{<number>} should be used to
  %% acknowledge financial support and will be used by metadata
  %% extraction tools.
%   This material is based upon work supported by the
%   \grantsponsor{GS100000001}{National Science
%     Foundation}{http://dx.doi.org/10.13039/100000001} under Grant
%   No.~\grantnum{GS100000001}{nnnnnnn} and Grant
%   No.~\grantnum{GS100000001}{mmmmmmm}.  Any opinions, findings, and
%   conclusions or recommendations expressed in this material are those
%   of the author and do not necessarily reflect the views of the
%   National Science Foundation.
  Work partially supported by the PDR project Subgame perfection in graph games (F.R.S.- FNRS), the ARC project Non-Zero Sum Game Graphs: Applications to Reactive Synthesis and Beyond (Fédération Wallonie-Bruxelles), the EOS project Verifying Learning Artificial Intelligence Systems (F.R.S.-FNRS \& FWO), and the COST Action 16228 GAMENET (European Cooperation in Science and Technology).
\end{acks}

%% Bibliography
%\bibliography{bibfile}

%\newpage
\bibliography{references}
%

%% Appendix
%\section{Appendix}

% 	\input{appendix}

\end{document}